\documentclass[aos]{imsart}

\RequirePackage{amsthm,amsmath,amsfonts,amssymb}
\RequirePackage[authoryear]{natbib}
\RequirePackage[colorlinks,citecolor=blue,urlcolor=blue]{hyperref}
\RequirePackage{graphicx}
\graphicspath{{plots/}}

\usepackage{amsmath}

\DeclareMathOperator*{\argmin}{arg\,min}

\usepackage{float}

\usepackage{booktabs} 
\usepackage{multirow}
\usepackage{tabularx} 
\usepackage{xltabular}

\startlocaldefs
\theoremstyle{plain}

\newtheorem{definition}{Definition}
\newtheorem{theorem}{Theorem}
\newtheorem{proposition}{Proposition}
\newtheorem{corollary}{Corollary}
\newtheorem{lemma}{Lemma}
\newtheorem{remark}{Remark}
\newtheorem{model}{Model}[section]

\newtheorem{uassump}{}

\newtheorem{rassump}{}

\newtheorem{lassump}{}

\endlocaldefs

\begin{document}

\setcounter{rassump}{-1}

\begin{frontmatter}
\title{Uniform Convergence of Generalized Conditional Fr\'{e}chet Means with Applications to Weighted Fr\'{e}chet Aggregation and Exceedance Set Estimation}
\thankstext{t1}{Equal contribution, ordered alphabetically.}
\begin{aug}
\author[A]{\fnms{Houren}~\snm{Hong}\thanksref{t1}\ead[label=e1]{houren.hong@anu.edu.au}}
\author[B]{\fnms{Jiazhen}~\snm{Xu}\thanksref{t1}\ead[label=e2]{jiazhen.xu@mq.edu.au}}
\and
\author[A]{\fnms{Andrew} ~\snm{T.A.}~\snm{Wood}\ead[label=e3]{andrew.wood@anu.edu.au}}
\address[A]{Research School of Finance, Actuarial Studies and Statistics, Australian National University\printead[presep={,\ }]{e1,e3}}

\address[B]{Department of Actuarial Studies and Business Analytics, Macquarie University\printead[presep={,\ }]{e2}}
\end{aug}

\begin{abstract}
The statistical analysis of object oriented data in non-Euclidean spaces heavily relies on generalized conditional Fr\'{e}chet means, notably in the context of Fr\'{e}chet regression. However, establishing the uniform convergence of these estimators presents several theoretical challenges. The difficulties are caused primarily by the absence of linear structures in general metric spaces, rendering standard techniques for verifying the asymptotic uniform equicontinuity of the estimator largely intractable. To overcome this limitation, this paper introduces an alternative theoretical framework for establishing uniform convergence that bypasses the need to verify uniform equicontinuity, under a novel structural condition on the empirical cost function of the generalized conditional Fr\'{e}chet means. We demonstrate that this analytical condition is satisfied by various prominent Fr\'{e}chet regression models across broad classes of metric spaces. Leveraging these foundational uniform convergence guarantees, we subsequently extend two widely used frameworks from Euclidean to  non-Euclidean spaces: (i) a weighted Fr\'{e}chet aggregation framework that facilitates both distributed regression and robust median-of-means regression; and (ii) an exceedance set estimation framework to identify critical covariate regions where the conditional generalized Fr\'{e}chet mean surpasses a prescribed threshold, alongside a metric to quantify the aggregate magnitude of the exceedance. The theoretical properties of these proposed methods are empirically validated through Monte Carlo simulations and an application to dynamic transportation networks in New York City.
\end{abstract}



\end{frontmatter}


\section{Introduction}

The acquisition of object oriented data has become increasingly prevalent across various scientific domains. Beyond traditional Euclidean vectors, such data encompass diverse structures, including directional data, functional curves, geometric shapes, observations on Riemannian manifolds, phylogenetic trees, audio recordings, and images \citep{ bortolussi2006aptreeshape, kolaczyk2014statistical,dryden2016statistical}. Departing from classical statistical methodologies tailored for real-valued observations, the core paradigm of object oriented  data analysis  is to treat these complex data structures as the fundamental atoms of analysis \citep{marron2014overview, MD21}.

Analyzing such data naturally motivates the study of generalized Fr\'{e}chet means for parameters residing in general metric spaces. Prominent examples include Fr\'{e}chet means \citep{frechet1948elements}, Fr\'{e}chet medians \citep{arnaudon2012medians}, and generalized Fr\'{e}chet means \citep{huckemann2011intrinsic,schotz2022strong,lee2026general,park2026generalized}. To describe the conditional relationship between an object-valued response and associated predictors, generalized conditional Fr\'{e}chet means in Fr\'{e}chet regression models \citep{davis2010population, lin2017extrinsic, petersen2019frechet, capitaine2024frechet, qiu2024semi, iao2025deep} have been developed.

Uniform convergence of generalized conditional Fr\'{e}chet means over the space of predictors is a fundamental prerequisite for many downstream theoretical guarantees. For instance, once this property is established, uniform convergence rates can be derived using peeling techniques; see the proof of Theorem 1 in \cite{petersen2019frechet}. Furthermore, key theoretical results in model selection, variable selection, and level set estimation heavily depend on this uniform convergence property; see the theoretical results in \citep{dong2026variable,song2026inference,xu2025quantifying}. Despite its critical role, establishing such uniformity has remained a formidable analytical challenge. This difficulty primarily stems from the challenging task of verifying asymptotic uniform equicontinuity in probability \citep[Chapter 1.5]{vdVW96}. Because general metric spaces lack a linear Euclidean structure, generalized conditional Fr\'{e}chet means typically do not admit closed-form solutions, posing a significant challenge to the verification of asymptotic uniform equicontinuity. Foundational work on Fr\'{e}chet regression \citep{petersen2019frechet, chen2022uniform} has relied on the uniqueness and existence of both population and empirical conditional Fr\'{e}chet means, together with a specific uniformity assumption. This uniformity assumption is designed to guarantee that the asymptotic uniform equicontinuity in probability of the empirical cost function of the empirical conditional Fr\'{e}chet mean transfers directly to the empirical mean itself. However, verifying this condition analytically in a general setting is typically very challenging, and its verification in the current literature has thus far been established using properties not available in general; specifically, linear structures of certain metric spaces equipped with extrinsic distances.

This motivates the development of an alternative theoretical framework to establish the uniform convergence of generalized conditional Fr\'{e}chet means, one that bypasses the uniqueness assumption and  the challenging verification of asymptotic uniform equicontinuity for the estimator itself; see Theorems~\ref{thm::uniform convergence 1} and \ref{thm::uniform convergence 2} below. The core of this approach is Proposition \ref{prop::uniform bound from loss to estimator}, which is built upon an analytical condition \ref{condition R0} on the empirical cost function of the generalized conditional Fr\'{e}chet mean set. As detailed in Proposition \ref{prop::uniform bound from loss to estimator}, this condition allows us to formally link the bound of the empirical cost difference to the distance between the empirical and population estimators. For geodesic metric spaces, this condition naturally reduces to a local geodesic convexity requirement. We demonstrate that this condition is satisfied by various generalized conditional Fr\'{e}chet means across several broad classes of metric spaces in Propositions~\ref{prop::convex Euclidean assump R0}--\ref{prop::Hadamard manifold assump R0}. Building upon this theoretical foundation, we apply these techniques to establish the uniform convergence of multiple prominent Fr\'{e}chet regression models. Crucially, while these uniform convergence results are of substantial independent interest, this framework provides the fundamental prerequisite for deriving theoretical guarantees for two special topics explored later: a weighted Fr\'{e}chet aggregation approach; and a generalized Fr\'{e}chet mean-based framework for exceedance set estimation.

The weighted Fr{\'e}chet aggregation framework is established based on estimating a specific class of generalized conditional Fr\'{e}chet means, Fr{\'e}chet regression estimators, on disjoint data blocks of the whole data before assembling them into a unified global estimator. This divide-and-conquer approach is motivated by two major practical considerations. The first major issue concerns the heavy computational cost of Fr{\'e}chet regression. Fr{\'e}chet regression usually involves intensive optimization over intrinsic distances \citep{bhattacharya2012extrinsic,lin2017extrinsic}, rendering direct model fitting prohibitive for massive datasets. By partitioning the data for parallel processing across local machines, we establish a distributed Fr{\'e}chet regression paradigm. While the statistical properties of distributed computing are heavily studied in Euclidean spaces \citep{rosenblatt2016optimal,hector2020doubly,hector2021distributed,chen2021distributed,li2024selective}, analogous developments for Fr{\'e}chet regression models remain remarkably sparse. The proposed aggregation framework directly addresses this gap by enabling scalable, distributed computing for non-Euclidean data. Specifically, Theorem~\ref{theorem::mom converge} establishes the uniform convergence of the weighted Fr{\'e}chet aggregate.

A second motivation for the weighted Fr{\'e}chet aggregation framework stems from the need for robust estimation of the conditional Fr{\'e}chet mean. Although a robust estimator based on the Huber loss for the Fr{\'e}chet mean on Riemannian manifolds has been developed \citep{lee2026huber}, extending it to Fr{\'e}chet regression models is highly nontrivial. \cite{li2025robust} propose a weight-regularized Fr{\'e}chet regression. However, its performance is highly sensitive to multiple tuning parameters, the selection of which incurs extensive computational costs. To address this, our weighted aggregation framework naturally accommodates a median-of-means (MoM) approach, where local block-wise estimators are aggregated via a conditional Fr{\'e}chet median. The resulting MoM Fr{\'e}chet regression is both computationally efficient and free of tuning parameters. Furthermore, while the statistical properties of MoM estimators are well studied in Euclidean and specific non-Euclidean settings \citep{devroye2016sub,lugosi2019mean,minsker2015geometric,yun2023exponential,lin2024robust, kim2025robust}, their extension to Fr{\'e}chet regression remains underdeveloped. Crucially, existing concentration inequalities for object-valued MoM typically rely on restrictive geometric conditions. For instance, \cite{lin2024robust} requires Lipschitz continuity of the Riemannian logarithm map on a prescribed geodesic ball to establish a concentration inequality for the geometric median of blockwise estimators, while \cite{yun2023exponential} and \cite{kim2025robust} require bounded sectional curvature. By considering a substantially more general metric space, our framework necessitates and provides fundamental new theoretical developments for the MoM Fr{\'e}chet regression estimator, built upon the uniform consistency of the estimator under the uncontaminated setting; see Theorem~\ref{theorem::mom deviation}. Moreover, in Corollary~\ref{corollary::mom}, we provide a finite-sample accuracy guarantee for the MoM Fr\'echet estimator at a prescribed confidence level, equipped with an explicit choice of the number of blocks. To the best of our knowledge, such an explicit block-selection rule has not been provided in the existing object-valued MoM literature.

Beyond point estimation of the conditional relationship between an object-valued response and associated predictors modeled via Fr\'{e}chet regression models, a fundamental objective in the analysis of object oriented data is the identification of critical regions within the covariate space where the underlying data-generating process exhibits extreme or anomalous behavior. Specifically, we consider the problem of exceedance set estimation, which seeks to characterize the subset of the predictor space where a user-defined functional of the Fr\'{e}chet regression estimator surpasses a prescribed threshold. In many scientific domains, however, merely localizing this exceedance set is geometrically insufficient. It is equally vital to quantify the aggregate magnitude of the exceedance, a metric that evaluates the cumulative severity or total volume of the threshold violation across the identified region. Such methodologies are highly sought after in applied fields like climatology and urban informatics. In these domains, practitioners must not only identify specific periods when time-varying compositional pollution metrics or dynamic transportation network activities exceed predefined critical thresholds, but also evaluate the total environmental or systemic impact of the event. While the theoretical foundations of exceedance set estimation are well established for Euclidean and Hilbert-valued responses, see, e.g., \cite{rigollet2009optimal}, \cite{mammen2013confidence}, and \cite{kundu2025exceedance}, there is only limited work extending these concepts to general non-Euclidean metric spaces. The most relevant work is \cite{chen2022uniform} on locating extrema based on the local Fr\'{e}chet  regression model, which is a special case of exceedance set estimation. Ultimately, consistent estimation of both the exceedance set and its corresponding aggregate magnitude relies heavily on the uniform convergence guarantees established herein for various Fr\'{e}chet  regression models, as shown in Propositions~\ref{prop::level set est} and \ref{prop::exceedance}.

In summary, we propose:
\begin{itemize}
    \item A theoretical framework involving an analytical condition establishing uniform convergence of generalized conditional Fr\'{e}chet mean sets over the space of associated predictors. This approach effectively bypasses the challenging verification of asymptotic uniform equicontinuity of generalized conditional Fr\'{e}chet means.
    \item Uniform convergence results for various Fr\'{e}chet-regression-based estimators, alongside the verification that the proposed analytical condition is satisfied across several broad families of metric spaces.
    \item A unified weighted aggregation framework encompassing both distributed and median-of-means Fr\'{e}chet regression. We establish key theoretical guarantees for this framework, including a finite-sample concentration inequality for the median-of-means approach derived from our uniform convergence results.
    \item A framework for exceedance set estimation of the generalized conditional Fr\'{e}chet mean, achieving consistency grounded in the established uniform convergence property.
\end{itemize}


The remainder of this paper is organized as follows. Section~\ref{sect::uniform conv} introduces our general theoretical framework for establishing the uniform convergence of generalized conditional Fr\'{e}chet mean set. In Section~\ref{sect::regression uniform conv}, we focus on an important subclass of general generalized Fr\'{e}chet mean, Fr\'{e}chet-regression-based estimators, and subsequently establishes formal uniform convergence guarantees for various prominent Fr\'{e}chet regression models covering global, local, kernel and $k$-nearest neighbor Fr\'{e}chet regression models.  Section~\ref{sect::MoM} develops the weighted aggregation Fr\'{e}chet framework, while Section~\ref{sect::level set} formulates the methodology for the exceedance set estimation of the generalized conditional Fr\'{e}chet mean. Finally, Sections~\ref{sect::sim data} and \ref{sect::real data} demonstrate the empirical performance of our proposed methods through Monte Carlo simulations and an application to dynamic transportation networks in New York City, respectively. A summary of notation used throughout the paper, all proofs, and data generation setting for simulation studies are deferred to the supplementary material.

\section{Uniform convergence of generalized conditional Fr\'{e}chet means}\label{sect::uniform conv}

Consider a random object $Y$ in a data space $\Theta$ with a well-defined metric, and an associated predictor $X$ in a covariate space $\mathcal{X}$, where $\mathcal{X}$ is a compact metric space with metric $d_{\mathcal{X}}$. The set of population generalized  Fr\'{e}chet means, conditional on $X=x$, is defined as
\begin{align}\label{formual::general population Frechet M est}
    m(x)=\argmin_{\nu\in\Omega} M(\nu,x)
\end{align}
where  $\Omega$, the parameter space, is assumed to be a complete metric space with metric $d$, and $M: \Omega \times \mathcal{X} \to \mathbb{R}$ is a population cost function. Here, $m(x)$ is set-valued and is assumed to be compact and closed. 


Let $(\Xi, \mathcal{F}, \mathbb{P})$ be a probability space large enough for a stochastic process $\{(Y_i, X_i)\}_{i=1}^\infty \subset \Theta \times \mathcal{X}$ to be measureable. All probabilities in this paper are with respect to $(\Xi, \mathcal{F}, \mathbb{P})$. Given a sequence of observations $\{(Y_i, X_i)\}_{i=1}^n \subset \Theta \times \mathcal{X}$, the corresponding empirical generalized Fr\'{e}chet mean for given $x \in \mathcal{X}$ is
\begin{align}\label{formual::general sample Frechet M est}
    \widehat{m}(x)=\argmin_{\nu\in\Omega} \widehat{M}(\nu,x) \subseteq \Omega,
\end{align}
where $\widehat{M}: \Omega \times \mathcal{X} \to \mathbb{R}$ is an empirical cost function constructed from the observations $\{(Y_i, X_i)\}_{i=1}^n$. In this section we do not impose any specific structure on $\{(Y_i, X_i)\}_{i=1}^\infty$, but in subsequent sections of the paper, the $(Y_i, X_i)$ are assumed to be independent and identically distributed with common distribution $P$. Various examples of generalized conditional Fr\'{e}chet means, such as Fr\'{e}chet regression estimators, are discussed in Section \ref{sect::regression uniform conv}. It is worth noting that our setting for generalized conditional Fr\'{e}chet means allows the parameter space $\Omega$ to differ from the data space $\Theta$, which includes the restricted conditional Fr\'{e}chet mean case where  $\Omega\subset\Theta$; see the discussion of the restricted Fr\'{e}chet mean in \cite{evans2024limit}.

Establishing the uniform convergence of the empirical generalized conditional Fr\'{e}chet means $\widehat{m}(x)$ defined in \eqref{formual::general sample Frechet M est} commonly relies on empirical process theory. This standard approach entails demonstrating pointwise convergence for all $x \in \mathcal{X}$ alongside asymptotic uniform equicontinuity in probability (see, e.g., Theorem 18.14 in \cite{vd98asym}). However, verifying this equicontinuity condition is notoriously difficult, due to the fact that many generalized conditional Fr\'{e}chet means lack closed-form expressions; see the discussion in Chapter 1.5, especially pages 38-39, in \cite{vdVW96}. To circumvent this analytical bottleneck, we provide a sufficient condition for the uniform convergence of $\widehat{m}(x)$ that bypasses the need to establish asymptotic uniform equicontinuity in probability. To this end, we introduce the following definition and assumption.

\begin{definition}[Convex metric space]\label{def::convex metric space}
    A metric space $(\mathcal{C}, d)$ is said to be convex if, for all distinct points $x,y\in\mathcal{C}$, there exists $z\in\mathcal{C}\setminus\{x,y\}$ such that $d(x,z)+d(z,y)=d(x,y)$
\end{definition}

\begin{rassump}\label{condition R0}
(i) For each $x\in\mathcal{X}$, $\widehat{m}(x)$ is non-empty almost surely. (ii) The parameter space $\Omega$ is totally bounded; there exists a Polish topological space $\mathcal{M}$ such that $\Omega\subseteq \mathcal{M}$; and, with probability tending to one as $n \rightarrow \infty$, for every $m \in m(x)$, $\mathcal{M}$ has a subset $\mathcal{C}_m \subseteq \mathcal{M}$, which is convex with respect to a metric $d_\mathcal{C}$ on $\mathcal{C}_m$, such that $m\in\mathcal{C}_m$ and $\mathcal{C}_m\cap \widehat{m}(x)\neq \emptyset$.  (iii) For any $x\in\mathcal{X}$ and any $m \in m(x)$, there exists $\widehat{m} \in \widehat{m}(x) \cap \mathcal{C}_m$ and a point $\widetilde{m} \in \mathcal{C}_m \setminus \{m, \widehat{m}\}$ such that $d_\mathcal{C}(m,\widetilde{m})+d_\mathcal{C}(\widetilde{m},\widehat{m})=d_\mathcal{C}(m,\widehat{m})$, and
\begin{align*}
    \widehat{M}(\widetilde{m},x) \leq \delta(x) \widehat{M}(\widehat{m},x) + \{1-\delta(x)\} \widehat{M}(m,x)
\end{align*}
holds almost surely, where $\delta(x)=d_\mathcal{C}(m,\widetilde{m})/d_\mathcal{C}(m,\widehat{m})$.
\end{rassump}

Condition~\ref{condition R0} (i) imposes a standard existence assumption on the empirical generalized conditional Fr\'{e}chet means. Condition~\ref{condition R0} (ii) introduces a geometric structural assumption on a subset of an embedding space of $\Omega$. This assumption, standard in definitions of global geodesic convexity when $\mathcal{C}_m=\Omega$ for every $m\in m(x)$, see, e.g., Chapter 11 of \cite{boumal2023introduction}, facilitates the construction of an intermediate point $\widetilde{m}$ utilized in Condition~\ref{condition R0} (iii). The assumption regarding the existence of both $\mathcal{M}$ and $\mathcal{C}_m$ in Condition~\ref{condition R0} (ii) is essential. To understand why, consider an alternative formulation of Condition~\ref{condition R0} (ii) which requires that a totally bounded, separable, and convex subset $\mathcal{C}_m$ with metric $d_\mathcal{C}$ satisfy $\mathcal{C}_m \subseteq \Omega$ and such that $m\in\mathcal{C}_m$ and $\mathcal{C}_m\cap \widehat{m}(x)\neq \emptyset$ with probability tending to one as $n\to\infty$. Now  consider an extreme case where the restricted domain is strictly discrete, such as $\Omega = \{\nu_1, \nu_2\}$, embedded within a larger convex metric space $\mathcal{M}$ equipped with metric $d_\mathcal{C}$. Because $\Omega$ consists of two isolated points, it lacks the continuous structure required to contain any non-trivial convex subsets. Consequently, this alternative version of the condition fails, but the proposed Condition~\ref{condition R0} (ii) holds.

Condition~\ref{condition R0} (iii) can be interpreted as a slightly weaker form of the local convexity requirement on the empirical cost function $\widehat{M}(\nu,x)$ within each neighborhood of the population generalized Fr\'{e}chet mean $m\in m(x)$. If $\Omega$ is a convex subset of a Euclidean space, setting both $d$ and $d_\mathcal{C}$ to the induced Euclidean distance reduces a sufficient condition of Condition~\ref{condition R0} to the standard local convexity condition found in classical $M$-estimation theory. For a general geodesic space $\Omega$, a sufficient condition of this condition formalizes the local geodesic convexity of the empirical cost function. Thus, Condition~\ref{condition R0} (iii), together with Condition~\ref{condition R0} (ii), serves as a natural extension of local convexity from linear spaces to convex metric spaces \citep{khamsi2011introduction,abdelhakim2016convexity}. It is worth noting that while a local convexity condition requires the inequality in Condition~\ref{condition R0} to hold for all $\widetilde{m} \in \mathcal{C}_m$, Condition~\ref{condition R0} (iii) only requires the existence of a single point $\widetilde{m}$ that satisfies this inequality.

Before presenting our main results, we introduce  notation for distances involving sets.

\begin{definition} [Distances involving sets]\label{Hausdorff}
For a metric space $(\mathcal{S}, d)$, with  $x \in \mathcal{S}$ and $\mathcal{A}, \mathcal{B} \subseteq \mathcal{S}$, we shall write $d(x, \mathcal{A})=\inf_{y \in \mathcal{A}}d(x,y)$ and the one-sided Hausdorff  distance $d (\mathcal{A},\mathcal{B}) =\sup_{x \in \mathcal{A}} \inf_{y \in \mathcal{B}} d(x, y)$ \citep[Definition 1.2]{schotz2022strong}, with an analogous definition for $d(\mathcal{A},x)=d(x,\mathcal{A})$.
\end{definition}

The critical step in establishing the uniform convergence of the empirical generalized conditional Fr\'{e}chet mean set $\widehat{m}(x)$ defined in \eqref{formual::general sample Frechet M est} is formalized in the following proposition. Specifically, Condition~\ref{condition R0} serves to bridge the bounds on the empirical cost difference, $\widehat{M}(\nu, x) - \sup_{m\in m(x)}\widehat{M}(m, x)$ and, the one-sided Hausdorff  distance $d (\widehat{m}(x),m(x))$  between the sets $\widehat{m}(x)$ and $m(x)$ in Definition \ref{Hausdorff}. 

\begin{proposition}\label{prop::uniform bound from loss to estimator}
Under Condition~\ref{condition R0}, for any fixed $x\in\mathcal{X}$ and any positive $\zeta$, if $\widehat{M}(\nu, x) - \sup_{m \in m(x)} \widehat{M}(m, x) > 0$ holds for all $\nu \notin B_\zeta(m(x))$ where $B_\zeta(m(x))=\cup_{m\in m(x)}\{\nu\in\mathcal{M}:d(\nu,m)<\zeta\}$, then $d (\widehat{m}(x),m(x)) < \zeta$.
\end{proposition}

Proposition \ref{prop::uniform bound from loss to estimator} extends Lemma 9.21 of \cite{wainwright2019high} to a broader class of non-Euclidean spaces. Moreover, it dispenses with the global uniqueness requirement for both the sample and population minimizer, i.e. neither $m(x)$ or $\widehat{m}(x)$ is required to be a singleton. The principal advantage of this approach is that it circumvents traditional empirical process techniques, which are difficult to apply to nonlinear statistics which are not available in explicit form. Rather than attempting to verify the challenging asymptotic uniform equicontinuity condition, this proposition allows us to deduce the uniform convergence of the estimator directly from the bounds on the empirical cost function. Proposition \ref{prop::uniform bound from loss to estimator} allows us to establish Theorem \ref{thm::uniform convergence 1} below, which provides a sufficient condition for the uniform convergence of the empirical  generalized conditional Fr\'{e}chet means $\widehat{m}(x)$. 

\begin{theorem}\label{thm::uniform convergence 1}
Under Condition~\ref{condition R0}, if for any $\zeta>0$,
\begin{align}\label{improved sufficient condition}
\mathbb{P}\left [\inf_{x \in \mathcal{X} } \inf_{\nu \notin B_\zeta(m(x))} \left \{ \widehat{M}(\nu, x)  -  \sup_{m \in m(x)} \widehat{M}(m, x) \right \}  >0 \right ]\to 1,
\end{align}
as $n \to \infty$, then $\sup_{x \in \mathcal{X}} d (\widehat{m}(x),m(x)) =o_p(1)$ holds.
\end{theorem}

The key step in the proof of Theorem~\ref{thm::uniform convergence 1} is that, from Proposition~\ref{prop::uniform bound from loss to estimator}, we can obtain
\begin{align*}
\mathbb{P} \left \{\sup_{x \in \mathcal{X}}d( \widehat{m}(x), m(x))>\zeta \right \} \leq \mathbb{P}\left [\inf_{x \in \mathcal{X} } \inf_{\nu \notin B_\zeta(m(x))} \left \{ \widehat{M}(\nu, x)  -  \sup_{m \in m(x)} \widehat{M}(m, x) \right \}  >0 \right ].
\end{align*}
This means once the condition (\ref{improved sufficient condition}) in Theorem \ref{thm::uniform convergence 1} is satisfied, one can get the uniform convergence of the generalized Fr\'{e}chet mean under Condition~\ref{condition R0}.

Under the following additional identifiability assumption, Condition \eqref{improved sufficient condition} of Theorem \ref{thm::uniform convergence 1} can be verified directly by establishing the uniform convergence of the empirical cost function $\widehat{M}(\nu,x)$, as detailed in Theorem \ref{thm::uniform convergence 2}.

\begin{rassump}\label{condition R1}
For any $\zeta>0$,
  \begin{align*}
    &\inf_{x\in\mathcal{X}}\inf_{\nu\in\mathcal{M}:d(\nu,m(x))>\zeta} \left\{ M(\nu,x) - \sup_{m\in m(x)}M(m,x) \right\}>0.
  \end{align*}
\end{rassump}

Condition \ref{condition R1} ensures identifiability for the population generalized conditional Fr\'{e}chet means. When $m(x)$ is a singleton, this condition represents a standard prerequisite in classical $M$-estimation theory, see, e.g., Theorem 5.7 in \cite{vd98asym}. Analogous conditions are routinely employed within the Fr\'{e}chet regression literature, where they have been verified across a variety of metric spaces (\citealt{petersen2019frechet}).

\begin{theorem}\label{thm::uniform convergence 2}
Under Conditions~\ref{condition R0} and \ref{condition R1}, if
\begin{align}\label{improved sufficient condition 2}
\sup_{x\in\mathcal{X}}\sup_{\nu\in\mathcal{M}} \left| \widehat{M}(\nu, x) -M(\nu,x) \right| =o_p(1) ,
\end{align}
then $\sup_{x \in \mathcal{X}} d( \widehat{m}(x) , m(x)) =o_p(1)$ holds.
\end{theorem}

Theorem \ref{thm::uniform convergence 2} demonstrates that under Conditions~\ref{condition R0} and \ref{condition R1}, establishing the uniform convergence of the empirical generalized conditional Fr\'{e}chet means $\widehat{m}(x)$ reduces to verifying the condition \eqref{improved sufficient condition 2}. This condition essentially serves as a uniform law of large numbers. A standard approach to verifying this is to first establish the weak convergence $\widehat{M} \rightsquigarrow M$ in $l^{\infty}(\Omega\times \mathcal{X})$ which is the space of uniformly bounded functions on the product space $\Omega\times \mathcal{X}$, and subsequently apply the continuous mapping theorem, see, e.g., Theorem 1.3.6 in \cite{vdVW96}. By Theorem 1.5.4 of \cite{vdVW96}, this weak convergence holds if two requirements are met: (i) pointwise convergence such as $\widehat{M}(\nu, x) - M(\nu, x) = o_p(1)$ for every $(\nu, x) \in \Omega \times \mathcal{X}$, and (ii) the asymptotic uniform equicontinuity of $\widehat{M}$ in probability. Specifically, the latter requires that for any $\iota > 0$,
\[
\limsup_{n \rightarrow \infty} \mathbb{P}\left( \sup_{d(\nu_1,\nu_2)<\zeta_1} \sup_{d_\mathcal{X}(x_1,x_2)<\zeta_2} \left| \widehat{M}(\nu_1,x_1)-  \widehat{M}(\nu_2,x_2) \right| >\iota \right)\to0,
\]
as $\zeta_1,\zeta_2\to0$.

\subsection{Comparison with related work on the unique generalized conditional Fr\'{e}chet mean}\label{subsect::mean set uniform conv}

To facilitate comparison with existing works that assume a unique population conditional Fr\'{e}chet  mean, we focus on the case where the population generalized conditional Fréchet mean set $m(x)$ in (\ref{formual::general population Frechet M est}) is a singleton. Then, the condition (\ref{improved sufficient condition}) in Theorem~\ref{thm::uniform convergence 1} reduces to the following that for any $\zeta>0$,
\[
\mathbb{P}\left [\inf_{x \in \mathcal{X} } \inf_{\nu \notin B_\zeta(m(x))} \left \{ \widehat{M}(\nu, x)  -   \widehat{M}(m(x), x) \right \}  >0 \right ]\to 1.
\]
This condition is different from a condition widely utilized in existing literature on Fr\'{e}chet regression models; see e.g., \cite{petersen2019frechet}, \cite{chen2022uniform} and \cite{zhou2022network}. Specifically, these works commonly assume that for any $\zeta > 0$, there exists a $\tau = \tau(\zeta) > 0$ such that
\begin{align}\label{sample identificaion condition}
    \lim_{n\to\infty} \mathbb{P}\left [ \inf_{x\in \mathcal{X}} \inf_{\nu\in\mathcal{M}:d(\nu,\widehat{m}(x))>\zeta}  \left\{ \widehat{M}(\nu,x) - \widehat{M}(\widehat{m}(x),x)  \right\}>\tau \right] = 1.
\end{align}
This condition guarantees that the asymptotic uniform equicontinuity of the empirical cost function $\widehat{M}(\cdot,x)$ for generalized Fr\'{e}chet mean $\widehat{m}(x)$ implies the asymptotic uniform equicontinuity of $\widehat{m}(x)$.  However, \eqref{sample identificaion condition} can be very difficult to verify when it involves implicitly defined sample quantities. Establishing probability bounds for an infimum taken over a region that explicitly depends on the sample quantity $\widehat{m}(x)$ is technically demanding. Standard approaches to verifying \eqref{sample identificaion condition} typically rely on establishing a stricter  equation of the form $\widehat{M}(\nu,x) - \widehat{M}(\widehat{m}(x),x) = C d^2(\nu,\widehat{m}(x))$ for some fixed constant $C>0$. This equation has only been verified for a restricted class of metric spaces equipped with extrinsic Euclidean distances; see, e.g., Propositions 1 and 2 in \cite{petersen2019frechet}. 

If one further assumes the uniqueness of the empirical generalized conditional Fr\'{e}chet mean $\widehat{m}(x)$, and if one can show the continuity of the function $d(\widehat{m}(x),m(x))$, that is, for any sequence $\{x_\varkappa\}\in\mathcal{X}$ such that $x_\varkappa\to x$, one has
\begin{align}\label{formula::continuity of estimator}
    \left|d(\widehat{m}(x_\varkappa),m(x_\varkappa))-d(\widehat{m}(x),m(x))\right|\to 0,    
\end{align}
almost surely, then following the techniques from \cite{ziezold1977expected} and \cite{huckemann2011intrinsic}, we can extend the pointwise strong consistency of the generalized conditional Fr\'{e}chet mean to uniform strong consistency. The verification of (\ref{formula::continuity of estimator}) requires  the almost surely continuity of $\widehat{m}(\cdot)$; that is $d(\widehat{m}(x_\varkappa),\widehat{m}(x))\to 0$ holds almost surely for any sequence $\{x_\varkappa\}\in\mathcal{X}$ such that $x_\varkappa\to x$. However, for a generalized Fr\'{e}chet conditional mean which does not admit a closed form solution, it does not appear to be feasible to verify this continuity property without further assumptions. Indeed, if one could verify the almost surely continuity of $\widehat{m}(\cdot)$, the asymptotic equicontinuity of the generalized Fr\'{e}chet conditional mean would then be established, leading to the uniform convergence of the generalized Fr\'{e}chet conditional mean.

Finally, when focusing on time-varying Fr\'{e}chet medians under the uniqueness assumption, \cite{xu2025robust} develops several techniques to establish the uniform convergence based on Assumption (A3). Importantly, Condition~\ref{condition R0}, under the additional uniqueness assumption, is still weaker than Assumption (A3) in \cite{xu2025robust}. We do not require the entire space $\Omega$ to be a convex metric space, nor do we restrict $d$ and $d_\mathcal{C}$ to be identical metrics. For instance, consider the case where the parameter space $\Omega=\mathcal{M}=\mathcal{C}_m$ is the unit sphere $\mathcal{S}^2 \subset \mathbb{R}^3$ equipped with the extrinsic distance $d_{\rm IE}(\nu_1, \nu_2) = \|\nu_1 - \nu_2\|_{\rm E}$, where $\|\cdot\|_{\rm E}$ is the Euclidean norm. Because $(\mathcal{S}^2, d_{\rm IE})$ is not a convex metric space, Assumption (A3) of \cite{xu2025robust} fails. Under Condition~\ref{condition R0}, however, we can select $d_\mathcal{C}$ to be the great-circle distance. This choice ensures that $(\mathcal{S}^2, d_\mathcal{C})$ operates as a convex metric space (\citealt{khamsi2011introduction}), while simultaneously permitting the use of the extrinsic distance $d_{\rm IE}$ for extrinsic model fitting \citep{lin2017extrinsic}. This will be discussed further in Proposition~\ref{prop::Riemannian space extrinsic assump R0} for the extrinsic Fr\'{e}chet regression framework.

\section{Uniform convergence for Fr\'{e}chet regression}\label{sect::regression uniform conv}

In this section, we investigate the uniform convergence of Fr\'{e}chet-regression-based estimators. To characterize the conditional relationship between $Y$ and $X$, a natural choice for the population cost function is the conditional squared distance cost, $R(\nu,x) = \mathbb{E}[d^2(\nu,Y) \mid X=x]$. By setting $M(\nu,x) = R(\nu,x)$ in \eqref{formual::general population Frechet M est}, the general parameter $m(x)$ reduces to the conditional Fr\'{e}chet mean $r(x) = \argmin_{\nu\in\Omega} R(\nu,x)$. Various regression frameworks model this conditional Fr\'{e}chet mean by incorporating a weighting function into the cost, leading to the weighted population target
\begin{align}\label{formula::conditional Fr\'{e}chet mean}
    \widetilde{r}(x)=\argmin_{\nu\in\Omega} \widetilde{R}(\nu,x),~~\widetilde{R}(\nu,x)=\mathbb{E}\left[ w(X,x) d^2(\nu,Y) \right].
\end{align}
The existing literature on Fr\'{e}chet regresion commonly assumes the conditional Fr\'{e}chet mean, either $r(x)$ or $\widetilde{r}(x)$, is unique. In the following, we adapt this uniqueness assumption. The corresponding empirical estimator for the conditional Fr\'{e}chet mean $r(x)$ is then defined as
\begin{align*}
    \widehat{r}(x)=\argmin_{\nu\in\Omega} \widehat{R}(\nu,x),~~\widehat{R}(\nu,x)= \sum_{i=1}^n \widehat{w}(X_i,x) d^2(\nu,Y_i)
\end{align*}
based on a set of data-dependent weights $\{\widehat{w}(X_i,x)\}_{i=1}^n$. In the following, we present five examples of such Fr\'{e}chet-regression-based estimators.

\begin{model}[Global Fr\'{e}chet regression (\citealt{petersen2019frechet})]\label{exmp::global F regression}
 When $\mathcal{X}\subset \mathbb{R}^p$ for some positive integer $p$, the global Fr\'{e}chet regression given by considering the empirical weights $\widehat{w}(X_i,x)=n^{-1}+n^{-1}(X_i-\widehat{\mu})^\top \widehat{\Sigma}^{-1} (x-\widehat{\mu})$ with $\widehat{\mu}=n^{-1}\sum_{i=1}^nX_i$ and $\widehat{\Sigma}=n^{-1}\sum_{i=1}^n (X_i-\widehat{\mu})(X_i-\widehat{\mu})^\top$.
\end{model}

\begin{model}[Local Fr\'{e}chet regression (\citealt{petersen2019frechet})]\label{exmp::local F regression}
When $\mathcal{X}\subset \mathbb{R}$, the local Fr\'{e}chet regression is constructed via considering the empirical weights $\widehat{w}(X_i,x)=n^{-1}\widehat{\sigma}^{-2}K_h(X_i-x)[\widehat{\tau}_2(x)-\widehat{\tau}_1(x)]$ where $\widehat{\tau}_l(x)=n^{-1}\sum_{i=1}^n[K_h(X_i-x)(X_i-x)^l]$ for $l\in\{0,1,2\}$, $\widehat{\sigma}^2=\widehat{\tau}_0(x)\widehat{\tau}_2(x)-\widehat{\tau}_1^2(x)$, and $K_h(\cdot)=K(\cdot/h)/h$ for a smoothing kernel $K(\cdot)$ and a bandwidth sequence $h>0$.
\end{model}


\begin{model}[Kernel Fr\'{e}chet regression (\citealt{davis2010population})]\label{exmp::kernel F regression}
When $\mathcal{X}\subset \mathbb{R}$, a kernel-based Fr\'{e}chet mean is constructed based on the empirical weights are $\widehat{w}(X_i,x)=K_h(X_i-x)/\sum_{j=1}^nK_h(X_j-x)$.
\end{model}

\begin{model}[$k$-nearest neighbor Fr\'{e}chet regression]\label{exmp::kNN F regression}
The $k$-nearest neighbor (kNN) Fr\'{e}chet regression is developed considering the empirical weights $\widehat{w}(X_i,x)=1/k$ if $X_i\in {\rm kNN}(x,d_\mathcal{X})$ and 0 otherwise, where ${\rm kNN}(x,d_\mathcal{X})$ represents the set of $k$ nearest labeled neighbors to $x$ based on a suitable
distance function $d_\mathcal{X}$.
\end{model}

\begin{model}[Deep Fr\'{e}chet regression (\citealt{iao2025deep,zhou2025end})]\label{exmp::deep F regression}
When $\mathcal{X}\subset \mathbb{R}^p$, the deep Fr\'{e}chet regression is developed considering the empirical weights $\{ \widehat{w}(X_i,x): \sum_{i=1}^n \widehat{w}(X_i,x)=1,\widehat{w}(X_i,x)\geq 0~{\rm for}~i=1,2,\ldots,n\}$ learning from  a fully connected neural network with multiple hidden layers using rectified linear unit activations, followed by a softmax output layer. 
\end{model}

Model~\ref{exmp::global F regression} generalizes classical Euclidean linear regression to metric spaces. Specifically, when $\Omega = \mathbb{R}$ is equipped with the Euclidean distance $d_{\rm E}$, the estimator $\widehat{r}(x)$ in Model~\ref{exmp::global F regression} coincides with the fitted values obtained via the Euclidean linear regression using ordinary least squares. Similarly, Model~\ref{exmp::local F regression} extends Euclidean local linear regression to the Fr\'{e}chet setting. Beyond this local linear approach, alternative localized estimators include the kernel and $k$-NN Fr\'{e}chet regressions detailed in Models~\ref{exmp::kernel F regression} and \ref{exmp::kNN F regression}, respectively. For the specific case where the predictor space $\mathcal{X}\subset \mathbb{R}^p$, $k$-NN Fr\'{e}chet regression was analyzed by \cite{qiu2024semi}. Model~\ref{exmp::deep F regression} integrates neural network architectures into the Fr\'{e}chet regression framework. We remark that the literature encompasses a wide array of additional Fr\'{e}chet regression methodologies  (\citealt{capitaine2024frechet,qiu2024random, gyorfi2026metric}), including random-forest-based estimators. Notably, while the empirical weights $\widehat{w}$ may take negative values in the global and local linear frameworks of Model~\ref{exmp::global F regression} and \ref{exmp::local F regression}, the weights employed in the vast majority of other Fr\'{e}chet regression models are strictly nonnegative.

In the following, we develop several propositions to verify that Condition~\ref{condition R0} holds for some specific metric spaces considering Models~\ref{exmp::global F regression}--\ref{exmp::deep F regression}.

\begin{proposition}\label{prop::convex Euclidean assump R0}
Let $\Omega$ be a convex and totally bounded subset of a Euclidean space, equipped with the induced Euclidean metric $d=d_{\rm IE}$. Then, Condition~\ref{condition R0} is satisfied with $d_\mathcal{C}=d_{\rm IE}$, $m(x)=r(x)$, $\widehat{M}(\nu,x)=\widehat{R}(\nu,x)$ described in Models~\ref{exmp::global F regression}--\ref{exmp::deep F regression}.
\end{proposition}

A specific metric space covered in Proposition~\ref{prop::convex Euclidean assump R0} is the space of graph Laplacians of undirected networks with the Frobenius distance; this is  considered later, in Sections~\ref{sect::sim data} and \ref{sect::real data}. To see this, recall that a graph Laplacian is a symmetric matrix, and thus, the Frobenius distance between any two graph Laplacians is equivalent to the standard Euclidean distance between their half-vectorizations.

\begin{proposition}\label{prop::Riemannian space assump R0}
Let $(\Omega,d)$ be a compact Riemannian manifold,  and ${\rm inj}(\Omega)$ be the injectivity radius of the manifold, defined in the proof. If the support of $Y_1,Y_2,\ldots,Y_n$ is contained within a ball $B(p, r) \subset \Omega$ with $r < 2^{-1}{\rm inj}(\Omega)$, then Condition~\ref{condition R0} holds with $d_\mathcal{C}=d$, $m(x)=r(x)$, $\widehat{M}(\nu,x)=\widehat{R}(\nu,x)$ for Models~\ref{exmp::kernel F regression}--\ref{exmp::deep F regression}. If we further assume $W^-<[\{ {\rm inj}(\Omega)\}^2-2r{\rm inj}(\Omega)]/(4r{\rm inj}(\Omega)+r^2)$ where $W^-=\sum_{i\notin\mathcal{I}^+}|\widehat{w}(X_i,x)|$ with $\mathcal{I}^+=\{i\in\{1,2,\ldots,n\}:\widehat{w}(X_i,x)\geq 0\}$, then Condition~\ref{condition R0} holds with $d_\mathcal{C}=d$, $m(x)=r(x)$, $\widehat{M}(\nu,x)=\widehat{R}(\nu,x)$ for Models~\ref{exmp::global F regression}--\ref{exmp::local F regression}.
\end{proposition}

Proposition~\ref{prop::Riemannian space assump R0} examines the case where $\Omega$ is a compact Riemannian manifold. In this setting, Condition~\ref{condition R0} naturally reduces to a weaker requirement for the local geodesic convexity, as discussed in Section~\ref{sect::uniform conv}. Crucially, this condition can be formally verified by evaluating the differentiability of the empirical cost function $\widehat{R}(\nu,x)$ around each minimizer $\widehat{r}\in \widehat{r}(x)$.

As demonstrated by Proposition~\ref{prop::Riemannian space assump R0}, the possibly negative empirical weights $\widehat{w}$, such as those encountered in Models \ref{exmp::global F regression} and \ref{exmp::local F regression}, necessitate additional conditions to guarantee the differentiability of the empirical cost function $\widehat{R}(\cdot,x)$. However, this technical complication is circumvented under the extrinsic Fr\'{e}chet regression framework (\citealt{lin2017extrinsic}), as formalized in Proposition~\ref{prop::Riemannian space extrinsic assump R0}.

\begin{proposition}\label{prop::Riemannian space extrinsic assump R0}
Let $\Omega$ be a compact Riemannian manifold with the induced Euclidean distance $d=d_{\rm IE}$.  Then Condition~\ref{condition R0} holds with $d_\mathcal{C}=d_\Omega$, where $d_\Omega$ is the corresponding intrinsic metric, $m(x)=r(x)$, $\widehat{M}(\nu,x)=\widehat{R}(\nu,x)$ for Models~\ref{exmp::global F regression}--\ref{exmp::deep F regression}.
\end{proposition}

Notably, when the compact Riemannian manifold $\Omega$ is equipped with the induced Euclidean distance $d_{\rm IE}$, the space $(\Omega,d_{\rm IE})$ is unlikely a convex metric space, and thus the requirement of the convex metric space in Assumption (A3) in \cite{xu2025robust} will typically not be satisfied in this situation.

The following proposition studies Hadamard spaces; these spaces are discussed in e.g., \cite{BH11}. The global non-positive curvature characteristic of Hadamard spaces guarantee that the squared distance is geodesically convex. Consequently, the empirical cost function $\widehat{R}(\nu, x)$ inherits this convexity when the empirical weights are nonnegative.

\begin{proposition}\label{prop::Hadamard space assump R0}
Let $(\Omega,d)$ be a Hadamard space.  Then  Condition~\ref{condition R0} holds with $d_\mathcal{C}=d$, $m(x)=r(x)$, $\widehat{M}(\nu,x)=\widehat{R}(\nu,x)$ for Models~\ref{exmp::kernel F regression}--\ref{exmp::deep F regression}.
\end{proposition}

If we focus on a  subclass of Hadamard spaces possessing smooth structure, Condition~\ref{condition R0} may be verified as in the proposition below using the differentiability of $\widehat{R}(\nu,x)$ at $\widehat{r}(x)$.

\begin{proposition}\label{prop::Hadamard manifold assump R0}
Let $(\Omega,d)$ be a Hadamard manifold, Condition~\ref{condition R0} holds with $d_\mathcal{C}=d$, $m(x)=r(x)$, $\widehat{M}(\nu,x)=\widehat{R}(\nu,x)$ for Models~\ref{exmp::global F regression}--\ref{exmp::deep F regression}.
\end{proposition}

\subsection{Global and local Fr\'{e}chet regression}\label{subsect::regression 1}

We now use Theorem~\ref{thm::uniform convergence 2} to establish uniform convergence for both global and local Fr\'{e}chet regression models in Models~\ref{exmp::global F regression} and \ref{exmp::local F regression}, respectively. Different from \cite{petersen2019frechet} and \cite{chen2022uniform} where the uniform convergence results are established  via assuming the existence and uniqueness of either population or empirical Fr\'{e}chet regression estimator,  together with the condition (\ref{sample identificaion condition}), we aim to use Condition~\ref{condition R0} to derive the desired asymptotic properties.

\begin{theorem}\label{prop::global uniform conv}
    Consider the global Fr\'{e}chet regression in Model~\ref{exmp::global F regression}, then $\widetilde{r}(x)$ in (\ref{formula::conditional Fr\'{e}chet mean}) is the argument of the minimum of the population cost function $\widetilde{R}(\nu,x)=\mathbb{E}\left[ w(X,x) d^2(\nu,Y) \right]$ where $w(X,x)=1+(X-\mu)^\top \Sigma^{-1}(x-\mu)$, $\mu=\mathbb{E}(X)$ and $\Sigma=\mathbb{E}[(X-\mu)(X-\mu)^\top]$. Under Conditions~\ref{condition R0} and \ref{condition R1}, and further assume that the observations $\{(X_i,Y_i)\}_{i=1}^n$ are independent and identically distributed, then
    \[
    \sup_{x\in\mathcal{X}}d(\widehat{r}(x),r(x)) \leq \sup_{x\in\mathcal{X}}d(\widetilde{r}(x),r(x)) + o_p(1).
    \] 
\end{theorem}

\begin{remark}
     For the global Fr\'{e}chet regression in Model~\ref{exmp::global F regression}, although Theorem~\ref{prop::global uniform conv} establishes that $\sup_{x\in\mathcal{X}}d(\widehat{r}(x),\widetilde{r}(x))=o_p(1)$, a bias term $\sup_{x\in\mathcal{X}}d(\widetilde{r}(x),r(x))$ may exist if the weights in $\widetilde{R}(\nu,x)$ in (\ref{formula::conditional Fr\'{e}chet mean}) fail to reflect the true but unknown structure in $R(\nu,x)$, that is, if $\widetilde{R}(\nu,x)\neq R(\nu,x)$.
\end{remark}

As discussed in Section~\ref{sect::uniform conv}, the uniform convergence of  the generalized conditional Fr\'{e}chet means in Theorem~\ref{prop::global uniform conv} can be achieved by verifying the condition \eqref{improved sufficient condition 2} in Theorem~\ref{thm::uniform convergence 2}.
For the sample generalized Fr\'{e}chet mean under the local Fr\'{e}chet regression setting in Model~\ref{exmp::local F regression}, we require the following additional assumptions.

\begin{uassump}\label{condition U1}
Let the kernel $K$ be a symmetric probability density function on $\mathbb{R}$ that is uniformly continuous. For for $l_1,l_2\in\{1,2,4,6\}$, define $H_{l_1l_2}=\int_\mathbb{R} K^{l_1}(x)x^{l_2}dx<\infty$, and $H_{14}$ and $H_{26}$ are bounded. Assume that the derivative $K^\prime$ exists and is bounded over the support of $K$, i.e., $\sup_{K(x)>0}|K^\prime(x)|<\infty$. Moreover, require that $\int_\mathbb{R}t^2|K^\prime(x)|\sqrt{|x\log |x||}dx<\infty$.
\end{uassump}

\begin{uassump}\label{condition U2}
The marginal density $g_X$ of $X$ and the conditional densities $g_{X|Y}(\cdot,y)$ of $X$ given $Y=y$ exist and are continuous on $\mathcal{X}$ and twice continuously differentiable on the interior of $\mathcal{X}$, denoted as $\mathcal{X}^\circ$, the latter for all $y\in\Omega$. The marginal density $g_X$ is bounded away from zero on $\mathcal{X}$, i.e., $\inf_{x\in\mathcal{X}}g_X(x)>0$. The uniform boundedness of the  second-order derivative $g_X^{\prime\prime}$ and the second-order partial derivatives $\frac{\partial^2 g_{X|Y}(x,y)}{\partial x^2}$ is required, i.e., $\sup_{x\in\mathcal{X}^\circ}|g_X^{\prime\prime}(x)|<\infty$, and $\sup_{x\in\mathcal{X}^\circ,y\in\Omega}\left| \frac{\partial^2 g_{X|Y}(x,y)}{\partial x^2} \right|<\infty$. Furthermore, for any open set $\Omega_{\rm sub}\subset \Omega$, $\mathbb{P}[Y\in \Omega_{\rm sub}|X=x]$ is continuous as a function of $x$. For any $x\in\mathcal{X}$, $R(\nu,x)$ is equicontinuous, i.e., $\limsup_{z\to x}\sup_{\nu\in\Omega}|R(\nu,z)-R(\nu,x)|=0$.
\end{uassump}

\ref{condition U1} is used to apply results of \cite{silverman1978weak} and \cite{mack1982weak}. \ref{condition U2} is a standard distributional assumption for local nonparametric regression. \ref{condition U1} and \ref{condition U2} guarantee the asymptotic uniform equicontinuity of the population cost function for the local Fr\'{e}chet regression $\widetilde{R}_L$. \ref{condition U1} and \ref{condition U2} are consistent with those used in \cite{chen2022uniform}. 

The uniform convergence result of the local Fr\'{e}chet regression estimator is shown in the theorem below.

\begin{theorem}\label{prop::local uniform conv}
    Consider the local Fr\'{e}chet regression in Model~\ref{exmp::local F regression}, under  Conditions~\ref{condition R0}, \ref{condition R1}, \ref{condition U1}, and \ref{condition U2}, and further assume that the observations $\{(X_i,Y_i)\}_{i=1}^n$ are independent and identically distributed. If $h \to 0$ and $nh / \log n \to \infty$ as $n \to \infty$,
    \[
    \sup_{x\in\mathcal{X}}d(\widehat{r}(x),r(x))=o_p(1).
    \]
\end{theorem}

Unlike Theorem~\ref{prop::global uniform conv}, the construction of the local Fr\'{e}chet regression in Model~\ref{exmp::local F regression} leads to the bias $\sup_{x\in\mathcal{X}}d(\widetilde{r}(x),r(x))=o(1)$ under Conditions~\ref{condition U1} and \ref{condition U2}.

\subsection{Kernel and $k$-NN Fr\'{e}chet regression}\label{subsect::regression 2}

We now study the the uniform convergence of the generalized Fr\'{e}chet mean in Models~\ref{exmp::kernel F regression} and \ref{exmp::kNN F regression}, respectively. The uniform convergence result of the kernel Fr\'{e}chet regression in Model~\ref{exmp::kernel F regression} is provided in the following theorem.

\begin{theorem}\label{thm::kernel uniform conv}
    Consider the kernel Fr\'{e}chet regression in Model~\ref{exmp::kernel F regression}, under  Conditions~\ref{condition R0}, \ref{condition R1}, \ref{condition U1} and \ref{condition U2}, and further assume that the observations $\{(X_i,Y_i)\}_{i=1}^n$ are independent and identically distributed. If $h \to 0$ and $nh / \log n \to \infty$ as $n \to \infty$,
    \[
    \sup_{x\in\mathcal{X}}d(\widehat{r}(x),r(x))=o_p(1).
    \]
\end{theorem}

For the sample generalized Fr\'{e}chet mean under the $k$-NN Fr\'{e}chet regression setting in Model~\ref{exmp::local F regression}, we require additional assumptions as follows.

\begin{uassump}\label{condition U3}
Let $g_X$ denote the marginal probability measure of $X$, and $B_\mathcal{X}(x, \zeta) = \{u \in \mathcal{X} : d_\mathcal{X}(u, x) \leq \zeta\}$ denote a closed ball in $\mathcal{X}$. There exists a strictly increasing, continuous function $\phi: [0, \infty) \to [0, \infty)$ with $\phi(0)=0$, such that for all sufficiently small $\zeta > 0$, $\inf_{x \in \mathcal{X}} g_X(B_\mathcal{X}(x, \zeta)) \geq \phi(\zeta) > 0$.
\end{uassump}

\begin{uassump}\label{condition U4}
The collection of all closed balls $\mathcal{C}_\mathcal{X} = \{B_\mathcal{X}(x, \zeta) : x \in \mathcal{X}, \zeta > 0\}$ forms a Vapnik-Chervonenkis class of sets.
\end{uassump}

\begin{uassump}\label{condition U5}
For any $x\in\mathcal{X}$,  $R(\nu,x) = \mathbb{E}[d^2(\nu,Y)|X=x]$ is equicontinuous with respect to the predictor space $\limsup_{z\to x}\sup_{\nu\in\Omega}|R(\nu,z)-R(\nu,x)|=0$.
\end{uassump}

Conditions~\ref{condition U3}--\ref{condition U5} control both the geometric complexity of the predictor space and the smoothness of the underlying conditional distribution. Condition \ref{condition U3} is a small ball probability assumption governing the marginal distribution of the predictor $X$. By requiring that the probability mass of any closed ball is uniformly bounded away from zero by a monotonically increasing function $\phi(\zeta)$, this condition ensures that there are no regions in $\mathcal{X}$ with vanishingly small densities. This guarantees that for any query point $x \in \mathcal{X}$, the $k$-nearest neighbors will be geometrically localized around $x$ as the sample size $n \to \infty$, which is essential for controlling the stochastic variance of the local estimator. For instance, if $\mathcal{X}\subset \mathbb{R}^p$, this condition is satisfied with $\phi(\zeta) \propto \zeta^p$. Condition~\ref{condition U4} imposes a constraint on the topological complexity of the predictor space $\mathcal{X}$. Requiring the collection of all closed metric balls, $\mathcal{C}_\mathcal{X}$, to form a Vapnik-Chervonenkis (VC) class allows us to utilize uniform laws of large numbers over these local neighborhoods. It is worth noting that, the collection of closed balls in any finite-dimensional Euclidean space, as well as in many compact Riemannian manifolds, inherently forms a VC class. Condition~\ref{condition U5} serves as the nonparametric smoothness assumption required to control the approximation error of the $k$-NN estimator. Condition~\ref{condition U5} is equivalent to the last part of Condition~\ref{condition U2} for local Fr\'{e}chet regression models.

The uniform convergence result of the $k$-NN Fr\'{e}chet regression in Model~\ref{exmp::kNN F regression} is shown in the following theorem.

\begin{theorem}\label{thm::kNN uniform conv}
    Consider the $k$-NN Fr\'{e}chet regression in Model~\ref{exmp::kNN F regression}, under  Conditions~\ref{condition R0}, \ref{condition R1}, \ref{condition U3}--\ref{condition U5}, and further assume that the observations $\{(X_i,Y_i)\}_{i=1}^n$ are independent and identically distributed. If the sequence of neighbors $k = k_n$ satisfies $k \to \infty$, $k/n \to 0$, and $k/\log n \to \infty$ as $n \to \infty$,
    \[
    \sup_{x\in\mathcal{X}}d(\widehat{r}(x),r(x))=o_p(1).
    \]
\end{theorem}

\section{Weighted Fr\'{e}chet aggregation for Fr\'{e}chet regressions}\label{sect::MoM}
This section studies the weighted Fr\'{e}chet aggregation framework where a convex aggregation of block-wise Fr{\'e}chet regression estimators is developed. Specifically, let $B \geq 1$ be the number of disjoint blocks, for each block $b=1,\ldots, B$, let $\mathcal{P}_b$ be the joint probability distribution of $(Y^{(b)}, X^{(b)})\in \Omega \times \mathcal{X}$ and $\{(Y_i^{(b)}, X_i^{(b)})\}_{i=1}^{n_b}$ be an independent and identically distributed sample from $\mathcal{P}_b$. We apply the same regression method to each block $b$ to obtain a unique blockwise estimator of the conditional Fr\'echet mean, e.g., for any $x\in\mathcal{X}$, 
\begin{equation}\label{e:blockwise}
\widehat{r}_{b}(x) = \argmin_{\nu\in \Omega} \; n_b^{-1} \sum_{i=1}^{n_b} \widehat{\omega}(X_i^{(b)}, x) d^2 (\nu, Y_i^{(b)}),    
\end{equation}
where the $\{\widehat{\omega}(X_i^{(b)},x)\}_{i=1}^{n_b}$ are weights in Fr{\'e}chet regression models; see Models~\ref{exmp::global F regression}-\ref{exmp::deep F regression} in Section~\ref{sect::regression uniform conv}. 
We then aggregate the collection of block-wise estimators $\{\widehat{r}_{b}(x)\}_{b=1}^B$ and define the weighted Fr{\'e}chet estimator of $\ell$-th order at $x\in \mathcal{X}$ as
\begin{equation}\label{e:mom}
\widehat{r}(x,\pmb{\kappa};\ell) = \argmin_{\nu \in \Omega} ~  \sum_{b=1}^{B} \mathbf{\kappa}_{b} d^{\ell}(\nu, \widehat{r}_{b}(x)),
\end{equation}
where $\ell\in\{1,2\}$ and $\pmb{\kappa}=(\kappa_1,\ldots,\kappa_B)^\top$ is a weight vector in the simplex
\begin{equation}\label{e:Kappa}
\mathcal{K} = \left\{\pmb{\kappa}\in \mathbb{R}^B: \sum_{b=1}^B \kappa_b =1 ~\text{and}~ \kappa_b \ge 0 ~ \text{for} ~ b=1,\ldots,B\right\}.    
\end{equation}
Furthermore, for each block, the conditional Fr{\'e}chet mean is assumed to be globally unique and is defined as
\[r_b(x) = \argmin_{\nu\in \Omega} \; \mathbb{E}_{\mathcal{P}_b}\left[d^2(\nu, Y^{(b)})\big\vert X^{(b)}=x\right],\]
and the population version of $\widehat{r}(x,\pmb{\kappa};\ell)$ is then given by 
\begin{equation*}
r(x,\pmb{\kappa};\ell) = \argmin_{\nu \in \Omega} ~  \sum_{b=1}^{B} \mathbf{\kappa}_{b} d^{\ell}(\nu, r_b(x)), \quad \text{where } \ell\in\{1,2\}.
\end{equation*}   
Notably, the proposed weighted Fr\'{e}chet aggregation framework in (\ref{e:mom}) covers two special models shown in the following examples.
\begin{model}[Distributed Fr\'{e}chet regression]\label{exmp::distributed F regression}
 When $\ell=2$ in (\ref{e:mom}), it gives the distributed Fr\'{e}chet regression where the block-wise estimators are aggregated via the Fr\'{e}chet mean.
\end{model}
\begin{model}[MoM Fr\'{e}chet regression]\label{exmp::MoM F regression}
 When $\ell=1$ in (\ref{e:mom}), it gives the MoM Fr\'{e}chet regression where the block-wise estimators are aggregated via the Fr\'{e}chet median.
\end{model}

\subsection{Distributed Fr{\'e}chet regression}
Multisource data are common in the distributed regression literature, where independent samples collected across sites may follow distinct joint laws, i.e., $\mathcal{P}_k\neq \mathcal{P}_j$ for $k\neq j$ and $k,j=1,\ldots, B$. 
The distributional heterogeneity imposes technical challenges in the study of statistical properties for the corresponding estimators. In this section, we therefore establish the uniform convergence of the weighted Fr{\'e}chet estimator of the $\ell$-th order under heterogeneous distributions.

To this end, we first impose the following condition that ensures the existence of $\widehat{r}(x,\pmb{\kappa};\ell)$ and $r(x,\pmb{\kappa};\ell)$.
\begin{rassump}\label{condition::convexity}
For each $x\in\mathcal{X}$, the block-wise Fr{\'e}chet estimators $\{\widehat{r}_b(x)\}_{b=1}^B$ and their population quantities $\{r_b(x)\}_{b=1}^B$ satisfy Condition~\ref{condition R0}, i.e., $\widehat{r}_b(x), r_b(x)\in \mathcal{C}_{b,x}$, where $\mathcal{C}_{b,x}\subseteq \Omega$ is a convex subset of $\Omega$. Moreover, we assume that $\mathcal{C}_x = \bigcap_{b=1}^B \mathcal{C}_{b,x} \neq \emptyset$.
\end{rassump}
\noindent
Since $\mathcal{C}_x$ is a nonempty and metrically convex subset of $\Omega$, this condition implies that $\mathcal{C}_x$ has a metric convex hull \citep[Theorem 2.17 of ]{khamsi2011introduction}, i.e., given a weight vector $\pmb{\kappa}\in \mathcal{K}$ and $\ell=\{1,2\}$,
\[\mathcal{C}_{x,\pmb{\kappa}}=\left\{\argmin_{\nu\in \Omega} \sum_{b=1}^B \kappa_b d^{\ell}(\nu, r_b(x)) \Big| r_b(x)\in \mathcal{C}_{b,x}, \bigcap_{b=1}^B \mathcal{C}_{b,x} \neq \emptyset \right\}\neq \emptyset.\]
Consequently, both $\widehat{r}(x,\pmb{\kappa};\ell)$ and $r(x,\pmb{\kappa};\ell)$ exist, even though they may not necessarily be unique. Moreover, in order to establish the uniform convergence of $\widehat{r}(x,\pmb{\kappa};\ell)$, analogously to Theorem~\ref{thm::uniform convergence 2}, Conditions~\ref{condition R0} and \ref{condition R1} should be adapted accordingly as follows.
\begin{rassump}
\label{condition R0 adpt}
Write $\widehat{r}^{(\ell)} := \widehat{r}(x,\pmb{\kappa};\ell)$ and $r^{(\ell)} :=r(x,\pmb{\kappa};\ell)$ and assume $\widehat{r}^{(\ell)},r^{(\ell)}\in\mathcal{C}_{x,\pmb{\kappa}}$ with probability tending to one as $n \rightarrow \infty$, where $\ell=\{1,2\}$. For each $x\in\mathcal{X}$ and $\pmb{\kappa}\in \mathcal{K}$, 
there exists $\widetilde{r}^{(\ell)}\in \mathcal{C}_{x,\pmb{\kappa}}$ such that $\widetilde{r}^{(\ell)}\notin \{\widehat{r}^{(\ell)}, r^{(\ell)}\}$ and 
$$d(r^{(\ell)}, \widetilde{r}^{(\ell)}) + d(\widehat{r}^{(\ell)}, \widetilde{r}^{(\ell)}) = d(r^{(\ell)}, \widehat{r}^{(\ell)}).$$
Moreover, the empirical loss function $\widehat{L}^{(\ell)}(\nu; x, \pmb{\kappa})=\sum_{b=1}^B \kappa_b d^{\ell}(\nu,\widehat{r}_b(x))$ satisfies
$$\widehat{L}^{(\ell)}(\widetilde{r}^{(\ell)}; x, \pmb{\kappa}) \leq \delta(x) \widehat{L}^{(\ell)}(\widehat{r}^{(\ell)}; x, \pmb{\kappa}) + \{1-\delta(x)\} \widehat{L}^{(\ell)}(r^{(\ell)}; x, \pmb{\kappa}),$$
for some constant $\delta(x) \in [0,1]$.
\end{rassump}
\begin{rassump}\label{condition R1 adpt}
For each $x\in\mathcal{X}$ and $\pmb{\kappa}\in \mathcal{K}$, suppose that $r^{(\ell)}$ is unique and the population loss function $L^{(\ell)}(\nu;x,\pmb{\kappa}) = \sum_{b=1}^B \kappa_b d^{(\ell)}(\nu, r_b(x))$ satisfies
$$\inf_{x\in\mathcal{X},\pmb{\kappa}\in\mathcal{K}}\inf_{\nu\in\Omega:d(\nu,r^{(\ell)})>\zeta} \left( L^{(\ell)}(\nu;x,\pmb{\kappa}) - L^{(\ell)}(m^{(\ell)}; x,\pmb{\kappa}) \right)>0,$$
for any $\zeta>0$.
\end{rassump}
By noting that $\mathcal{C}_{x,\pmb{\kappa}}$ is a metrically convex set, the convexity in condition~\ref{condition R0 adpt} is straightforward to verify. Condition \ref{condition R1 adpt} is a standard identifiability condition for an M-estimator. The following theorem establishes the uniform convergence of $\widehat{r}(x,\pmb{\kappa};\ell)$.

\begin{theorem}\label{theorem::mom converge}
Assume that the block-wise Fr{\'e}chet estimators are consistent uniformly over $x\in \mathcal{X}$, i.e, $\sup_{x\in \mathcal{X}} d(\widehat{r}_{b}(x), r_{b}(x))=o_p(1)$ for $b=1, \ldots, B$. Under Conditions~\ref{condition::convexity}, \ref{condition R0 adpt} and \ref{condition R1 adpt}, the weighted Fr{\'e}chet estimator $\widehat{r}(x,\pmb{\kappa};\ell)$ for $\ell\in\{1,2\}$ satisfies
\[\sup_{x\in \mathcal{X}, \pmb{\kappa}\in \mathcal{K}} d(\widehat{r}(x,\pmb{\kappa};\ell),r(x,\pmb{\kappa};\ell)) =o_p(1).\]     
\end{theorem} 
\noindent 
This theorem states the uniform convergence of $\widehat{r}(x,\pmb{\kappa};\ell)$ over $x\in \mathcal{X}$ and $\pmb{\kappa}\in \mathcal{K}$. When $\ell=2$ and blocks are homogeneous, so that $\mathcal{P}_1 = \ldots = \mathcal{P}_B=\mathcal{P}$, the population aggregation $r^{(2)}(x, \pmb{\kappa})$ reduces to $r(x)$, i.e., $r^{(2)}(x, \pmb{\kappa}) = \argmin_{\nu\in\Omega} \sum_{b=1}^B \kappa_b d^2(\nu, r(x)) = r(x)$. In this setting, the role of $\pmb{\kappa}$ is therefore to improve the finite-sample behavior of $\widehat{r}^{(2)}(x, \pmb{\kappa})$, with the aim of reducing bias or achieving the same asymptotic efficiency as $\widehat{r}(x)$ using the full datasets; see for example, the jackknifing aggregation method for Euclidean responses \citep{wu2023subsampling}. 
Under distributional heterogeneity, however, the block-wise population targets $r_b(x)$ may differ, and therefore $\pmb{\kappa}$ plays a more important role by determining not only the finite-sample behavior of $\widehat{r}(x,\pmb{\kappa};\ell)$ but also the population aggregate $r^{(2)}(x, \pmb{\kappa})$ to which it converges.

When $\ell=1$, the aggregation $r^{(1)}(x, \pmb{\kappa})$ defines the weighted Fr{\'e}chet median, also known as the weighted geometric median \citep{fletcher2009geometric}. Hence, $\widehat{r}^{(1)}(x, \pmb{\kappa})$ acts as a robust location estimator and Theorem \ref{theorem::mom converge} proves consistency of the weighted Fr{\'e}chet median estimator. In practice, however, the target of interest is often the conditional Fr{\'e}chet mean, from which the Fr{\'e}chet median generally differs. Therefore, we shall study $\widehat{r}^{(1)}(x, \pmb{\kappa})$ as a robust mean estimator in the subsequent section.

\subsection{Robust Mean Estimation via MoM}
Robust estimation of the conditional Fr{\'e}chet mean is one of the primary motivations for introducing the weighted Fr{\'e}chet estimator. As discussed previously, robust approaches via modifying the loss function, i.e., robust M-estimation, may become inadequate in general metric spaces, while median-of-means (MoM)
type estimators have been shown to exhibit strong robustness in several non-Euclidean settings. In particular, the weighted Fr{\'e}chet estimator with $\ell=1$ includes the MoM estimator as a special case and therefore provides an extension of MoM methods to metrically convex spaces. In this section, we will study its robustness and establish a deviation bound under block-wise contamination. 

Specifically, suppose that $\{(X_i, Y_i)_{i=1}^n\}\in\mathcal{X}\times \Omega$ is an i.i.d sample from a joint probability distribution $\mathcal{P}$, while a proportion of the responses may be contaminated. We then randomly and evenly split the sample into $B$ disjoint blocks and obtain the block-wise Fr{\'e}chet estimators as in~\eqref{e:blockwise}. The primary objective is to estimate $r(x) = \argmin_{\nu\in\Omega} \mathbb{E}_{\mathcal{P}} [d^2(\nu, Y )\big\vert X = x]$.
Moreover, the weighted Fr{\'e}chet estimator in~\eqref{e:mom} encompasses the MoM
estimator by specifying $\ell = 1$ and $\kappa_b = B^{-1}$, i.e., $\widehat{r}^{(1)}(x,\pmb{\kappa}) = \argmin_{\nu\in\Omega} B^{-1}\sum_{b=1}^B d(\nu, \widehat{r}_{b}(x))$,
where $\pmb{\kappa} = (B^{-1},\ldots,B^{-1})^\top$. Intuitively, although blocks with contaminated observations may generate distorted estimators, uncontaminated blocks continue to produce local estimators concentrated around $r(x)$. Therefore, the weighted Fr{\'e}chet median $\widehat{r}^{(1)}(x,\pmb{\kappa})$ mitigates the influence of contaminated blocks and yields a robust estimator of $r(x)$.
The following theorem formalizes this intuition by establishing a deviation bound when a proportion $\tau$ of the block-wise estimators are allowed to be arbitrarily corrupted.

\begin{theorem}\label{theorem::mom deviation}
Let $\{\widehat{r}_{b}(x)\}_{b=1}^B$ be the block-wise Fr{\'e}chet estimators as defined in \eqref{e:blockwise}, and let $\widehat{r}^{(1)}(x, \pmb{\kappa})$ be the first-order weighted Fr{\'e}chet median estimator defined in \eqref{e:mom} with $\pmb{\kappa}=(\kappa_1,\ldots,\kappa_B)^\top\in \mathcal{K}$. Suppose that the block index set is partitioned as $\{1,\ldots,B\}=\mathcal{B}_1\cup\mathcal{B}_2$ and $\mathcal{B}_1\cap\mathcal{B}_2=\emptyset$, where $\mathcal{B}_2$ denotes the set of blocks that are arbitrarily corrupted. The corrupted blocks have total weight $\sum_{b\in \mathcal{B}_2}\kappa_b = \tau$, for some $\tau\in [0,\frac{\alpha-p}{1-p})$ and $0< p < \alpha <1/2$. Consequently, $\sum_{b\in \mathcal{B}_1} \kappa_b = (1-\tau)$.
Assume that the blocks that are not corrupted are homogeneous, in the sense that $\mathcal{P}_b = \mathcal{P}$ for all $b\in \mathcal{B}_1$, and that there exists $\epsilon>0$ such that 
\begin{equation}\label{e:block prob}
\mathbb{P}\left\{\sup_{x\in \mathcal{X}} d(\widehat{r}_{b}(x), r(x))>\epsilon\right\}\le p, \qquad b\in \mathcal{B}_1.
\end{equation}
Define $C_{\alpha} = 2(1-\alpha)/(1-2\alpha)$ and 
\begin{equation*}
\psi(\alpha-\tau, p, \pmb{\kappa}) = \sup_{\lambda>0} \left\{\lambda (\alpha-\tau) - \sum_{b\in \mathcal{B}_1} \log (1-p + p \exp(\lambda \kappa_b))\right\}.
\end{equation*}
Then, it holds that
\begin{equation*}
\sup_{x\in \mathcal{X}}  \mathbb{P}\left\{d(\widehat{r}^{(1)}(x,\pmb{\kappa}), r(x))>C_{\alpha}\epsilon \right\}\leq 
\exp\left\{-\psi(\alpha-\tau, p, \pmb{\kappa})\right\}.
\end{equation*}
\end{theorem}

This theorem shows that the Fr{\'e}chet median takes a collection of independent estimators $\{\widehat{r}_b(x)\}_{b=1}^B$, which are only weakly concentrated around $r(x)$, into an aggregate estimator that achieves exponential concentration \citep{minsker2015geometric,kim2025robust}. The constant $C_{\alpha}$ quantifies the geometric discrepancy of $\{\widehat{r}_b(x)\}_{b=1}^B$ around $\widehat{r}^{(1)}(x,\pmb{\kappa})$; see also the geometric interpretation of $C_{\alpha}$ below Lemma 2.1 of \citet[][p.2311]{minsker2015geometric}. This constant coincides with that in \citet{minsker2015geometric} for Banach spaces and it may be improved by imposing additional Riemannian structure, such as the Lipschitz continuity on the Riemannian logarithm map \citep{lin2024robust} or bounded non-positive curvature \citep{kim2025robust}. The condition \eqref{e:block prob} is reasonable provided that the sample sizes of blocks in $\mathcal{B}_1$ are
sufficient and the block-wise Fr{\'e}chet estimators are consistent. Apart from the extension to regression modelling, this theorem extends Theorem 3.1 of \cite{minsker2015geometric} from Euclidean responses to object-valued responses. Related extensions have been studied, for example, by \cite{lin2024robust} for smooth Riemannian manifolds, such as spheres and spaces of positive-definite matrices, and by \cite{kim2025robust} for CAT spaces with bounded sectional curvature, e.g., metric trees. Our result is formulated for metrically convex spaces, a weaker geometric assumption that broadens the scope of applicability. Specifically, it includes prominent examples beyond those in \cite{lin2024robust} and \cite{kim2025robust}, such as Wasserstein spaces over $\mathbb{R}^d, d\ge 2$.

Moreover, in the following collorary, we establish a concentration result of the MoM estimator at a prescribed confidence level with an explicit choice of the number of blocks. Such a choice of the number of blocks is not specified in the related results of \cite{lin2024robust} and \cite{kim2025robust}.

\begin{corollary}\label{corollary::mom}
Let $\{\widehat{r}_{b}(x)\}_{b=1}^B$ be the block-wise Fr{\'e}chet estimators obtained from either the global or local Fr{\'e}chet regression, with the corresponding conditions of Theorem~\ref{prop::global uniform conv} and Theorem~\ref{prop::local uniform conv}, respectively. Furthermore, suppose the equal-weight median-of-means Fr{\'e}chet estimator, $\widehat{r}^{(1)}(x, \pmb{\kappa})$ with $\pmb{\kappa}=(1,\ldots, 1)^\top/B$ and $n_1=\ldots=n_B=\lfloor n/B\rfloor$ where $n=\sum_{b=1}^B n_b$, satisfies conditions in Theorem \ref{theorem::mom deviation}. 
For a fixed confidence level $0<\delta\le 1$ and contamination proportion $0\le \tau<1/2$, we set the number of blocks as $B:= B(\delta,\tau)= \lfloor \frac{\log(1/\delta)}{(1-\tau)\iota(\tau)}\rfloor + 1$. Then, it holds that
\begin{equation*}
    \sup_{x\in\mathcal{X}} \mathbb{P}\left(d(\widehat{r}^{(1)}(x, \pmb{\kappa}), r(x))>U(n,\delta,\tau)\right)\leq \delta,
\end{equation*}
where 
\begin{gather*}
U(n,\delta,\tau)=c(\tau)n^{-\upsilon_j} \left\{\log\left(\frac{2}{\delta(1-\tau)}\right)\right\}^{\upsilon_j},    \\
c(\tau) = \frac{3 - 4\tau^2}{1-4\tau^2} \frac{2D}{(1/2 -\tau)^2} \{(1-\tau) \iota(\tau)\}^{-\upsilon_j},\\
\iota(\tau)=\frac{3/4-\tau^2}{1-\tau}\log\left(\frac{3/4-\tau^2}{(1-\tau)\{1-(1/2-\tau)^2/2\}}\right)+\frac{(1/2-\tau)^2}{1-\tau}\log\left(\frac{2}{1-\tau}\right),
\end{gather*}
and $D>0$ is some constant and $\upsilon_j$, $j=1,2$, corresponds the uniform convergence rate for the block-wise Fr{\'e}chet estimators, i.e., $\upsilon_1 = 1/(2(\alpha' - 1))$ for the global Fr{\'e}chet regression as in as in Theorem 2 of \cite{petersen2019frechet} and $\upsilon_2 = 1/(\alpha+2\beta-3+\eta)$ for the local Fr{\'e}chet regression as in Theorem 2 of \cite{chen2022uniform}. 
\end{corollary}

This corollary provides a finite-sample accuracy guarantee for the MoM Fr{\'e}chet estimator at confidence level $1-\delta$. Specifically, with probability at least $1-\delta$, $\widehat{r}^{(1)}(x, \pmb{\kappa})$ approximates $r(x)$ with uniform estimation error bounded by $U(n,\delta, \tau)$. The factor $n^{-\upsilon_j}$ in $U(n,\delta, \tau)$ inherits from the block-wise Fr{\'e}chet estimator; hence a larger rate exponent $\upsilon_j$ yields a sharper deviation bound. 
Moreover, for some fixed $\tau_0$ and the number of blocks $B(\delta, \tau_0)$, the MoM Fr{\'e}chet estimator allows for $\tau_0 B_0$ blocks to be contaminated with $B_0 := \lfloor \log(1/\delta) / \iota(0) \rfloor + 1$, as $B_0/B(\delta,\tau_0)\le 1$. 
Suppose further that the number of contaminated blocks grows slower than the sample size, i.e., $\tau_0 B_0 = o(n)$ as $n\to \infty$, and set the confidence level $\delta_0:= \exp(- \{\lceil B_0 \rceil-1\}/\iota(0))$, $\widehat{r}^{(1)}(x, \pmb{\kappa})$ provides asymptotic consistent estimation of $r(x)$ because $\delta_0\to 0$ and $U(n, \delta_0, \tau_0)\to 0$ as $n\to\infty$.

On the other hand, allowing $\tau$ to approach $1/2$ increases robustness in terms of high breakdown but worsens the constant $c(\tau)=O((1/2-\tau)^{-2\upsilon_j+3})$, and hence enlarges the deviation bound $U(n,\delta, \tau)$. Thus, this corollary makes explicit the trade-off between estimation accuracy and robustness.

\section{Exceedance set via Fr\'{e}chet regressions}\label{sect::level set}

In this section, we consider the problem of exceedance set estimation for conditional Fr\'{e}chet mean utilizing certain Fr\'{e}chet regression model such as the ones in Models~\ref{exmp::global F regression}--\ref{exmp::deep F regression}. Our objective is to characterize regions of the covariate space $\mathcal{X}$ where a given functional of the conditional Fr\'{e}chet mean exceeds a prescribed threshold.

Here, we further assume $m(x)$ in (\ref{formual::general population Frechet M est}) is unique. For a compact subset $\bar{\mathbb{R}}$ of $\mathbb{R}$, consider a functional $\Lambda:\Omega^L\to \bar{\mathbb{R}}$ that quantifies a scalar quantity of interest of random objects in the product space $\Omega^L=\underbrace{\Omega \times \cdots \times \Omega}_{L \text{ times}}$, and the conditional Fr\'{e}chet means $\{m_l(x)\}_{l=1}^L$ given $X=x$, one has
\[
\Delta(x)=\Lambda\left[\{m_l(x)\}_{l=1}^L\right],~x\in\mathcal{X}.
\]
Given a specified threshold $\rho$, we define the exceedance set $L(\rho)$ as
\[
L(\rho)=\{x\in\mathcal{X}:\Delta(x)\geq\rho\},~\xi\in\mathbb{R}.
\]
The set $L(\rho)$ is a collection of all points in the compact domain $\mathcal{X}$ where $\Delta(x)$ exceeds or equals to the threshold $\rho$. The estimator of the exceedance set $L(\rho)$ is 
\[
\widehat{L}(\rho)=\{x\in\mathcal{X}:\Lambda[\{\widehat{m}_l(x)\}_{l=1}^L)]>\rho\}
\]
where $\widehat{m}(x)$ is a uniform consistent estimator of $m(x)$ such as $\sup_{x\in\mathcal{X}}d(\widehat{m}(x),m(x))=o_p(1)$, see the uniform convergence results established in Sections~\ref{subsect::regression 1}, \ref{subsect::regression 2}, and \ref{sect::MoM}. Note that when setting $\rho=\max_{x\in\mathcal{X}}\Delta(x)$ in $L(\rho)$, it reduces to the location of extrema in \cite{chen2022uniform}. To establish the theoretical property of this level-set estimation, one assumption is needed as follows.

\begin{lassump}\label{condition L1}
    There exists some constants $D_2$ and $\lambda_1>1$ such that $\left|\Lambda(\nu_1)-\Lambda(\nu_2)\right|\leq D_2 d^{\lambda_1}(\nu_1,\nu_2)$ holds for any $\nu_1,\nu_2\in\Omega$.
\end{lassump}

\ref{condition L1} is a common assumption in the level-set estimation, see, e.g., \cite{willett2007minimax}, which guarantees the consistency of the level-set estimator $\widehat{L}(\rho)$.

\begin{proposition}\label{prop::level set est}
    Under \ref{condition L1}, if $\sup_{x\in\mathcal{X}}d(\widehat{m}_l(x),m_l(x))=o_p(1)$ for all $l=1,2,\ldots,L$, then
    \[
    \mathbb{P}\left\{\liminf_{n\to\infty}\widehat{L}(\rho)  \supset {\rm int}\{L(\rho)\}\right\}\to 1
    \]
    where ${\rm int}\{L(\rho)\}$ is the interior of $L(\rho)$ and 
    \[
    \mathbb{P}\left\{\limsup_{n\to\infty}\widehat{L}(\rho) \subset \bar{L}(\rho)\right\}\to 1
    \]
    where $\bar{L}(\rho)$ is the closure of $L(\rho)$.
\end{proposition}
    
The level set estimator $\widehat{L}(\rho)$ consistently recovers the true level set $L(\rho)$ up to its boundary. Proposition \ref{prop::level set est} establishes this consistency result under the condition that $\widehat{m}(x)$ is a uniform consistent estimator of the conditional Fr\'{e}chet mean. Such uniform consistency has been rigorously established in different contexts such as Theorem~\ref{prop::global uniform conv} for the global Fr\'{e}chet regression estimate, Theorem~\ref{prop::local uniform conv} for the local Fr\'{e}chet regression estimate, and Theorem \ref{theorem::mom converge} for the MoM Fr\'{e}chet regression estimate. Together, these results provide a unified theoretical foundation guaranteeing that the proposed level set estimator inherits the consistency properties of the underlying regression estimators.

In many applications, interest lies in the aggregate magnitude of the exceedance rather than the actual set $L(\rho)$ itself. Accordingly, inspired by \cite{kundu2025exceedance} on Hilbert-valued responses, we focus on the exceedance function $\ell(\rho)$ given by 
\[
\ell(\rho) = \frac{\lambda \{L(\rho)\} }{ \lambda(\mathcal{X}) },~\rho\in\mathbb{R},
\]
where $\lambda$ denotes the Lebesgue measure. The value $\ell(\rho)$ reflects the proportion of the compact domain $\mathcal{X}$ occupied by the exceedance event. The corresponding empirical estimator takes the form of $\widehat{\ell}(\rho)= \lambda \{\widehat{L}(\rho\}) / \lambda(\mathcal{X})$, with the asymptotic behavior shown in the following proposition.

\begin{proposition}\label{prop::exceedance}
     Under \ref{condition L1}, if $\sup_{x\in\mathcal{X}}d(\widehat{m}(x),m(x))=o_p(1)$, and $\lim_{\varepsilon \to 0} \sum_{\rho\in \bar{\mathbb{R}}} \lambda(\{x \in \mathcal{X} : |\Delta(x) - \rho| \leq \varepsilon\}) = 0$, then $\sup_{\rho\in \bar{\mathbb{R}}}\left|\widehat{\ell}(\rho)-\ell(\rho)\right|=o_p(1)$.
\end{proposition}

In Proposition~\ref{prop::exceedance}, $\lim_{\varepsilon \to 0} \sum_{\rho\in \bar{\mathbb{R}}} \lambda(\{x \in \mathcal{X} : |\Delta(x) - \rho| \leq \varepsilon\}) = 0$ is a uniform margin condition which means that the probability density function of $\Delta(x)$ exists and is bounded over the interval $\bar{\mathbb{R}}$.

Note that the exceedance function $\ell(\cdot)$ has the characteristics of a survival function as it is a decreasing function of $\rho$. We further have that $F_e(\rho)=1-\ell(\rho)$ behaves like a distribution function, leading to a construction of the exceedance quantile function
\[
Q_e(\varrho)=F_e^{-1}(\varrho)=\inf \left\{ \rho:F_e(\rho)\geq \varrho,~{\rm for~all~}\varrho\in(0,1) \right\}.
\]
The corresponding empirical estimators of $F_e(\rho)$ and $Q_e(\varrho)$ are $\widehat{F}_e(\rho)=1-\widehat{\ell}(\rho)$ and $\widehat{Q}_e(\varrho)=\widehat{F}_e^{-1}(\varrho)$, respectively.

\section{Simulation studies}\label{sect::sim data}

In this section we first evaluate the finite-sample performance of the weighted Fr\'{e}chet aggregation in Section~\ref{subsect::sim MMoM}, and then study the exceedance set estimation in Section~\ref{subsect::exceedance set simulation}.

\subsection{Weighted Fr\'{e}chet aggregation}\label{subsect::sim MMoM}

To evaluate the finite-sample performance of the proposed weighted Fr\'{e}chet aggregation framework in Section~\ref{sect::MoM}, we adopt the network data generation mechanism introduced by \cite{zhou2022network} whose details can be found in Section~\ref{supp subsect::sim MMoM} in the supplementary material. The response located in the space of graph Laplacian with the Frobenius metric $d_{\rm F}$, which is the same with the parameter space $\Omega$.

To quantify the estimation accuracy of the Fr\'{e}chet regression estimators defined in \eqref{e:mom}, we compute the mean integrated squared error (MISE) over $S=200$ independent replications as ${\rm MISE} = S^{-1} \sum_{s=1}^S \int_0^1 d_{\rm F}^2 \big(\widehat{r}^{(\ell,s)}(x), r(x)\big) dx$, where $\widehat{r}^{(\ell,s)}(x)$ represents the estimated conditional mean in the $s$-th replication and $r(x)$ is the true conditional Fr\'{e}chet mean. We evaluate the distributed Fr\'{e}chet regression where $\ell=2$ and the median-of-means Fr\'{e}chet regression where $\ell=1$. For both approaches, we employ uniform aggregation weights, setting $\kappa_1 = \kappa_2 = \dots = \kappa_B = 1/B$.

We first evaluate the uncontaminated setting without outliers. The total sample size is $n = N_B B$, where the block-wise sample size varies across $N_B \in \{200, 500, 1000\}$ and the number of blocks spans $B \in \{60, 80, 100\}$. The empirical MISE results, summarized in Table~\ref{table::uncontaminate setting}, reveal several expected asymptotic behaviors. For any fixed number of blocks $B$, the MISE monotonically decreases for both the distributed and MoM Fr\'{e}chet regressions as the block-wise sample size $N_B$ grows. Similarly, for any fixed $N_B$, increasing the total number of blocks systematically improves estimation accuracy. Most notably, the distributed estimator and the robust MoM estimator exhibit highly comparable MISE profiles across all $(N_B, B)$ configurations. 

\begin{table}[!b]
\caption{Mean integrated squared error ($\times 10^{-2}$) for number of blocks varying in $\{60,80,100\}$ and block-wise sample size $N_B\in\{200,500,1000\}$ with two Fr\'{e}chet models evaluated: distributed Fr\'{e}chet regression (DFR) and median-of-mean Fr\'{e}chet regression (MoMFR). \label{table::uncontaminate setting}}
\centering
\begin{tabular}{c|c|ccc}
  \hline
Model&Block Number&$N_{B}=200$ & $N_{B}=500$ & $N_{B}=1000$ \\ 
  \hline
    &60  &0.1018 & 0.0433 & 0.0280 \\ 
DFR &80  &0.0709 & 0.0437 & 0.0222 \\ 
    &100 &0.0498 & 0.0374 & 0.0170 \\ 
   \hline
       &60  &0.1012 & 0.0442 & 0.0264 \\ 
MoMFR  &80  &0.0699 & 0.0450 & 0.0225 \\ 
       &100 &0.0498 & 0.0388 & 0.0166 \\ 
   \hline
\end{tabular}
\end{table}

We then evaluate the contaminated setting with outliers. We inject adversarial outliers into the simulated graph Laplacian responses; further details may be found in Section~\ref{supp subsect::sim MMoM} in the supplementary material. For a comprehensive evaluation, we benchmark the proposed distributed and MoM Fr\'{e}chet regressions against two established methods covering the global Fr\'{e}chet regression detailed in Model \ref{exmp::global F regression}, and the weight-regularized Fr\'{e}chet regression proposed by \cite{li2025robust} with public code available at \url{https://github.com/lee199950120HAO/robust-FR}. We fix the total sample size at $n=1000$ and vary the contamination ratio $\epsilon \in \{0, 0.01, \dots, 0.06\}$. For both the distributed and MoM estimators, we partition the data into $B=100$ blocks. This gives a theoretical breakdown point of $0.05$ for the MoM regression.

The MISE for all four models are reported in Table~\ref{table::robust compare}. In the uncontaminated regime where $\epsilon=0$, the standard global Fr\'{e}chet regression attains the lowest MISE, while the distributed and MoM estimators incur a efficiency loss. However, as anticipated, the introduction of a small fraction of outliers such as $\epsilon = 0.01$ causes both the global and distributed estimators to break down. In contrast, the MoM estimator remains highly robust. While it eventually breaks down once the contamination ratio exceeds its theoretical limit where $\epsilon \geq 0.05$, its MISE remains strictly superior to those of the global and distributed baselines across all contaminated scenarios. The weight-regularized Fr\'{e}chet regression exhibits several limitations in this setting. As its robustness hinges critically on hyperparameter tuning, any suboptimal selection renders the estimator highly sensitive to outliers. Consequently, its MISE exceeds that of the MoM estimator for all $\epsilon > 0$. Furthermore, when the contamination ratio surpasses $0.02$, the performance of the weight-regularized estimator degrades severely, occasionally yielding errors larger than those of the global Fr\'{e}chet regression. 

\begin{table}[!tb]
\caption{Mean integrated squared error for contamination ratio varying in $\{0,0.01,\ldots,0.06\}$ and average run time with four Fr\'{e}chet models evaluated: global Fr\'{e}chet regression (GFR), distributed Fr\'{e}chet regression (DFR), weight-regularized  Fr\'{e}chet regression (WRFR), and median-of-mean Fr\'{e}chet regression (MoMFR). \label{table::robust compare}}
\centering
\begin{tabular}{c|cccc}
  \hline
Contamination Ratio&GFR & DFR & WRFR & MoMFR \\ 
  \hline
0  &  0.0187 & 0.0244 & 0.0188 & 0.0210 \\ 
0.01&  8.5493 & 9.0511 & 0.1049  & 0.0414 \\ 
0.02&  33.2608 & 35.1038 & 0.3642  & 0.1071 \\ 
0.03&  74.3178 & 77.8056 & 80.9982  & 0.3730 \\ 
0.04&  131.4262 & 136.0896 & 42.0440 & 9.2850 \\ 
0.05&  205.8379 & 213.2294 &  82.1643  & 54.7823 \\ 
0.06&   294.2659 & 302.1813 & 204.5589 & 132.4223 \\ 
   \hline
Run Time & 10 seconds & 1 seconds & 3 hours & 1 seconds \\
   \hline
\end{tabular}
\end{table}

Beyond studying robustness, we also evaluate computational efficiency by recording the average single-core run time per replication for each scenario in Table~\ref{table::robust compare}. Note that the reported times for both the distributed and MoM Fr\'{e}chet regressions represent the entire model-fitting procedure executed sequentially on a single core. Remarkably, both the distributed and MoM estimators require approximately 1 second to fit, executing faster than the standard global regression. Conversely, the weight-regularized Fr\'{e}chet regression demands an average of nearly 3 hours per replication. This high computational cost is driven by the extensive hyperparameter selection.

\subsection{Exceedance set estimation}\label{subsect::exceedance set simulation}

Here we  evaluate the finite-sample performance of the proposed exceedance set estimator via 200 Monte Carlo replicates across varying sample sizes $n \in \{100, 200, 500\}$. Similar to the data generation setting in Section~\ref{subsect::sim MMoM}, both the data space and the parameter space consist of the space of network Laplacians equipped with the Frobenius distance. The data generation procedure is specified in Section~\ref{supp subsect::exceedance set simulation} in the supplementary material.

The quantity of interest formulated in Section~\ref{sect::level set} is defined as $\Delta(x) = d_{\rm F}\big(r_{(1)}(x), r_{(2)}(x)\big)$, where $r_{(1)}(x), r_{(2)}(x)$ are the underlying conditional Fr\'{e}chet means for two classes of generated graph Laplacians given in Section~\ref{supp subsect::exceedance set simulation} in the supplementary material. We wish to evaluate the performance of the exceedance set estimator of $L(8)=\{x\in[0,1]:\Delta(x)\geq 8\}$ based on the local Fr\'{e}chet regression estimators for $r_{(1)}(x)$ and $r_{(2)}(x)$.

\begin{figure}[!b]
\centering
\includegraphics[width=1\linewidth]{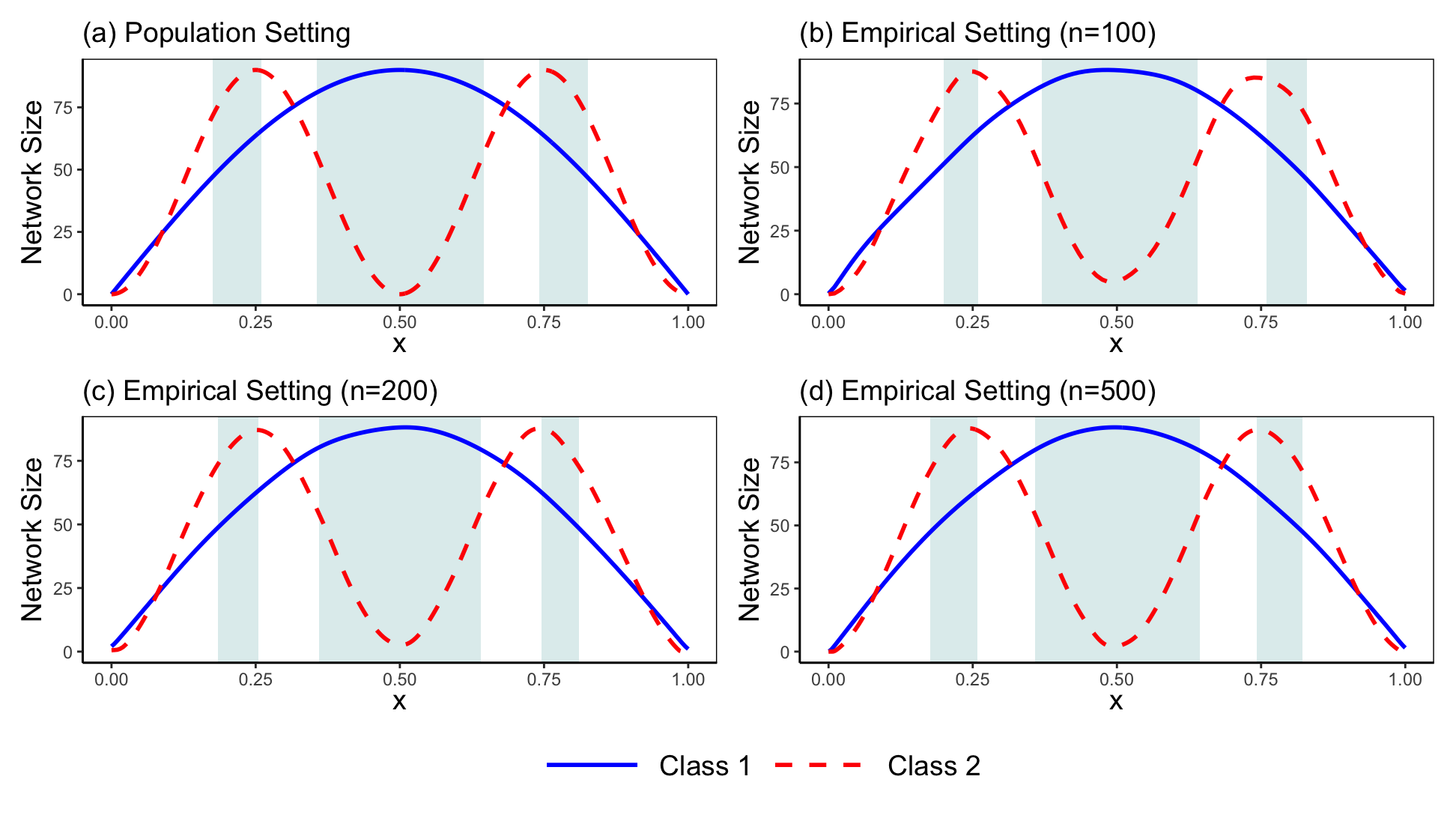} 
\caption{Network sizes and exceedance regions for (a) the true population setting, and empirical settings based on sample sizes of (b) $n=100$, (c) $n=200$, and (d) $n=500$. Solid blue and dashed red lines represent the graph Laplacians for Class 1 and Class 2, respectively. The exceedance regions are shaded in azure.}\label{fig::1}
\end{figure}

Figure~\ref{fig::1} visualizes the network size of the network Laplacians which is calculated via the sum of all the elements of the corresponding adjacency matrix, alongside the identified exceedance regions. Panel (a) in Figure~\ref{fig::1} illustrates that the population network size trajectory for the first class, $r_{(1)}(x)$, exhibits a single global maximum over the domain of the predictor $X$, whereas the second class, $r_{(2)}(x)$, displays a bimodal network size trajectory. Panels (b), (c), and (d) present the network sizes of the estimated conditional Fr\'{e}chet means, together with their corresponding empirical exceedance sets, for a representative simulation replicate across sample sizes $n \in \{100, 200, 500\}$. This replicate demonstrates that the empirical estimators successfully isolate the three disjoint exceedance regions with the accuracy of the estimated boundaries improving substantially as the sample size increases to $n=500$.

To study the finite-sample performance of the proposed exceedance set estimator, note that the estimated number of regions may differ from the ground truth due to sample variability. To deal with this, we employ a nearest-neighbor matching algorithm. Specifically, the midpoint of each true region is calculated and strictly paired with the estimated region possessing the closest midpoint, without imposing a maximum distance threshold. Consequently, if the estimator fails to detect a true underlying region, its boundaries are evaluated against the nearest available spatial estimate. In this framework, we report the empirical bias and mean squared error (MSE) of the boundary estimates for each region. To further characterize structural reliability, we report the proportion of replicates with mismatched region counts, alongside the average deviation in region counts. These results are presented in Table~\ref{table::level set simulation}.

\begin{table}[!b]
    \centering
    \caption{Finite-sample performance of the level set estimator across 200 Monte Carlo replicates. Region number accuracy represents the global rate of incorrect region counts and the mean deviation from the true count. Region boundary accuracy reports the empirical Bias ($\times 10^{-1}$) and mean squared error (MSE $\times 10^{-2}$) for the left boundaries (LB) and right boundaries (RB) of the three true regions, computed via unconstrained nearest-neighbor matching.}
    \label{table::level set simulation}
    \begin{tabular}{c c c c c c c c}
        \toprule
        \multirow{2}{*}{$n$} & \multicolumn{2}{c}{Region Number Accuracy} & \multirow{2}{*}{Region} & \multicolumn{4}{c}{Region Boundary Accuracy} \\
        \cmidrule(lr){2-3} \cmidrule(lr){5-8}
        & Error Rate (\%) & Avg. Diff. & & Bias (LB) & MSE (LB) & Bias (RB) & MSE (RB) \\
        \midrule
        \multirow{3}{*}{100} & \multirow{3}{*}{2.00} & \multirow{3}{*}{-0.02} 
        & 1 & 0.0275  & 0.0411  & -0.0280 &  0.1602 \\
        & & & 2 & -0.0135 &  0.0013 &   0.0200 &  0.0020 \\
        & & & 3 & 0.0405  & 0.1625  & -0.0230  & 0.0411 \\
        \midrule
        \multirow{3}{*}{200} & \multirow{3}{*}{0.00} & \multirow{3}{*}{0.00} 
        & 1 & 0.0260  & 0.0046 &  -0.0562 &  0.0052 \\
        & & & 2 & 0.0075 &  0.0005  & -0.0093 &  0.0005 \\
        & & & 3 & 0.0518  & 0.0044 &  -0.0273 &  0.0056 \\
        \midrule
        \multirow{3}{*}{500} & \multirow{3}{*}{0.00} & \multirow{3}{*}{0.00} 
        & 1 & 0.0109  & 0.0014 &  -0.0313  & 0.0015 \\
        & & & 2 & 0.0115 &  0.0003  & -0.0113  & 0.0003 \\
        & & & 3 & 0.0308  & 0.0015  & -0.0151  & 0.0016 \\
        \bottomrule
    \end{tabular}
\end{table}

Table~\ref{table::level set simulation} demonstrates that at a small sample size of $n=100$, the empirical exceedance set estimator occasionally fails to recover the exact number of disjoint regions across all simulation replicates. However, as the sample size increases to $n=200$ and $n=500$, the estimator consistently achieves perfect structural recovery, identifying the correct number of regions in every replicate. Furthermore, both the empirical bias and the MSE of the boundary estimates exhibit a decreasing trend as $n$ grows.

\section{Real data analysis}\label{sect::real data}

In this section, we study the structural discrepancy between weekday and weekend mobility patterns in the New York City Citi Bike sharing system. The system provides a publicly available historical trip dataset (\url{https://citibikenyc.com/system-data}) which contains the start and end times, as well as the spatial coordinates, of all trips between bike stations across the city at a second-level resolution. We analyze trip records from January 2019 through December 2019 to investigate the intra-day dynamic patterns of urban mobility. We restrict our analysis to the 94 most frequently utilized stations and partition each day into 20-minute intervals, yielding 72 data points per day. For each interval, we construct a network where the 94 nodes correspond to the selected stations, and the edge weights represent the aggregate number of trips between them. At each time point, the network structure is encoded by a $94 \times 94$ graph Laplacian matrix. Specifically, the graph Laplacian is defined as $L = D - A$, where $A$ is the $94 \times 94$ adjacency matrix whose entry $a_{ij}$ denotes the trip count between nodes $i$ and $j$, and $D$ is the diagonal degree matrix with elements $d_{ii} = \sum_{j=1}^{94} a_{ij}$.

To evaluate the dissimilarity between weekday and weekend transportation patterns, we stratify the dataset accordingly. For both subsets, the graph Laplacian serves as the object-valued response, while the covariate takes values in the normalized grid $x \in \{1/72, 2/72, \dots, 1\}$, representing the time index within the day. Figure~\ref{fig::2} illustrates the difference of the trip volumes of the weekday and weekend transportation networks. The weekday mobility pattern exhibits a bimodal structure with local maxima around 9 am and 5 pm, corresponding to traditional commuter rush hours. Conversely, the weekend trajectory lacks these commuting spikes, instead forming an extended, sustained plateau between 11 am and 7 pm.

\begin{figure}[!b]
\centering
\includegraphics[width=1\linewidth]{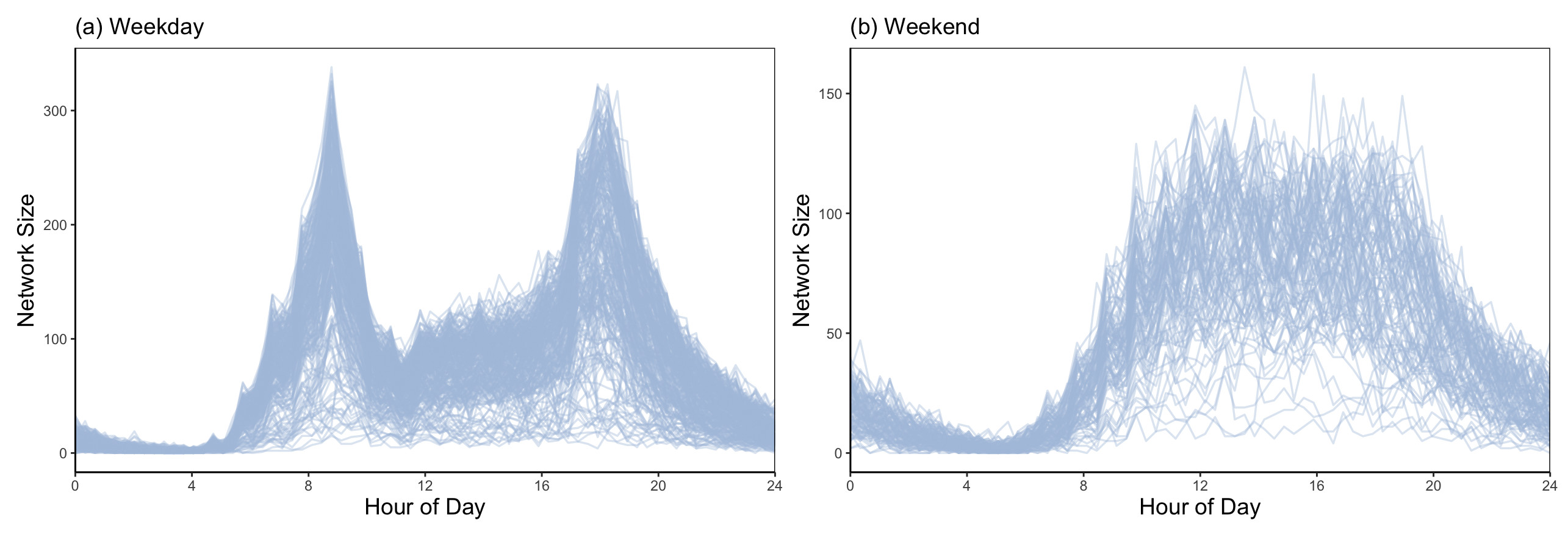} 
\caption{Network sizes of (a) weekday and (b) weekend transportation networks in the New York City Citi Bike sharing system.}\label{fig::2}
\end{figure}

To accommodate the extensive sample sizes and achieve robustness, we apply the MoM framework in Model~\ref{exmp::MoM F regression} in conjunction with the local Fr\'{e}chet regression estimator in Model~\ref{exmp::local F regression}. The block numbers for the weekday dataset and the weekend dataset are set to be 80 and 40, respectively. Utilizing the Frobenius metric $d_{\rm F}$, we estimate the conditional Fr\'{e}chet means evaluated at time index $x$ for the weekday and weekend datasets, denoted by $\widehat{m}_{\rm wd}(x)$ and $\widehat{m}_{\rm we}(x)$, respectively. A primary objective of this analysis is to identify specific periods during the day when the weekday and weekend transportation networks exhibit substantial structural divergence. Consequently, the quantity of interest formulated in Section~\ref{sect::level set} is defined as $\widehat{\Delta}(x) = d_{\rm F}\big(\widehat{m}_{\rm wd}(x), \widehat{m}_{\rm we}(x)\big)$.

\begin{figure}[!tb]
\centering
\includegraphics[width=1\linewidth]{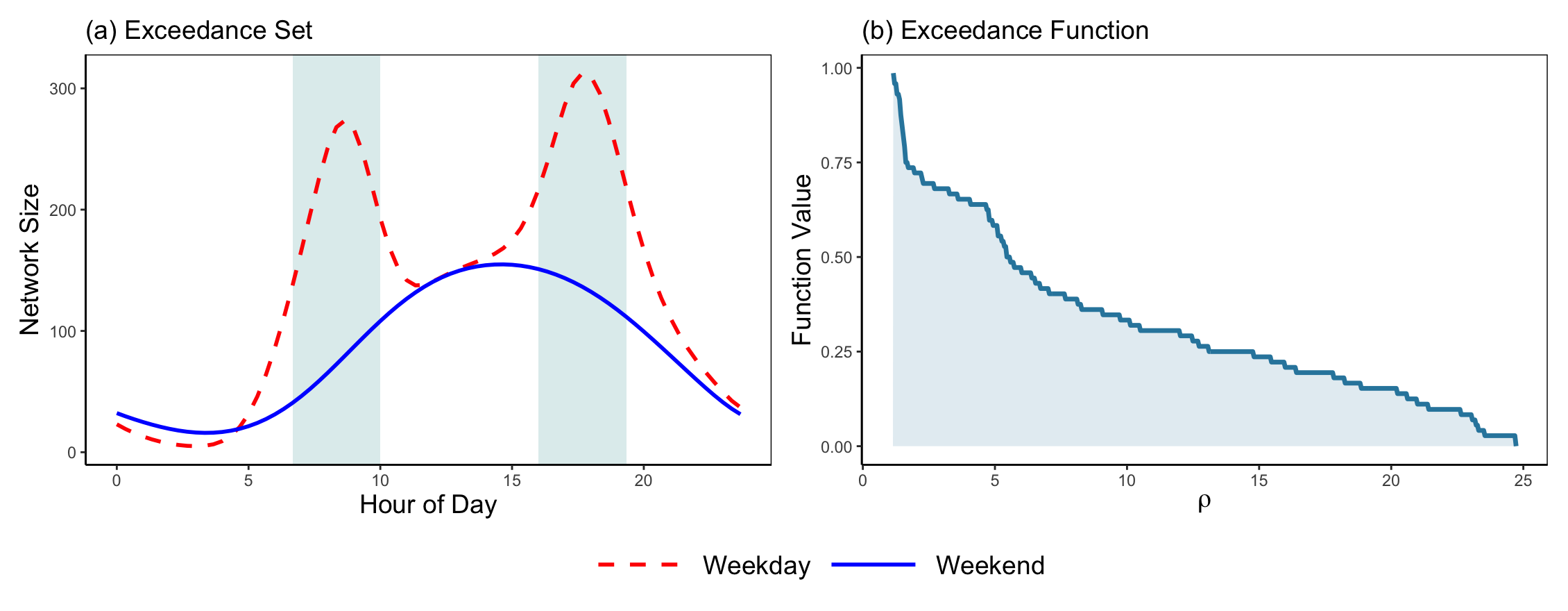} 
\caption{(a) Network sizes and exceedance regions for weekday and weekend transportation networks in the New York City Citi Bike sharing system. Dashed red  and solid blue lines represent the graph Laplacians for weekday and weekend, respectively. The exceedance regions are shaded in azure. (b) The exceedance function for the difference between weekday and weekend transportation networks.}\label{fig::3}
\end{figure}

Figure~\ref{fig::3} illustrates the application of the proposed exceedance set methodology to the New York City Citi Bike system. Panel (a) displays the estimated exceedance set, indicated by the shaded intervals, where the empirical structural discrepancy $\widehat{\Delta}(x)$ surpasses the 70th percentile of $d_{\rm F}\big(\widehat{m}_{\rm wd}(x), \widehat{m}_{\rm we}(x)\big)$. Notably, this exceedance set naturally encompasses the morning and evening commuter rush hours in weekdays, successfully isolating the specific temporal windows where weekday and weekend network dynamics fundamentally diverge. Panel (b) of Figure~\ref{fig::3} presents the corresponding empirical exceedance function, $\widehat{\ell}(\rho)$. Evaluating this function at the aforementioned 70th percentile threshold yields a function value around $0.26$. This mathematically reinforces the visual evidence in Panel (a), explicitly quantifying that these highly divergent commuting periods constitute over one-quarter of the total daily cycle.

\section{Conclusion}

In this paper, we advance the theoretical and methodological foundations of object oriented data analysis by addressing a critical analytical bottleneck in the study of generalized conditional Fr\'{e}chet mean set. Recognizing the profound difficulty of verifying asymptotic uniform equicontinuity in general metric spaces devoid of linear algebraic structure, we introduce a novel theoretical framework to establish uniform convergence. By shifting the analytical focus to an analytical condition imposed directly on the empirical cost function of the generalized conditional Fr\'{e}chet means, our approach bypasses the need for traditional equicontinuity verification for generalized conditional Fr\'{e}chet means.

Building upon the novel foundational guarantees for uniform convergence presented in this paper, we develop two downstream methodologies. First, we propose a weighted Fr\'{e}chet aggregation framework covering both distributed regression for computationally demanding, large-scale datasets and robust median-of-means estimation for contaminated data. Importantly, our analysis distinguishes itself from the existing literature by successfully establishing these properties without relying on restrictive smoothness assumptions. Second, we extend classical concepts of level set estimation to non-Euclidean responses, providing a framework for exceedance set estimation. This methodology empowers practitioners not only to geometrically isolate regions of extreme or anomalous behavior within the covariate space, but also to rigorously quantify the aggregate magnitude of these structural threshold violations. The empirical performance of these proposed methods is validated through comprehensive Monte Carlo simulations and real data analysis of New York City Citi Bike mobility patterns, underscoring their practical utility. 

\section*{Supplementary Material}

The supplementary material includes the summary of notations, proofs and auxiliary lemmas, and the data generation setting for simulation studies.

\clearpage

\begin{frontmatter}
\title{SUPPLEMENTARY MATERIAL for Uniform Convergence of Generalized Conditional Fr\'{e}chet Means, with Applications to Fr\'{e}chet Moment-of-Mean Regression}

\thankstext{t1}{Equal contribution, ordered alphabetically.}
\begin{aug}
\author[A]{\fnms{Houren}~\snm{Hong}\thanksref{t1}}
\author[B]{\fnms{Jiazhen}~\snm{Xu}\thanksref{t1}}
\and
\author[A]{\fnms{Andrew} ~\snm{T.A.}~\snm{Wood}}
\address[A]{Research School of Finance, Actuarial Studies and Statistics, Australian National University\printead[presep={,\ }]{e1,e3}}

\address[B]{Department of Actuarial Studies and Business Analytics, Macquarie University\printead[presep={,\ }]{e2}}
\end{aug}
\end{frontmatter}
\setcounter{section}{0}
\setcounter{equation}{0}

\makeatletter
\renewcommand\thesection{S\@arabic\c@section}
\renewcommand\thetable{S\@arabic\c@table}
\renewcommand \thefigure{S\@arabic\c@figure}
\renewcommand \theequation{S\arabic{equation}}
\makeatother

This Supplementary Material consists of six parts. Section~\ref{supp sect::notation} presents a summary of key notations. Proof of results in Sections~\ref{sect::uniform conv}, \ref{sect::regression uniform conv}, \ref{sect::MoM}, and \ref{sect::level set} are given in Sections~\ref{supp sect::proof1},\ref{supp sect::proof2}, \ref{supp sect::proof3}, and \ref{supp sect::proof4}, respectively. Section~\ref{supp sect::data generation} covers the data generation setting used in the simulation studies.

\section{Notations}\label{supp sect::notation}

Table~\ref{tab:notations} provides a summary of key notations used throughout the paper.

\renewcommand{\arraystretch}{1.3} 
\begin{xltabular}{\linewidth}{l X}
    
    \caption{Summary of Key Notations}
    \label{tab:notations} \\
    \toprule
    \textbf{Notation}  & \textbf{Description} \\
    \midrule
    \endfirsthead

    \caption[]{Summary of Key Notations (Continued)} \\
    \toprule
    \textbf{Notation}  & \textbf{Description} \\
    \midrule
    \endhead

    \midrule
    \multicolumn{2}{r}{\textit{\small Continued on next page...}} \\
    \endfoot

    \bottomrule
    \endlastfoot

    $\Theta$& data space $\Theta$ with a suitable metric \\
    $(\Omega,d)$& parameter space $\Omega$ with metric $d$ \\
    $(\mathcal{X},d_\mathcal{X})$ & covariate space $\mathcal{X}$ with metric $d_\mathcal{X}$ \\
    $m(x),M(\nu,x)$ & population generalized conditional Fr\'{e}chet mean set $m(x)$ with the cost function $M(\nu,x)$\\
    $\widehat{m}(x),\widehat{M}(\nu,x)$ & empirical generalized conditional Fr\'{e}chet mean set $\widehat{m}(x)$ with the empirical cost function $\widehat{M}(\nu,x)$\\
    $d(\nu, \mathcal{A}),d(\mathcal{B},\mathcal{A})$ & a metric $d(\nu, \mathcal{A}) = \inf_{a \in \mathcal{A}} d(\nu, a)$ for a non-empty set $\mathcal{A}$ and the one-sided Hausdorff  distance $d (\mathcal{B},\mathcal{A}) =\sup_{b \in \mathcal{B}} d(b, \mathcal{A})$ for non-empty sets $\mathcal{A}$ and $\mathcal{B}$\\
    $\mathcal{M},\mathcal{C}_m,d_\mathcal{C}$ & ambient space $\mathcal{M}$ such that $\Omega\subseteq \mathcal{M}$, a convex subset $\mathcal{C}_m \subseteq \mathcal{M}$ equipped with metric $d_\mathcal{C}$\\
    $\widetilde{m}$ & point used to establish local convexity of $\widehat{M}$ \\
    $\delta(x)$ & ratio of two geodesics used to establish local convexity of $\widehat{M}$ \\
    $d_{\rm E},d_{\rm IE},d_{\rm F}$ & the Euclidean distance $d_{\rm E}$, the induced Euclidean distance $d_{\rm IE}$, and the Frobenius metric $d_{\rm F}$ \\
    $r(x),R(\nu,x)$ & a specific conditional Fr\'{e}chet means $r(x)$ with the cost function $R(\nu,x)$ for Fr\'{e}chet regression models\\
    $\widetilde{r}(x),\widetilde{R}(\nu,x)$ & population Fr\'{e}chet-regression-based means $\widetilde{r}(x)$ with the cost function $\widetilde{R}(\nu,x)$\\
    $\widehat{r}(x),\widehat{R}(\nu,x)$ & empirical Fr\'{e}chet-regression-based means $\widehat{r}(x)$ with the cost function $\widehat{R}(\nu,x)$\\
    $r_b(x)$,$\widehat{r}_b(x)$ & the block-wise conditional Fr\'echet means and their estimators\\
    $\widehat{r}(x,\pmb{\kappa};\ell)$, $r(x,\pmb{\kappa};\ell)$& the weighted Fr\'echet estimator of $\ell$-th order with a weight vector $\pmb{\kappa}$\\
    $\tau$ & the proportion of contaminated blocks\\
    $\zeta$ & distance between any two points in a certain space\\
    $\nu$ & a generic point in $\Omega$  \\
    $w(X,x)$ & weight in Fr\'{e}chet regression models \\
    $\Lambda,\Delta(x)$ & two functions used in the construction of exceedance sets $\Lambda$ and $\Delta(x)$ that quantifies a scalar quantity given object-valued inputs\\
    $L(\rho), \ell(\rho)$ & the population exceedance set $L(\rho)$ and the population exceedance function $\ell(\rho)$\\
    $\widehat{L}(\rho), \widehat{\ell}(\rho)$ & the empirical exceedance set $\widehat{L}(\rho)$ and the empirical exceedance function $\widehat{\ell}(\rho)$\\

\end{xltabular}

\section{Proof of results in Section \ref{sect::uniform conv}}\label{supp sect::proof1}

Throughout the proofs, for statements involving suprema over uncountable sets, probabilities should be understood as outer probabilities to avoid standard measurability issues.


\begin{proof}[Proof of Proposition~\ref{prop::uniform bound from loss to estimator}]

We prove the contrapositive statement of Proposition~\ref{prop::uniform bound from loss to estimator}, i.e. we show that for every fixed $x\in\mathcal{X}$, if there exists an empirical minimizer $\widehat{m} \in \widehat{m}(x)$ such that $d(\widehat{m}, m(x)) \geq \zeta$, then there exists some $\widetilde{m}\neq \widehat{m}$ with $d(\widetilde{m}, m(x)) \geq \zeta$ such that $\widehat{M}(\widetilde{m},x) - \sup_{m \in m(x)} \widehat{M}(m, x)  \leq 0$.

As $m(x)$ is assumed to be closed, there exists $m^* \in m(x)$ such that
\[
d(\widehat{m}, m^*) = d(\widehat{m}, m(x)) \geq \zeta.
\]
Here we choose the $m^*$ such that $d(\widehat{m},m^*)\leq d(\widehat{m},m)$ for any other $m\in m(x)$. By Condition~\ref{condition R0}(ii), $m^*$ and $\widehat{m} $ are contained within a convex subset $\mathcal{C}_{m^*}$ and there exists a $\widetilde{m}\in \mathcal{C}_{m^*}$ such that $d(\widetilde{m},m^*)\geq \zeta$ and thus $d(\widetilde{m}, m(x)) \geq \zeta$. Let $\delta = d_\mathcal{C}(m^*, \widetilde{m}) / d_\mathcal{C}(m^*, \widehat{m})$. By Condition~\ref{condition R0} (iii), we obtain
\[
\widehat{M}(\widetilde{m}, x) \leq \delta \widehat{M}(\widehat{m}, x) + (1-\delta) \widehat{M}(m^*, x),
\]
which gives
\[
\widehat{M}(\widetilde{m}, x) - \widehat{M}(m^*, x) \leq \delta \left\{ \widehat{M}(\widehat{m}, x) - \widehat{M}(m^*, x) \right\} \leq 0
\]
where the last inequality holds as $\widehat{m} \in \widehat{m}(x)$ is a minimizer of the empirical cost function. Finally, because $ \widehat{M}(m^*, x)\leq \sup_{m\in m(x)} \widehat{M}(m, x)$ by definition, it follows that
\[
\widehat{M}(\widetilde{m}, x) - \sup_{m\in m(x)} \widehat{M}(m, x) \leq \widehat{M}(\widetilde{m}, x) - \widehat{M}(m^*, x) \leq 0
\]
This completes the proof.
\end{proof}

\begin{proof}[Proof of Theorem \ref{thm::uniform convergence 1}]
By Proposition~\ref{prop::uniform bound from loss to estimator}, for any fixed $x\in\mathcal{X}$, one always gets
\[
   \left\{  \inf_{\nu \notin B_\zeta(m(x))} \left\{ \widehat{M}(\nu, x)  - \sup_{m\in m(x)} \widehat{M}(m, x) \right\} >0  \right\} \subseteq \left \{ d (\widehat{m}(x),m(x))   \leq \zeta  \right \}.
\]
Then, it is sufficient to show that 
\begin{align}\label{formula::supre subset}
   \left  \{ \inf_{x \in \mathcal{X}} \inf_{\nu \notin B_\zeta(m(x))} \left\{ \widehat{M}(\nu, x)  - \sup_{m\in m(x)} \widehat{M}(m, x) \right\} >0  \right   \} \subseteq \left \{\sup_{x \in \mathcal{X}}d (\widehat{m}(x),m(x)) \leq \zeta \right \}.
\end{align}
To see this, note that (\ref{formula::supre subset}), in conjunction with 
\[
\mathbb{P}\left [\inf_{x \in \mathcal{X}} \inf_{\nu \notin B_\zeta(m(x))} \left\{ \widehat{M}(\nu, x)  - \sup_{m\in m(x)} \widehat{M}(m, x)\right\} >0 \right ] \to 1,
\]
holds for every  $\zeta>0$, will give that 
\[
\sup_{x \in \mathcal{X}} d( \widehat{m}(x), m(x)) =o_p(1).
\]

We now derive (\ref{formula::supre subset}). Let 
\begin{align*}
&\mathcal{A}_x = \left\{ \inf_{\nu \notin B_\zeta(m(x))} \left\{ \widehat{M}(\nu, x)  - \sup_{m\in m(x)} \widehat{M}(m, x) \right\}  > 0 \right\},\\ 
&\mathcal{B}_x = \left\{ d (\widehat{m}(x),m(x)) \leq \zeta \right\}, \\
{\rm and} ~&\mathcal{A} = \left\{ \inf_{x \in \mathcal{X}} \inf_{\nu \notin B_\zeta(m(x))} \left\{ \widehat{M}(\nu, x)  - \sup_{m\in m(x)} \widehat{M}(m, x) \right\} > 0 \right\}.
\end{align*}
First observe that if the event $\mathcal{A} $ occurs, the infimum of the margin over the entire space $\mathcal{X}$ is strictly positive. This implies that the margin evaluated at every fixed $x \in \mathcal{X}$ must also be strictly positive. Therefore, the occurrence of $\mathcal{A}$ guarantees the occurrence of $\mathcal{A}_x$ for all $x \in \mathcal{X}$. This implies that $\mathcal{A} \subseteq \bigcap_{x \in \mathcal{X}} \mathcal{A}_x$.

Recall that by Proposition~\ref{prop::uniform bound from loss to estimator}, $\mathcal{A}_x \subseteq \mathcal{B}_x$ holds for all $x$, then $\bigcap_{x \in \mathcal{X}} \mathcal{A}_x \subseteq \bigcap_{x \in \mathcal{X}} \mathcal{B}_x$. Note that $\bigcap_{x \in \mathcal{X}} \mathcal{B}_x$ represents the scenario where $d(\widehat{m}(x), m(x)) \leq \zeta$ holds simultaneously for all $x \in \mathcal{X}$. If the distance is bounded by $\zeta$ everywhere in $\mathcal{X}$, it follows that the supremum of the distance over $\mathcal{X}$ is also bounded by $\zeta$. This implies $\bigcap_{x \in \mathcal{X}} \mathcal{B}_x \subseteq \left\{ \sup_{x \in \mathcal{X}} d(\widehat{m}(x), m(x)) \leq \zeta \right\}$, which indicates (\ref{formula::supre subset}) holds. This completes the proof.
\end{proof}

\begin{proof}[Proof of Theorem \ref{thm::uniform convergence 2}]
By Theorem \ref{thm::uniform convergence 1}, it is sufficient to verify condition (\ref{improved sufficient condition}), that is,
\[
\mathbb{P}\left [\inf_{x \in \mathcal{X}} \inf_{\nu \notin B_\zeta(m(x))} \left\{ \widehat{M}(\nu, x)  -  \sup_{m\in m(x)} \widehat{M}(m, x)  \right\} >0 \right ] \to 1,
\]
for any $\zeta>0$ as $n \to \infty$. Therefore, we next show that condition (\ref{improved sufficient condition}) follows from Condition~\ref{condition R1} 
and (\ref{improved sufficient condition 2}) in Theorem \ref{thm::uniform convergence 1}. In particular, by Condition~\ref{condition R1}, for given $\zeta>0$, there exists a $\beth=\beth(\zeta)>0$ such that
\[
\beth =  \inf_{x\in\mathcal{X}}\inf_{\nu\in\Omega:d(\nu,m(x))>\zeta} \left\{M(\nu,x)- \sup_{m\in m(x)}M(m,x)\right\}.,
\]
while (\ref{improved sufficient condition 2}) states that
\[
\sup_{x\in\mathcal{X}}\sup_{\nu\in\mathcal{M}} \left| \widehat{M}(\nu, x) -M(\nu,x) \right| =o_p(1) .
\]

Let
\[
\widetilde{M}(\nu, x)=\widehat{M}(\nu, x) - \sup_{m\in m(x)}\widehat{M}(m, x)-M(\nu,x)+ \sup_{m\in m(x)}M(m,x).
\]
From the definitions of $\widehat{M}(\nu, x)$ and $\widetilde{M}(\nu, x)$, we have
\[
\widehat{M}(\nu, x)- \widehat{M}(m(x), x) = \widetilde{M}(\nu, x) +  M(\nu,x)- \sup_{m\in m(x)}M(m,x).
\]
Therefore, 
\begin{align*}
&\inf_{x \in \mathcal{X}} \inf_{\nu \notin B_\zeta(m(x))} \left\{\widehat{M}(\nu, x)  - \sup_{m\in m(x)}\widehat{M}(m, x) \right\} \\
\geq & \inf_{x \in \mathcal{X}} \inf_{\nu \notin B_\zeta(m(x))}  \widetilde{M}(\nu, x) +  \inf_{x \in \mathcal{X}} \inf_{\nu \notin B_\zeta(m(x))} \left\{M(\nu,x)- \sup_{m\in m(x)} M(m,x)\right\} \\
\geq & \beth- \sup_{x \in \mathcal{X}} \sup_{\nu \in \mathcal{M}} | \widetilde{M}(\nu, x) |.
\end{align*}
\begin{align*}
    \sup_{x \in \mathcal{X}} \sup_{\nu \in \mathcal{M}} \left| \widetilde{M}(\nu, x) \right| & \leq \sup_{x\in\mathcal{X}}\sup_{\nu\in\mathcal{M}} \left\{ \left| \widehat{M}(\nu, x) -M(\nu,x) \right| + \left|   \sup_{m\in m(x)}\widehat{M}(m, x)-  \sup_{m\in m(x)} M(m,x)\right|\right\}\\
    \\
    &\leq \sup_{x\in\mathcal{X}}\sup_{\nu\in\mathcal{M}} \left| \widehat{M}(\nu, x) -M(\nu,x) \right| + \sup_{x\in\mathcal{X}}\sup_{m\in m(x)} \left| \widehat{M}(m, x) -M(m,x) \right|\\
    &\leq 2 \sup_{x\in\mathcal{X}}\sup_{\nu\in\mathcal{M}} \left| \widehat{M}(\nu, x) -M(\nu,x) \right|.
\end{align*}
It then follows from (\ref{improved sufficient condition 2}) that $\sup_{x \in \mathcal{X}} \sup_{\nu \in \mathcal{M}} \left| \widetilde{M}(\nu, x) \right| = o_p(1)$, and therefore (\ref{improved sufficient condition}) can be verified.
\end{proof}

\section{Proof of results in Section \ref{sect::regression uniform conv}}\label{supp sect::proof2}

\begin{proof}[Proof of Proposition \ref{prop::convex Euclidean assump R0}]
For the first part (i) of Condition~\ref{condition R0} that for each $x\in\mathcal{X}$, as $\Omega$ be a convex subset of the Euclidean space with the induced Euclidean distance $d_{\rm IE}$, $\widehat{m}(x)$ exists and is further unique, this follows by Proposition 1 in \cite{petersen2019frechet}. Now, choose $\mathcal{C}_m=\Omega$ for each $m\in m(x)$ and $d_\mathcal{C}$ to be the induced Euclidean distance, which shows the second part (ii) in Condition~\ref{condition R0} holds as $\Omega$ is a convex subset of the Euclidean space . 

We then aim to verify the third part (iii) of Condition~\ref{condition R0} via establishing the global convexity condition on the empirical cost function $\widehat{M}(\nu,x)$. Recall that we now have $\widehat{M}(\nu,x)=\widehat{R}(\nu,x)$. Denote $\|\cdot\|_{\rm IE}$ be the induced Euclidean norm and $\langle \cdot,\cdot\rangle_{\rm IE}$ be the standard inner product on the ambient space, one can see that 
\begin{align*}
    \widehat{M}(\nu,x) =& \widehat{R}(\nu;x) = n^{-1}\sum_{i=1}^n\widehat{w}(X_i,x)\|\nu-Y_i\|_{\rm IE}^2\\
    =&\left[ n^{-1}\sum_{i=1}^n\widehat{w}(X_i,x) \right] \|\nu\|_{\rm IE}^2 - 2\left\langle\nu,n^{-1}\sum_{i=1}^n\widehat{w}(X_i,x)Y_i\right\rangle_{\rm IE} \\
    &+ n^{-1}\sum_{i=1}^n\widehat{w}(X_i,x)\|Y_i\|_{\rm IE}^2 .
\end{align*}
Let $\bar{w}_1=n^{-1}\sum_{i=1}^n\widehat{w}(X_i,x)$ and $\bar{w}_2=n^{-1}\sum_{i=1}^n\widehat{w}(X_i,x)Y_i$. Note that for Models~\ref{exmp::global F regression}--\ref{exmp::deep F regression}, one always gets $\bar{w}_1=1$, which, together with the above result, implies that
\begin{align*}
     \widehat{M}(\nu,x) =& \|\nu\|_{\rm IE}^2 - 2\langle\nu,\bar{w}_2\rangle_{\rm IE} + C\\
     = & \|\nu-\bar{w}_2+\bar{w}_2\|_{\rm IE}^2 - 2\langle\nu,\bar{w}_2\rangle_{\rm IE} + C\\
     = & \|\nu-\bar{w}_2\|_{\rm IE}^2 + \|\bar{w}_2\|_{\rm IE}^2 + C
\end{align*}
where $C$ is constant not depending on $\nu$. Thus one can conclude that $\widehat{M}(\nu,x)$ is convex in $\nu$. This gives that $\widehat{m}(x)$ exists and the third part of Condition~\ref{condition R0} holds for Models~\ref{exmp::global F regression}--\ref{exmp::deep F regression}. This completes the proof that Condition~\ref{condition R0} holds for $(\Omega,d_{\rm IE})$.
\end{proof}

\begin{proof}[Proof of Proposition \ref{prop::Riemannian space assump R0}]
    Before presenting the proof, we first introduce some notions. The global injectivity radius of a Riemannian manifold $\Omega$, ${\rm inj}(\Omega)$, is defined as the infimum of the local injectivity radii over all points in $\Omega$ as ${\rm inj}(\Omega) = \inf_{p \in \Omega} {\rm inj}(p)$. Here, ${\rm inj}(p)$ is the local injectivity radius at a point $p$ on $\Omega$ and is the largest radius $\zeta>0$ for which the exponential map ${\rm Exp}_p : B_\zeta(p) \to \Omega$ is a diffeomorphism.

    Note that a compact Riemannian manifold is a convex metric space and thus the second part (ii) of Condition~\ref{condition R0} holds by choosing $\mathcal{C}_m=\Omega$ for each $m\in m(x)$ and $d_\mathcal{C}=d$. Recall that the distance $d(\cdot, \cdot)$ is continuous, and thus the empirical cost function $\widehat{M}(\cdot,x)$ is also continuous for every $x\in\mathcal{X}$. By the extreme value theorem, a continuous function on a compact space is bounded and guaranteed to attain its global minimum. This gives that the first part (i) of Condition~\ref{condition R0} holds.

    We then aim to derive the third part (iii) of Condition~\ref{condition R0}, which can be verified by showing the differentiability of $\widehat{M}(\nu,x)$ around each $m\in m(x)$. To see this, note that the differentiability of $\widehat{M}(\nu,x)$ around each $m$ ensures that $\widehat{M}(\nu,x)$ is not strictly geodesic concave around $m$. This implies that, for each $m\in m(x)$, one can pick the $\widehat{m}^*$ such that $0<d(\widehat{m}^*,m)\leq d_(\widehat{m},m)$ for any other $\widehat{m}\in \widehat{m}(x)$, there exists $\widetilde{m}\notin\{m,\widehat{m}^*\}$ such that $d_\mathcal{C}(m,\widetilde{m})+d_\mathcal{C}(\widetilde{m},\widehat{m}^*)=d_\mathcal{C}(m,\widehat{m}^*)$ and 
    \[
    \widehat{M}(\widetilde{m},x) \leq \delta(x) \widehat{M}(\widehat{m}^*,x) + \{1-\delta(x)\} \widehat{M}(m,x)
    \]
    holds for all $x\in\mathcal{X}$ where $\delta(x)=d(m,\widetilde{m})/d(m,\widehat{m}^*)$. The sufficient condition of the differentiability of $\widehat{M}(\nu,x)$ around each $m\in m(x)$ is that none of $Y_1,Y_2,\ldots,Y_n$ and $\widehat{m}(x)$ will meet at the cut locus. Observe that when the support of $Y_1,Y_2,\ldots,Y_n$ is contained within a ball $B(p, r) \subset \Omega$ with $r < 2^{-1}{\rm inj}(\Omega)$, together with the positive weights $\widehat{w}(X_i,x)$ for any $i=1,2,\ldots,n$ and $x\in\mathcal{X}$, gives that each $\widehat{m}\in \widehat{m}(x)$ also locates within the ball $B(p, r)$. Thus one can conclude that Models~\ref{exmp::kernel F regression}--\ref{exmp::deep F regression} satisfy the third part (iii) of Condition~\ref{condition R0}.
    
    We then evaluate the situation where some of the weights can be negative. To ensure that each $m\in m(x)$ avoids the cut loci of $Y_1,Y_2,\ldots,Y_n$ for any $x\in\mathcal{X}$, one needs 
    \begin{align}\label{ineq::prop::Riemannian space 1}
    \widehat{M}(q,x)>\widehat{M}(p,x)    
    \end{align} 
    for any $d(p,q)={\rm inj}(\Omega)$. Note that by the triangle inequality, we obtain $d(Y_i,p)\leq r$, $d(Y_i,q)\geq {\rm inj}(\Omega)-r$, and $d(Y_i,q)\leq {\rm inj}(\Omega)+r$ for every $i=1,2,\ldots,n$. This, together with~(\ref{ineq::prop::Riemannian space 1}), gives 
    \[
    W^+ ({\rm inj}(\Omega)-r)^2 - W^-({\rm inj}(\Omega)+r)^2 > W^+r^2,
    \]
    where $W^+=\sum_{i\in\mathcal{I}^+}\widehat{w}(X_i,x)$ and $W^-=\sum_{i\notin\mathcal{I}^+}|\widehat{w}(X_i,x)|$ with $\mathcal{I}^+=\{i\in\{1,2,\ldots,n\}:\widehat{w}(X_i,x)\geq 0\}$. The above result holds when $W^-<\left[\left\{ {\rm inj}(\Omega)\right\}^2-2r{\rm inj}(\Omega)\right]/(4r{\rm inj}(\Omega)+r^2)$. This completes the proof.

\end{proof}

\begin{proof}[Proof of Proposition \ref{prop::Riemannian space extrinsic assump R0}]
    Following the the proof of Proposition~\ref{prop::Riemannian space assump R0}, a compact Riemannian manifold is a convex metric space with the corresponding intrinsic metric $d_\Omega$ and thus the second part (ii) of Condition~\ref{condition R0} holds by choosing $d_\mathcal{C}=d_\Omega$. Moreover, for $d=d_{\rm IE}$, the empirical cost function $\widehat{M}(\cdot,x)$ is also continuous for every $x\in\mathcal{X}$. Then applying the same techniques to those used in the proof of Proposition~\ref{prop::Riemannian space assump R0}, the firth part (i) of Condition~\ref{condition R0} holds. The third part (iii) of Condition~\ref{condition R0} is verified by borrowing the similar arguments to those used in the proof of Proposition~\ref{prop::convex Euclidean assump R0}. This completes the proof.  
\end{proof}

\begin{proof}[Proof of Proposition \ref{prop::Hadamard space assump R0}]
    First note that a Hadamard space is a convex metric space and thus the convex subset requirement holds in the second part (ii) of Condition~\ref{condition R0} by choosing $\mathcal{C}_m=\Omega$ for each $m\in m(x)$ and $d_\mathcal{C}=d$. The existence of the sample generalized Fr\'{e}chet mean will be verified together with the global convexity of the cost function below.

    Let $\nu_1,\nu_2,\nu_3\in\Omega$ and let $\gamma:[0,1]\to \Omega$ denote the geodesic between $\nu_1$ and $\nu_2$. As $(\Omega,d)$ is a Hadamard space, the following inequality holds
    \[
    d^2(\gamma(t),\nu_3)\leq t d^2(\nu_2,\nu_3) + (1-t) d^2(\nu_1,\nu_3) - t(1-t) d^2(\nu_1,\nu_2),
    \]
    for $t\in[0,1]$. This indicates that, for each $m\in m(x)$ and each $\widehat{m}\in \widehat{m}(x)$, we always get 
    \begin{align}\label{ineq::cat 0 ineq}
        d^2(\widetilde{m},Y_i)\leq \delta(x) d^2(\widehat{m},Y_i) + \{1-\delta(x)\} d^2(m,Y_i) - t(1-t) d^2(m,\widehat{m}),
    \end{align}
    for $i=1,2,\ldots,n$ by borrowing the notations in Condition~\ref{condition R0}. Recall that for Models~\ref{exmp::kernel F regression}--\ref{exmp::deep F regression}, one can always write $\widehat{M}(\nu,x)=n^{-1}\sum_{i=1}^n \widehat{w}_i(x) d^2(\nu,Y_i)$ with $\sum_{i=1}^n \widehat{w}_i(x)=1$. Thus, to verify that
    \[
    \widehat{M}(\widetilde{m},x) \leq \delta(x) \widehat{M}(\widehat{m},x) + \{1-\delta(x)\} \widehat{M}(m,x)
    \]
    in Condition~\ref{condition R0} holds, it is sufficient to show that
    \[
    \sum_{i=1}^n \widehat{w}_i(x) d^2(\widetilde{m},Y_i) \leq \delta(x)\sum_{i=1}^n \widehat{w}_i(x) d^2(\widehat{m},Y_i) + \{1-\delta(x)\}\sum_{i=1}^n \widehat{w}_i(x)d^2(m,Y_i) .
    \]
    Observe that when $d(\widehat{m},m)=0$, the inequality above holds in equality. Thus we only need to verify the inequality above by considering $d(\widehat{m},m)>0$. Let 
    \[
    r_i(x)=\delta(x)d^2(\widehat{m},Y_i)+\{1-\delta(x)\}d^2(m,Y_i) - d^2(\widetilde{m},Y_i).
    \]
    It is sufficient to show that $\sum_{i=1}^n \widehat{w}_i(x)r_i(x)\geq 0$. Note that by (\ref{ineq::cat 0 ineq}), we obtain
    \begin{align}\label{ineq::ri(x)}
        r_i(x) \geq  t(1-t) d^2(m,\widehat{m})\geq C > 0,
    \end{align}
    for some positive constant $C$ as $\widetilde{m}\notin\{m,\widehat{m}\}$. Note that for Models~\ref{exmp::kernel F regression}--\ref{exmp::deep F regression}, the weights $\widehat{w}_i(x) \geq 0$ and thus the first and the third parts of Condition~\ref{condition R0} hold, which completes the proof.
    
\end{proof}

\begin{proof}[Proof of Proposition \ref{prop::Hadamard manifold assump R0}]
    Note that a Hadamard manifold is a special case of a Hadamard space and is a complete and smooth manifold. Therefore, following arguments analogous to those used in the proofs of Propositions~\ref{prop::Riemannian space assump R0} and \ref{prop::Hadamard space assump R0}, we only need to verify Condition~\ref{condition R0} (iii).

    To this end, we first choose $d_\mathcal{C}$ and define a geodesic $\gamma:[0,1]\to \Omega$ where $\gamma(0)=r(x)$ and $\gamma(0)=\widehat{r}(x)$. Thus, Condition~\ref{condition R0} (iii) reduces to show that there exist a $t\in(0,1)$ such that
    \[
    \widehat{R}(\gamma(t), x) \leq t \widehat{R}(\gamma(1), x) + (1-t) \widehat{R}(\gamma(0), x).
    \]
    Let $\widehat{G}(t,x)=\widehat{R}(\gamma(t), x)$. Then, it is sufficient to show $\widehat{G}(t,x)$ is partially differentiable with respect to $t\in(0,1)$. As a Hadamard manifold is a complete, simply connected Riemannian manifold with non-positive sectional curvature everywhere, the cut locus is empty and the injectivity radius is infinite. Given the special form of $\widehat{G}(t,x)$ which is a weighted sum of the squared distances, $\widehat{G}(t,x)$ is partially differentiable with respect to $t\in(0,1)$. This completes the proof.
\end{proof}

\begin{proof}[Proof of Theorem \ref{prop::global uniform conv}]
    First observe that 
    \begin{align}\label{ineq:: one sided H distance}
    \sup_{x\in\mathcal{X}}d(\widehat{r}(x),r(x)) \leq \sup_{x\in\mathcal{X}}d(\widetilde{r}(x),r(x)) + \sup_{x\in\mathcal{X}}d(\widehat{r}(x),\widetilde{r}(x)).    
    \end{align} 
    To see this, note that by the triangle inequality, we get $\inf_{r\in r(x)} d(\widehat{r},r) \leq d(\widehat{r},\widetilde{r}) + \inf_{r\in r(x)} d(r,\widetilde{r})$, which gives $\inf_{r\in r(x)} d(\widehat{r},r) \leq d(\widehat{r},\widetilde{r}) + \sup_{\widetilde{r}\in \widetilde{r}(x)} \inf_{r\in r(x)} d(r,\widetilde{r})$. Thus one has $\inf_{r\in r(x)} d(\widehat{r},r) \leq  \inf_{\widetilde{r}\in \widetilde{r}(x)}  d(\widehat{r},\widetilde{r}) + \sup_{\widetilde{r}\in \widetilde{r}(x)} \inf_{r\in r(x)} d(r,\widetilde{r})$ via taking the infimum over all $\widetilde{r}\in \widetilde{r}(x)$. This leads to $\sup_{\widehat{r}\in \widehat{r}(x)}\inf_{r\in r(x)} d(\widehat{r},r) \leq  \sup_{\widehat{r}\in \widehat{r}(x)}\inf_{\widetilde{r}\in \widetilde{r}(x)}  d(\widehat{r},\widetilde{r}) + \sup_{\widetilde{r}\in \widetilde{r}(x)} \inf_{r\in r(x)} d(r,\widetilde{r})$, and thus (\ref{ineq:: one sided H distance}) holds.

    We then aim to show that $\sup_{x\in\mathcal{X}}d(\widehat{r}(x),\widetilde{r}(x))=o_p(1)$. Recall that for the global Fr\'{e}chet regression in Model~\ref{exmp::global F regression}, $\widetilde{R}(\nu,x)=\mathbb{E}\left[ w(X,x) d^2(\nu,Y) \right]$ where $w(X,x)=1+(X-\mu)^\top \Sigma^{-1}(x-\mu)$, $\mu=\mathbb{E}(X)$ and $\Sigma=\mathbb{E}[(X-\mu)(X-\mu)^\top]$, while $\widehat{R}(\nu,x)=\sum_{i=1}^n \widehat{w}(X_i,x) d^2(\nu,Y_i)$ where  $\widehat{w}(X_i,x)=n^{-1}+n^{-1}(X_i-\widehat{\mu})^\top \widehat{\Sigma}^{-1} (x-\widehat{\mu})$ with $\widehat{\mu}=n^{-1}\sum_{i=1}^nX_i$ and $\widehat{\Sigma}=n^{-1}\sum_{i=1}^n (X_i-\widehat{\mu})(X_i-\widehat{\mu})^\top$. By Theorem \ref{thm::uniform convergence 2}, it is sufficient to show that
    \[
    \sup_{x\in\mathcal{X}}\sup_{\nu\in\Omega} \left| \widehat{R}(\nu, x) -\widetilde{R}(\nu,x) \right| =o_p(1) .
    \]
    Define 
    \[
    \bar{R}_G(\nu,x)=n^{-1}\sum_{i=1}^n \{1+(X_i-\mu)^\top \Sigma^{-1}(x-\mu)\}d^2(\nu,Y_i).
    \]
    One can see that 
    \begin{align*}
        \sup_{x\in\mathcal{X}}\sup_{\nu\in\Omega} \left| \widehat{R}(\nu, x) -\widetilde{R}(\nu,x) \right|  \leq &\sup_{x\in\mathcal{X}}\sup_{\nu\in\Omega} \left| \widehat{R}(\nu, x) -\bar{R}_G(\nu,x) \right| \\
        &+ \sup_{x\in\mathcal{X}}\sup_{\nu\in\Omega} \left| \bar{R}_G(\nu, x) -\widetilde{R}(\nu,x) \right|.
    \end{align*}
    Thus, in the remainder of this proof, we first show that 
    \begin{align}\label{eq::global uniform conv part 1}
        \sup_{x\in\mathcal{X}}\sup_{\nu\in\Omega} \left| \widehat{R}(\nu, x) -\bar{R}_G(\nu,x) \right|=o_p(1),
    \end{align}
    and then we establish that
    \begin{align}\label{eq::global uniform conv part 2}
        \sup_{x\in\mathcal{X}}\sup_{\nu\in\Omega} \left| \bar{R}_G(\nu, x) -\widetilde{R}(\nu,x) \right|=o_p(1).
    \end{align}

    \textbf{Step I}. Observe that 
    \begin{align*}
        &\left| \widehat{R}(\nu, x) -\bar{R}_G(\nu,x) \right| \\
        =& \left| n^{-1}\sum_{i=1}^n \left\{(X_i-\widehat{\mu})^\top \widehat{\Sigma}^{-1}(x-\widehat{\mu})-(X_i-\mu)^\top \Sigma^{-1}(x-\mu)\right\}d^2(\nu,Y_i) \right|\\
        =& \left|n^{-1}\sum_{i=1}^n A_i(x) d^2(\nu, Y_i)\right|,
    \end{align*}
     where $A_i(x)=(X_i-\widehat{\mu})^\top \widehat{\Sigma}^{-1}(x-\widehat{\mu})-(X_i-\mu)^\top \Sigma^{-1}(x-\mu)$. Rearranging $A_i(x)$ as follows,
     \[
     A_i(x)= (X_i-\mu)^\top \left\{ \widehat{\Sigma}^{-1}(x-\widehat{\mu})- \Sigma^{-1}(x-\mu) \right\} + (\mu-\widehat{\mu})^\top\widehat{\Sigma}^{-1}(x-\widehat{\mu}),
     \]
     we obtain 
     \begin{align}\label{ineq::global uni 1}
         &\sup_{x\in\mathcal{X}}\sup_{\nu\in\Omega} \left| \widehat{R}(\nu, x) -\bar{R}_G(\nu,x) \right|  \nonumber\\
         \leq & \sup_{x\in\mathcal{X}} \| \widehat{\Sigma}^{-1}(x-\widehat{\mu})- \Sigma^{-1}(x-\mu) \|_{\rm E} M_{n,1} \nonumber \\
         &+ \sup_{x\in\mathcal{X}} \|\mu-\widehat{\mu}\|_{\rm E} \|\widehat{\Sigma}^{-1}\|_{\rm op}\| x-\widehat{\mu} \|_{\rm E} M_{n,2},
     \end{align}
     where $\|\cdot\|_{\rm E}$ is the Euclidean norm, $\|\cdot\|_{\rm op}$ is the spectral norm, $M_{n,1}=\sup_{\nu\in\Omega}n^{-1}\sum_{i=1}^n\|X_i-\mu\|_{\rm E}d^2(\nu,Y_i)$, and $M_{n,2}=\sup_{\nu\in\Omega}n^{-1}\sum_{i=1}^nd^2(\nu,Y_i)$. As $\mathcal{X}$ is a compact set, $\widehat{\mu}-\mu=o_p(1)$, and $\widehat{\Sigma}^{-1}-\Sigma^{-1}=o_p(1)$ by the law of large numbers, we obtain that, for some constant $C_1$,
     \begin{align}\label{ineq::global uni 2}
         &\sup_{x\in\mathcal{X}} \| \widehat{\Sigma}^{-1}(x-\widehat{\mu})- \Sigma^{-1}(x-\mu) \|_{\rm E}  \nonumber\\
         =& \sup_{x\in\mathcal{X}} \| (\widehat{\Sigma}^{-1}-\Sigma^{-1})(x-\mu)+ \widehat{\Sigma}^{-1}(\mu-\widehat{\mu}) \|_{\rm E}  \nonumber\\
         \leq & C_1 \| \widehat{\Sigma}^{-1}-\Sigma^{-1} \|_{\rm op}+ \|\widehat{\Sigma}^{-1}\|_{\rm op} \|\mu-\widehat{\mu}\|_{\rm E}  \nonumber\\
         =&o_p(1).
     \end{align}
     Similarly, one can see that 
    \begin{align}\label{ineq::global uni 3}
        \sup_{x\in\mathcal{X}} \|\mu-\widehat{\mu}\|_{\rm E} \|\widehat{\Sigma}^{-1}\|_{\rm op}\| x-\widehat{\mu} \|_{\rm E}=o_p(1).
    \end{align}
    Recall that $\Omega$ is totally bounded, we also have $M_{n,1}=O_p(1)$ and $M_{n,2}=O_p(1)$. This, together with (\ref{ineq::global uni 1}), (\ref{ineq::global uni 2}) and (\ref{ineq::global uni 3}), indicates that (\ref{eq::global uniform conv part 1}) holds.
         
    \textbf{Step II}. It is sufficient to show 
    \begin{align}\label{formula::global conv weak conv}
        \bar{R}_G\rightsquigarrow \widetilde{R}
    \end{align}  
    in $l^{\infty}(\Omega\times \mathcal{X})$, where $\rightsquigarrow$ denotes weak convergence. To see this, note that when (\ref{formula::global conv weak conv}) holds, an application of the continuous mapping theorem (Theorem 1.3.6 in \cite{vdVW96}) gives (\ref{eq::global uniform conv part 2}). 
    
    By Theorem 1.5.4 of \cite{vdVW96}, (\ref{formula::global conv weak conv}) can be established via showing (i) the pointwise convergence, i.e., $\bar{R}_G(\nu, x) -\widetilde{R}(\nu,x)=o_p(1)$ for all $\nu\in\Omega$ and $x\in\mathcal{X}$, and (ii) $\bar{R}_G$ is asymptotically equicontinuity in probability. The pointwise convergence follows by the law of large numbers, noting that $\mathbb{E}\{\bar{R}_G(\nu,x)\}=\widetilde{R}(\nu,x)$. For the asymptotically equicontinuity in probability, we aim to show that for any $\iota>0$, $\nu_1, \nu_2 \in \Omega$ and $x_1, x_2 \in \mathcal{X}$, 
    \begin{align}\label{formula::global conv equi cont}
         \limsup_n \mathbb{P}\left( \sup_{d(\nu_1,\nu_2)<\zeta_1} \sup_{\|x_1-x_2\|_{\rm E}<\zeta_2} \left| \bar{R}_G(\nu_1,x_1)-  \bar{R}_G(\nu_2,x_2) \right| >\iota \right)\to 0,
    \end{align}
    as $\zeta_1,\zeta_2\to0$

    By utilizing similar techniques to those used in (\ref{ineq::global uni 1}), one obtains
    \begin{align*}
    &\vert \bar{R}_G(\nu_1, x_1)  - \bar{R}_G(\nu_2, x_2) \vert  \\
    \leq &\left\vert  n^{-1}\sum_{i=1}^n \left\{ (X_i-\mu)^\top\Sigma^{-1}(x_1-x_2) d^2(\nu_2,Y_i) \right\} \right\vert  \\
    &+ \left\vert  n^{-1}\sum_{i=1}^n \left\{ 1+ (X_i-\mu)^\top\Sigma^{-1}(x_2-\mu)\left\{ d^2(\nu_1,Y_i) - d^2(\nu_2,Y_i) \right\} \right\} \right\vert  \\
    =& O_p\left( \| x_1 - x_2 \|_{\rm E}\right) + O_p\left(d(\nu_1,\nu_2)\right).
    \end{align*}
    This implies that 
    \[
    \sup_{d(\nu_1,\nu_2)<\zeta_1} \sup_{\|x_1-x_2\|_{\rm E}<\zeta_2} \left| \bar{R}_G(\nu_1,x_1)-  \bar{R}_G(\nu_2,x_2) \right| = O_p(\zeta_1) + O_p(\zeta_2),
    \]
    which proves (\ref{formula::global conv equi cont}).

    Combing Step I and Step II verifies the condition (\ref{improved sufficient condition 2}) in Theorem \ref{thm::uniform convergence 2}, and thus we have $\sup_{x\in\mathcal{X}}d(\widehat{r}_G(x),m(x))=o_p(1)$ by Theorem \ref{thm::uniform convergence 2}. This completes the proof.

\end{proof}

\begin{proof}[Proof of Theorem \ref{prop::local uniform conv}]

    Recall that for the local Fr\'{e}chet regression in Model~\ref{exmp::local F regression},
    $\widehat{R}(\nu,x)=\sum_{i=1}^n \widehat{w}(X_i,x) d^2(\nu,Y_i)$ where $\widehat{w}(X_i,x)=n^{-1}\widehat{\sigma}^{-2}K_h(X_i-x)[\widehat{\tau}_2(x)-\widehat{\tau}_1(x)]$ with $\widehat{\tau}_l(x)=n^{-1}\sum_{i=1}^n[K_h(X_i-x)(X_i-x)^l]$ for $l\in\{0,1,2\}$, $\widehat{\sigma}^2=\widehat{\tau}_0(x)\widehat{\tau}_2(x)-\widehat{\tau}_1^2(x)$. By Theorem \ref{thm::uniform convergence 2}, it is sufficient to show that
    \[
    \sup_{x\in\mathcal{X}}\sup_{\nu\in\Omega} \left| \widehat{R}(\nu, x) -R(\nu,x) \right| =o_p(1) .
    \]
    Observe that 
    \begin{align*}
        \sup_{x\in\mathcal{X}}\sup_{\nu\in\Omega} \left| \widehat{R}(\nu, x) -R(\nu,x) \right|  \leq &\sup_{x\in\mathcal{X}}\sup_{\nu\in\Omega} \left| \widehat{R}(\nu, x) -\widetilde{R}_L(\nu,x) \right| \\
        &+ \sup_{x\in\mathcal{X}}\sup_{\nu\in\Omega} \left| \widetilde{R}_L(\nu, x) -R(\nu,x) \right|,
    \end{align*}
    where $\widetilde{R}_L(\nu, x)=\mathbb{E}[w(X,x)d^2(\nu,Y)]$ with $w(X,x)=\sigma^{-2}K_h(X-x)[\tau_2(x)-\tau_1(x)]$, $\tau_l(x)=\mathbb{E}[K_h(X-x)(X-x)^l]$ for $l\in\{0,1,2\}$, $\sigma^2=\tau_0(x)\tau_2(x)-\tau_1^2(x)$.
    By (S.13) from \cite{chen2022uniform}, we have
    \[
    \sup_{x\in\mathcal{X}}\sup_{\nu\in\Omega} \left| \widetilde{R}_L(\nu, x) -R(\nu,x) \right| = O(h^2)
    \]
    as $h\to 0$, which implies that 
     \[
    \sup_{x\in\mathcal{X}}\sup_{\nu\in\Omega} \left| \widetilde{R}_L(\nu, x) -R(\nu,x) \right| = o(1).
    \]
    Thus it is sufficient to show 
    \begin{align*}
        \sup_{x\in\mathcal{X}}\sup_{\nu\in\Omega} \left| \widehat{R}(\nu, x) -\widetilde{R}_L(\nu,x) \right|=o_p(1)
    \end{align*}
    to establish the uniform convergence result. Following similar arguments to those used in Step II in the proof of Theorem \ref{prop::global uniform conv}, we only need to show (i) the pointwise convergence such as $\widetilde{R}(\nu, x) -\widetilde{R}_L(\nu,x)=o_p(1)$ for all $\nu\in\Omega$ and $x\in\mathcal{X}$, and (ii) $\widehat{R}-\widetilde{R}_L$ is asymptotically equicontinuity in probability. Following the proof of Lemma 2 in \cite{petersen2019frechet}, one has $\widehat{R}(\nu, x) -\widetilde{R}_L(\nu,x)=o_p(1)$ for all $\nu\in\Omega$ and $x\in\mathcal{X}$.
    
    For the asymptotically equicontinuity in probability, we aim to show that for any $\iota>0$, 
    \begin{align}\label{formula::local conv equi cont}
         \limsup_n \mathbb{P}\Bigg( \sup_{d(\nu_1,\nu_2)<\zeta_1} \sup_{|x_1-x_2|<\zeta_2} \Big| &\widehat{R}(\nu_1,x_1 )- \widetilde{R}_L(\nu_1,x_1)\nonumber\\
         &-  \widehat{R}(\nu_2,x_2) + \widetilde{R}_L(\nu_2,x_2)\Big| >\iota \Bigg)\to0,
    \end{align}
    as $\zeta_1,\zeta_2\to0$. Observe that 
    \begin{align*}
        &\sup_{d(\nu_1,\nu_2)<\zeta_1} \sup_{|x_1-x_2|<\zeta_2} \left| \widehat{R}(\nu_1,x_1 )- \widetilde{R}_L(\nu_1,x_1)-  \widehat{R}(\nu_2,x_2) + \widetilde{R}_L(\nu_2,x_2)\right| \\
        \leq & \sup_{d(\nu_1,\nu_2)<\zeta_1} \sup_{|x_1-x_2|<\zeta_2} \left| \widehat{R}(\nu_1,x_1 )- \widetilde{R}_L(\nu_1,x_1)-  \widehat{R}(\nu_1,x_2) + \widetilde{R}_L(\nu_1,x_2)\right| \\
        &+ \sup_{d(\nu_1,\nu_2)<\zeta_1} \sup_{|x_1-x_2|<\zeta_2} \left| \widehat{R}(\nu_1,x_2 )- \widetilde{R}_L(\nu_1,x_2)-  \widehat{R}(\nu_2,x_2) + \widetilde{R}_L(\nu_2,x_2)\right| \\
        \leq & \sup_{\nu\in\Omega} \sup_{|x_1-x_2|<\zeta_2} \left| \widehat{R}(\nu,x_1 )- \widetilde{R}_L(\nu,x_1)-  \widehat{R}(\nu,x_2) + \widetilde{R}_L(\nu,x_2)\right| \\
        &+ \sup_{d(\nu_1,\nu_2)<\zeta_1} \sup_{x\in\mathcal{X}} \left| \widehat{R}(\nu_1,x )- \widetilde{R}_L(\nu_1,x)-  \widehat{R}(\nu_2,x) + \widetilde{R}_L(\nu_2,x)\right| .
    \end{align*}
    This indicates that, to verify (\ref{formula::local conv equi cont}), it is sufficient to show for any $\iota>0$,
    \begin{align}\label{formula::local conv equi cont sub1}
         \limsup_n \mathbb{P}\Bigg( \sup_{\nu\in\Omega} \sup_{|x_1-x_2|<\zeta} \Big| &\widehat{R}(\nu,x_1 ) -  \widehat{R}(\nu,x_2) \Big| >\iota \Bigg)\to0,
    \end{align}
    \begin{align}\label{formula::local conv equi cont sub2}
         \limsup_{h\to0}   \sup_{\nu\in\Omega} \sup_{|x_1-x_2|<\zeta} \Big| &\widetilde{R}_L(\nu,x_1 ) -  \widetilde{R}_L(\nu,x_2) \Big| >\iota \to0,
    \end{align}
    and
    \begin{align}\label{formula::local conv equi cont sub3}
         \limsup_n \mathbb{P}\Bigg( \sup_{d(\nu_1,\nu_2)<\zeta} \sup_{x\in\mathcal{X}} \Big| &\widehat{R}(\nu_1,x )- \widetilde{R}_L(\nu_1,x)\nonumber\\
         &-  \widehat{R}(\nu_2,x) + \widetilde{R}_L(\nu_2,x)\Big| >\iota \Bigg)\to0,
    \end{align}
    as $\zeta\to0$. By (S.7) and (S.8) in \cite{chen2022uniform}, (\ref{formula::local conv equi cont sub1}) and (\ref{formula::local conv equi cont sub2}) hold. Therefore, we only need to verify (\ref{formula::local conv equi cont sub3}).

    Note that, as $\mathcal{X}$ is compact and $\Omega$ is totally bounded, 
    \begin{align*}
        &\left| \widehat{R}(\nu_1,x )- \widetilde{R}_L(\nu_1,x)-  \widehat{R}(\nu_2,x) + \widetilde{R}_L(\nu_2,x)\right|\\
        \leq & \left| \widehat{R}(\nu_1,x )-  \widehat{R}(\nu_2,x) \right| + \left| \widetilde{R}_L(\nu_1,x )-  \widetilde{R}_L(\nu_2,x) \right|\\
        = & O_p(d(\nu_1,\nu_2))+ O(d(\nu_1,\nu_2)),
    \end{align*}
    where the last equality holds by utilizing similar techniques to those used in the proof of Lemma 2 in \cite{petersen2019frechet}. This indicates that (\ref{formula::local conv equi cont sub3}) holds and thus we have $\sup_{x\in\mathcal{X}}d(\widehat{r}_L(x),m(x))=o_p(1)$ by Theorem \ref{thm::uniform convergence 2}. This completes the proof.

\end{proof}

\begin{proof}[Proof of Theorem \ref{thm::kernel uniform conv}]
    Recall that for the kernel Fr\'{e}chet regression in Model~\ref{exmp::kernel F regression}, $\widehat{R}(\nu,x)=\sum_{i=1}^n \widehat{w}(X_i,x) d^2(\nu,Y_i)$ where $\widehat{w}(X_i,x)=K_h(X_i-x)/\sum_{j=1}^nK_h(X_j-x)$. By Theorem \ref{thm::uniform convergence 2}, it is sufficient to show that
    \[
    \sup_{x\in\mathcal{X}}\sup_{\nu\in\Omega} \left| \widehat{R}(\nu, x) -R(\nu,x) \right| =o_p(1) .
    \]
    Following the arguments used in the proof of Proposition~\ref{thm::kernel uniform conv}, it is sufficient to show that
    \begin{align}\label{eq::kernel thm uniform 1}
        \sup_{x\in\mathcal{X}}\sup_{\nu\in\Omega} \left| \widehat{R}(\nu, x) -\widetilde{R}_{NW}(\nu,x) \right| =o_p(1), 
    \end{align}
    and
    \begin{align}\label{eq::kernel thm uniform 2}
        \sup_{x\in\mathcal{X}}\sup_{\nu\in\Omega} \left| \widetilde{R}_{NW}(\nu, x) -R(\nu,x) \right| =o(1), 
    \end{align}
    as $h \to 0$ and $nh / \log n \to \infty$ when $n \to \infty$, where $\widetilde{R}_{NW}(\nu, x)=\mathbb{E}[ K_h(X-x) d^2(\nu,Y) ]/\mathbb{E}[ K_h(X-x)]$.

    Note that (\ref{eq::kernel thm uniform 1}) follows by borrowing similar arguments to those used in the proof of Proposition~\ref{thm::kernel uniform conv}, thus we only need to verify (\ref{eq::kernel thm uniform 2}). By the law of total expectation, one gets
    \[
    \widetilde{R}_{NW}(\nu, x) -R(\nu,x) = \frac{\mathbb{E}\left[ K_h(X-x)\left\{ R(\nu,X)-R(\nu,x) \right\} \right]}{\mathbb{E}[K_h(X-x)]}.
    \]
    This result, together with \ref{condition U2}, gives 
    \begin{align*}
        \left|\widetilde{R}_{NW}(\nu, x) -R(\nu,x)\right| \leq  \sup_{|x_1-x_2|\leq C h} |R(\nu,x_1)-R(\nu,x_2)|
    \end{align*}
    for some positive constant $C$. Thus as $h\to 0$, (\ref{eq::kernel thm uniform 2}) holds. This completes the proof.
    
\end{proof}

\begin{proof}[Proof of Theorem~\ref{thm::kNN uniform conv}]
    Recall that for the $k$-NN Fr\'{e}chet regression in Model~\ref{exmp::kNN F regression}, $\widehat{R}(\nu,x)=k^{-1}\sum_{i=1}^n \sum_{i\in\mathcal{I}_k} d^2(\nu,Y_i)$ where $\mathcal{I}_k$ is the index set of the $k$ nearest neighbors to $x$. By Theorem \ref{thm::uniform convergence 2}, it is sufficient to show that
    \[
    \sup_{x\in\mathcal{X}}\sup_{\nu\in\Omega} \left| \widehat{R}(\nu, x) -R(\nu,x) \right| =o_p(1) .
    \]
    Let $H_k(x)$ denote the distance from $x$ to its $k$-th nearest neighbor in $\{X_1, \dots, X_n\}$, note that $\widehat{R}(\nu, x)$ can be represented as
    \[
    \widehat{R}(\nu, x) = \frac{n}{k}  \frac{1}{n} \sum_{i=1}^n \mathbb{I}(d_\mathcal{X}(X_i, x) \le H_k(x)) d^2(\nu, Y_i).
    \]
    As $\left| \widehat{R}(\nu, x) -R(\nu,x) \right|\leq  \left| \widetilde{R}_{kNN}(\nu,x) -R(\nu,x) \right|+ \left| \widehat{R}(\nu, x) -\widetilde{R}_{kNN}(\nu,x) \right|$ where $\widetilde{R}_{kNN}(\nu,x)=k^{-1} \sum_{i \in \mathcal{I}_k(x)} R(\nu, X_i)$. It is sufficient to verify
    \begin{align}\label{formula::kNN uni conv 1}
         \sup_{x\in\mathcal{X}}\sup_{\nu\in\Omega}\left| \widetilde{R}_{kNN}(\nu,x) -R(\nu,x) \right| = o(1),
    \end{align}
    and 
    \begin{align}\label{formula::kNN uni conv 2}
         \sup_{x\in\mathcal{X}}\sup_{\nu\in\Omega}\left|  \widehat{R}(\nu, x) -\widetilde{R}_{kNN}(\nu,x) \right| = o_p(1).
    \end{align}

    Regarding (\ref{formula::kNN uni conv 1}), observe that $\left| \widetilde{R}_{kNN}(\nu,x) -R(\nu,x) \right|\leq  \sup_{d_\mathcal{X}(u, x) \le H_k(x)} \left| R(\nu, u) - R(\nu, x) \right|$. As $\mathcal{X}$ is compact, under \ref{condition U5}, this leads to
    \begin{align}\label{formula::kNN uni conv 1-1}
        \sup_{x\in\mathcal{X}}\sup_{\nu\in\Omega}\left| \widetilde{R}_{kNN}(\nu,x) -R(\nu,x) \right| \leq C_1 \sup_{x\in\mathcal{X}} H_k(x)
    \end{align}
    for some constant $C_1$. Let $g_{\mathcal{X},n}$ denote the empirical measure of the covariates $X_1,X_2,\ldots,X_n$. By construction, $g_{\mathcal{X},n}(B_\mathcal{X}(x, H_k(x))) = k/n$. By \ref{condition U4}, $\mathcal{C}_\mathcal{X}$ is a Vapnik-Chervonenkis-class, which implies that $\sup_{B \in \mathcal{C}_\mathcal{X}} |g_{\mathcal{X},n}(B) - g_\mathcal{X}(B)| = O_p(n^{-1/2})$. Therefore, 
    \[
    \sup_{x \in \mathcal{X}} g_X(B_\mathcal{X}(x, H_k(x)))= O(k/n)
    \]
    almost surely. By \ref{condition U3}, we then obtain $\phi(\sup_{x \in \mathcal{X}} H_k(x)) = O(k/n)$. Because $\phi$ is strictly increasing with $\phi^{-1}(0)=0$, we deduce that $\sup_{x \in \mathcal{X}} H_k(x) = O\left( \phi^{-1}\left(k/n\right) \right)$. As $k/n\to 0$ as $n\to\infty$, these results, together with (\ref{formula::kNN uni conv 1-1}), imply that (\ref{formula::kNN uni conv 1}) holds.

    Regarding (\ref{formula::kNN uni conv 2}), let $D_i(\nu) = d^2(\nu, Y_i) - R(\nu, X_i)$ and $W_n(x, r, \nu) = n^{-1} \sum_{i=1}^n \mathbb{I}(d_\mathcal{X}(X_i, x) \le r) D_i(\nu)$. Note that $|D_i(\nu)| \leq 2C_2$ for some positive constant $C_2$. Consider the  function class
    \[
    \mathcal{F} = \left\{ (z, y) \mapsto \mathbb{I}\left\{d_\mathcal{X}(z, x) \le r \right\} \left\{d^2(\nu, y) - R(\nu, z)\right\} : x \in \mathcal{X}, \nu \in \Omega, r > 0 \right\}.
    \]
    By \ref{condition U4}, the indicators of balls form a Vapnik-Chervonenkis-class.  Consequently, $\mathcal{F}$ is a Donsker class with a finite uniform entropy integral.
    
    Recall that from the above analysis on (\ref{formula::kNN uni conv 1}), $\sup_x H_k(x) \leq r_{max}$ with probability approaching 1, where $r_{max} = O( \phi^{-1}(k/n))$. For any $f \in \mathcal{F}$ restricted to $r \le r_{max}$, one obtains 
    \begin{align*}
        \sup_{f \in \mathcal{F}, r \le r_{max}}\mathbb{E}[f^2] & = \sup_{f \in \mathcal{F}, r \le r_{max}}\mathbb{E}\left[ \left\{ \mathbb{I}(d_\mathcal{X}(X, x) \le r)(d^2(\nu, Y) - R(\nu, X)) \right\}^2 \right]\\
        &\leq C_3 \sup_{x \in \mathcal{X}} g_\mathcal{X}(B_\mathcal{X}(x, r_{max})) \le C_3k/n
    \end{align*}
    for some positive constant $C_3$. This gives 
    \[
    \sup_{x \in \mathcal{X}, \nu \in \Omega, r \le r_{max}} |W_n(x, r, \nu)|  = O_p\left( \frac{\sqrt{k \log n}}{n} \right).
    \]
    Let $S_n(\nu, x) = n W_n(x, H_k(x), \nu)/k$, we obtain $\sup_{x \in \mathcal{X}, \nu \in \Omega} |S_n(\nu, x)| = O_p\left( \sqrt{\log n/k} \right)$. As $k/\log n \to \infty$, $\sup_{x, \nu} |S_n(\nu, x)| = o_p(1)$. This ensures that (\ref{formula::kNN uni conv 2}) hold and thus completes the proof.
\end{proof}

\section{Proof of results in Section \ref{sect::MoM}}\label{supp sect::proof3}
Recall that the weighted Fr\'echet estimator of the $\ell$-th order ($\ell \in\{1,2\}$) at $x\in\mathcal{X}$ is defined as
\begin{equation}\label{e:mom_supp}
\widehat{r}(x,\pmb{\kappa};\ell) = \argmin_{\nu\in\Omega} \widehat{L}^{(\ell)}(\nu; x,\pmb{\kappa}),~~\widehat{L}^{(\ell)}(\nu; x,\pmb{\kappa}) = \sum_{b=1}^B \kappa_b d^\ell(\nu, \widehat{r}_b(x)),
\end{equation}
where $\{\widehat{r}_b(x)\}_{b=1}^B$ are the block-wise Fr\'echet estimators, as defined in \eqref{e:blockwise} and $\pmb{\kappa}=(\kappa_1,\ldots,\kappa_B)^\top$ is the weight vector in the simplex
\begin{equation*}
\mathcal{K} = \left\{\pmb{\kappa}\in \mathbb{R}^B\Big| \sum_{b=1}^B \kappa_b =1 ~\text{and}~ \kappa_b \ge 0 ~ \text{for} ~ b=1,\ldots,B\right\}.
\end{equation*}  
Note that $\widehat{r}(x,\pmb{\kappa};\ell)$ is essentially a weighted Fr{\'e}chet means as a convex combination of $\{\widehat{r}_b(x)\}_{b=1}^B$. It is natural to define the population quantity of $\widehat{r}(x,\pmb{\kappa};\ell)$ as
\begin{equation*}
r(x,\pmb{\kappa};\ell) = \argmin_{\nu\in\Omega} L^{(\ell)}(\nu; x,\pmb{\kappa}),~~L^{(\ell)}(\nu; x,\pmb{\kappa}) = \sum_{b=1}^B \kappa_b d^\ell(\nu, r_b(x)),
\end{equation*}
where $r_b(x)$ is the conditional Fr\'{e}chet mean in the $b$-th block, i.e., 
$$r_b(x) = \argmin_{\nu\in\Omega} \mathbb{E}\{d^2(\nu,Y)|X = x\}.$$

\begin{proof}[Proof of Theorem \ref{theorem::mom converge}] 
Analogously to the proofs of Proposition~\ref{prop::uniform bound from loss to estimator} and Theorem~\ref{thm::uniform convergence 2}, it is then sufficient to examine the uniform convergence of the loss function, i.e.,
\[\sup_{x\in \mathcal{X}, \pmb{\kappa}\in \mathcal{K}}~\sup_{\nu \in \Omega} ~ \left\vert \widehat{L}^{(\ell)}(\nu; x,\pmb{\kappa}) - L^{(\ell)}(\nu; x, \pmb{\kappa})\right\vert = o_p(1).\]
This can be established by showing the weak convergence of $\widehat{L}^{(\ell)}(\nu; x,\pmb{\kappa})$ in $\ell^{\infty}(\Omega\times\mathcal{X}\times \mathcal{K})$. Then, applying the continuous mapping theorem leads to the desired result. The weak convergence of $\widehat{L}^{(\ell)}(\nu; x,\pmb{\kappa})$ can be further established by verifying (i) its pointwise convergence and (ii) $\widehat{L}^{(\ell)}(\nu; x,\pmb{\kappa})$ is asymptotically equicontinuous. 

The point convergence of $\widehat{L}^{(\ell)}(\nu;x,\pmb{\kappa})$ is immediate, applying Slutsky's theorem and the point convergence of $\widehat{r}_b(x)$. It then remains to verify the asymptotic equicontinuity, that is, for any $\iota>0$, as $\zeta_1,\zeta_2,\zeta_3 \to 0$,
\begin{equation*}
\limsup_{n\to \infty} \mathbb{P}\left(\sup_{d(\nu_1,\nu_2)<\zeta_1; \|x_1 - x_2\|<\zeta_2; \|\pmb{\kappa}_1 - \pmb{\kappa}_2\|<\zeta_3} \left|\widehat{L}^{(\ell)}(\nu_1; x_1,\pmb{\kappa}_1) - \widehat{L}^{(\ell)}(\nu_2; x_2,\pmb{\kappa}_2)\right| > \iota\right) = 0.
\end{equation*}

Denote $\text{diam}(\Omega) = \sup_{\nu_1,\nu_2\in \Omega}d(\nu_1,\nu_2)$. It follows that, for $\ell=1$, 
\begin{align*}
&\left|\widehat{L}^{(1)}(\nu_1; x_1,\pmb{\kappa}_1) - \widehat{L}^{(1)}(\nu_2; x_2,\pmb{\kappa}_2)\right| \\
&= \left| \sum_{b=1}^B \kappa_{1b} d(\nu_1,\widehat{r}_b(x_1)) - \sum_{b=1}^B \kappa_{2b} d(\nu_2,\widehat{r}_b(x_2))\right|\\
&\le \left| \sum_{b=1}^B (\kappa_{1b}- \kappa_{2b}) d(\nu_1,\widehat{r}_b(x_1))\right| + \sum_{b=1}^B \kappa_{2b}\left|d(\nu_1,\widehat{r}_b(x_1))- d(\nu_2,\widehat{r}_b(x_2))\right|\\
&\le B\text{diam}(\Omega)\|\pmb{\kappa}_1 - \pmb{\kappa}_2\|
+ \sum_{b=1}^B \kappa_{2b} \left|d(\nu_1,\widehat{r}_b(x_1)) - d(\nu_2,\widehat{r}_b(x_1))\right| \\
&\quad + \sum_{b=1}^B \kappa_{2b} \left|d(\nu_2,\widehat{r}_b(x_1)) - d(\nu_2,\widehat{r}_b(x_2))\right|,
\end{align*}
where the second inequality can be obtained by the Cauchy-Schwarz inequality and $d(\nu_1, \widehat{r}_b(x_1))\le \text{diam}(\Omega)$, and the remaining inequalities are established using the triangle inequality. 
The first term vanishes since $\|\pmb{\kappa}_1 - \pmb{\kappa}_2\| \to 0$ as $\zeta_3 \to 0$. Applying the triangle inequality to the second term yields 
$$\sum_{b=1}^B \kappa_{2b} \left|d(\nu_1,\widehat{r}_b(x_1)) - d(\nu_2,\widehat{r}_b(x_1))\right|\le \sum_{b=1}^B \kappa_{2b} d(\nu_1, \nu_2) =d(\nu_1, \nu_2)<\zeta_1.$$ Similarly, the last term is given by
\begin{align*}
&\sum_{b=1}^B \kappa_{2b} \left|d(\nu_2,\widehat{r}_b(x_1)) - d(\nu_2,\widehat{r}_b(x_2))\right|\\
&\leq \sum_{b=1}^B \kappa_{2b} d(\widehat{r}_b(x_1), \widehat{r}_b(x_2)) \\
&\leq \sum_{b=1}^B \kappa_{2b} \left\{d(\widehat{r}_b(x_1), r_b(x_1))+d(\widehat{r}_b(x_2), r_b(x_2))+d(r_b(x_1), r_b(x_2))\right\}.
\end{align*}
The block-wise Fr\'{e}chet estimators are consistent uniformly over $x\in \mathcal{X}$, i.e., 
$$\sup_{x\in \mathcal{X}} d(\widehat{r}_b(x), r_b(x))= o_p(1).$$
Moreover, we have $d(r_b(x_1), r_b(x_2))\leq C \|x_1 - x_2\|< C\zeta_2$ for some constant $C>0$ due to the continuity of $r_b(x)$ over $x\in \mathcal{X}$. The continuity of $r_b(x)$ is a direct consequence of the equicontinuity condition of $M(\nu, x) = \mathbb{E}[d^2(Y,\mu)|X=x]$, e.g., Condition~\ref{condition U2}, and the Berge's maximum theorem. Hence, we have
$$\sum_{b=1}^B \kappa_{2b} \left|d(\nu_2,\widehat{r}_b(x_1)) - d(\nu_2,\widehat{r}_b(x_2))\right|\leq C \zeta_2 +o_p(1).$$
Consequently, as $\zeta_1,\zeta_2,\zeta_3\to 0$ and $n\to \infty$,
\begin{equation*}
\limsup_{n\to \infty} \mathbb{P}\left(\sup_{d(\nu_1,\nu_2)<\zeta_1; \|x_1 - x_2\|<\zeta_2; \|\pmb{\kappa}_1 - \pmb{\kappa}_2\|<\zeta_3} \left|\widehat{L}^{(1)}(\nu_1; x_1,\pmb{\kappa}_1) - \widehat{L}^{(1)}(\nu_2; x_2,\pmb{\kappa}_2)\right| > \iota\right) = 0.
\end{equation*}

Similarly, for $\ell=2$, we have
\begin{align*}
&\left|\widehat{L}^{(2)}(\nu_1; x_1,\pmb{\kappa}_1) - \widehat{L}^{(2)}(\nu_2; x_2,\pmb{\kappa}_2)\right|\\
& = \left|\sum_{b=1}^B \kappa_{1b} d^2(\nu_1,\widehat{r}_{b}(x_1)) - \sum_{b=1}^B \kappa_{2b} d^2(\nu_2,\widehat{r}_{b}(x_2))\right|\\
&\le \left|\sum_{b=1}^B (\kappa_{1b} -\kappa_{2b}) d^2(\nu_1,\widehat{r}_{b}(x_1))\right| + \sum_{b=1}^B \kappa_{2b}\left|d^2(\nu_1,\widehat{r}_{b}(x_1))- d^2(\nu_2,\widehat{r}_{b}(x_2))\right|\\
& \le B (\text{diam}(\Omega))^2 \|\pmb{\kappa}_1 - \pmb{\kappa}_2\| + 2 \text{diam}(\Omega) \sum_{b=1}^B \kappa_{2b}\left|d(\nu_1,\widehat{r}_{b}(x_1))- d(\nu_2,\widehat{r}_{b}(x_2))\right|,
\end{align*}
where the first inequality is derived by the triangle inequality and the second inequality is obtained by the Cauchy-Schwarz inequality. The second term has been investigated in the case $\ell=1$, which further implies, for $\ell=2$, 
as $\zeta_1,\zeta_2,\zeta_3\to 0$ and $n\to \infty$,
\[
\left|\widehat{L}^{(2)}(\nu_1; x_1,\pmb{\kappa}_1) - \widehat{L}^{(2)}(\nu_2; x_2,\pmb{\kappa}_2)\right| \to 0.
\]
This desired result is immediate.
\end{proof}

The following lemma is a cornerstone for proving Theorem~\ref{theorem::mom deviation}, which extends Lemma 2.1 in \citet{minsker2015geometric} to metrically convex spaces, e.g., examples in Section 2. Intuitively, this lemma asserts that, if a point $z\in \Omega$ is far away from $\widehat{r}^{(1)}(x,\pmb{\kappa})$, e.g., $d(z,\widehat{r}^{(1)}(x,\pmb{\kappa}))> C_{\alpha}(x)\epsilon$, it also has to be far away from at least an proportion of the block-wise Fr\'{e}chet estimators, e.g., $B_1\subseteq\{1,\ldots, B\}$ such that $\sum_{b\in B_1} \kappa_b B> \alpha(x) B$ and $d(z, \widehat{r}_{b}(x))>\epsilon$. 

\begin{lemma}
\label{l::mom geometric discrepancy}
Suppose $\widehat{r}^{(1)}(x, \pmb{\kappa})$ as defined in \eqref{e:mom_supp} uniquely exists almost surely over $x\in \mathcal{X}$ and $\pmb{\kappa}\in \mathcal{K}$, with $\mathcal{K}$ as per \eqref{e:Kappa} and the corresponding block-wise estimators $\{\widehat{r}_{b}(x)\}_{b=1}^B$ as per \eqref{e:blockwise}. For a fixed $x\in \mathcal{X}$ and $\pmb{\kappa}\in \mathcal{K}$, assume that $z\in \Omega$ is such that $d(z,\widehat{r}^{(1)}(x,\pmb{\kappa}))> C_{\alpha}\epsilon$, where $\epsilon>0$, $\alpha\in(0,1/2)$ and $C_{\alpha} = 2(1-\alpha)/\{1-2\alpha\}$. Then, there exists a subset $B_1\subseteq \{1,\ldots, B\}$ such that $d(z, \widehat{r}_{b}(x))> \epsilon$ for all $b\in B_1$ and $\sum_{b\in B_1} \kappa_b > \alpha$.
\end{lemma}

\begin{proof}[Proof of Lemma~\ref{l::mom geometric discrepancy}]
Assume that the implication is not true. Then, for each $x\in\mathcal{X}$, $\pmb{\kappa}\in \mathbb{{K}}$ and $z\in\Omega$ such that $d(z,\widehat{r}^{(1)}(x,\pmb{\kappa}))> C_{\alpha}\epsilon$, there exists a subset $B_1\subseteq\{1,\ldots, B\}$ such that 
$$\sum_{b\in B_1} \kappa_b \leq \alpha \quad \text{and} \quad d(\widehat{r}_{b}(x), z)> \epsilon \quad \text{for all } b\in B_1.$$
That is, if the point $z\in \Omega$ is far away from the weight Fr\'{e}chet estimator, i.e.,$d(z,\widehat{r}^{(1)}(x,\pmb{\kappa}))> C_{\alpha}\epsilon$, the proportion of the block-wise estimator $\widehat{r}_{b}(x)$ such that $d(\widehat{r}_{b}(x), z)> \epsilon$ is no larger than $\alpha$.

Moreover, for each $x\in\mathcal{X}$, $\pmb{\kappa}\in \mathbb{{K}}$ and $z\in\Omega$ such that $d(z,\widehat{r}^{(1)}(x,\pmb{\kappa}))> C_{\alpha}\epsilon$, the block-wise estimators in the complement subset $B_2 = \{1,\ldots, B\} \setminus B_1$ satisfy
$$\sum_{b\in B_2} \kappa_b \ge 1- \alpha \quad \text{and} \quad d(\widehat{r}_{b}(x), z)\leq \epsilon \quad \text{for all } b\in B_2.$$
Without loss of generality, we consider $\sum_{b\in B_1} \kappa_b = \alpha$ and $\sum_{b\in B_2} \kappa_b = 1-\alpha$.
Since $\widehat{r}^{(1)}(x,\pmb{\kappa})$ uniquely minimizes the loss function $\widehat{L}^{(1)}(\nu;x,\pmb{\kappa}) = \sum_{b=1}^B \kappa_b d(\nu, \widehat{r}_{b}(x))$, it follows that, for any $z\in \Omega$ such that $d(z,\widehat{r}^{(1)}(x,\pmb{\kappa}))> C_{\alpha}(x)\epsilon$,
\begin{equation}\label{e:mom loss discrepancy}
\frac{\widehat{L}^{(1)}(z;x,\pmb{\kappa}) - \widehat{L}^{(1)}(\widehat{r}^{(1)}(x,\pmb{\kappa});x,\pmb{\kappa})}{d(z,\widehat{r}^{(1)}(x,\pmb{\kappa}))} = \frac{\sum_{b=1}^B \kappa_b \{d(\widehat{r}_{b}(x), z) - d(\widehat{r}_{b}(x), \widehat{r}^{(1)}(x,\pmb{\kappa}))\}}{d(z,\widehat{r}^{(1)}(x,\pmb{\kappa}))}  > 0.
\end{equation}
Write $A_b(x) = d(\widehat{r}_{b}(x), z) - d(\widehat{r}_{b}(x), \widehat{r}^{(1)}(x,\pmb{\kappa}))$. For $b\in B_1$, the triangle inequality implies that 
$$A_b(x)\leq d(z,\widehat{r}^{(1)}(x,\pmb{\kappa})).$$ 
For $b\in B_2$, applying the triangle inequality again yields
$$d(\widehat{r}_{b}(x), \widehat{r}^{(1)}(x,\pmb{\kappa})) \geq d(z,\widehat{r}^{(1)}(x,\pmb{\kappa})) - d(\widehat{r}_{b}(x),z)> C_{\alpha}\epsilon - \epsilon,$$
which further implies
$$A_b(x) / d(z, \widehat{r}^{(1)}(x,\pmb{\kappa}))< 
\frac{\epsilon - (C_{\alpha}\epsilon - \epsilon)}{C_{\alpha}\epsilon} = \frac{2}{C_{\alpha}}-1.$$
Consequently, \eqref{e:mom loss discrepancy} simplifies to
\begin{align*}
\frac{\widehat{L}^{(1)}(z;x,\pmb{\kappa}) - \widehat{L}^{(1)}(\widehat{r}^{(1)}(x,\pmb{\kappa});x,\pmb{\kappa})}{d(z,\widehat{r}^{(1)}(x,\pmb{\kappa}))} 
&= \frac{\sum_{b\in B_1} \kappa_b A_b(x)}{d(z,\widehat{r}^{(1)}(x,\pmb{\kappa}))} + \frac{\sum_{b\in B_2} \kappa_b A_b(x)}{d(z,\widehat{r}^{(1)}(x,\pmb{\kappa}))}\\
& < \alpha + (1-\alpha)\left(\frac{2}{C_{\alpha}}-1\right) \le 0,
\end{align*}
where the last inequality holds whenever $C_{\alpha}\ge 2(1-\alpha)/(1-2\alpha)$ and $\alpha\in(0,1/2)$, which leads to a contradiction to \eqref{e:mom loss discrepancy}. 
\end{proof}

Note that $C_{\alpha}$ depends only on the proportion $\alpha\in(0,1/2)$, and coincides with the constant obtained for general Banach spaces \citep[see Lemma 2.1(b) of][]{minsker2015geometric}. The following theorem establishes the deviation bound of $\widehat{r}^{(1)}(x, \pmb{\kappa})$.

\begin{proof}[Proof of Theorem \ref{theorem::mom deviation}]
It follows from Lemma~\ref{l::mom geometric discrepancy} that, for every $x\in \mathcal{X}$ and $\pmb{\kappa}\in \mathcal{K}$, if the event $\{d(\widehat{r}^{(1)}(x,\pmb{\kappa}), r(x))>C_{\alpha}\epsilon_x\}$ occurs for some constant $\epsilon_{x}>0$, then $r(x)$ has to be also far away from at least an proportion of the block-wise estimators, i.e.,
\[\mathbb{P}\left\{d(\widehat{r}^{(1)}(x,\pmb{\kappa}), r(x))>C_{\alpha}\epsilon_x \right\}\le \mathbb{P}\left\{\sum_{b=1}^B\kappa_b \mathbb{I}\left\{d(\widehat{r}_b(x), r(x))>\epsilon_x\right\} > \alpha\right\},\]
which can further induce that
\begin{align*}
&\sup_{x\in \mathcal{X}} \mathbb{P}\left\{d(\widehat{r}^{(1)}(x,\pmb{\kappa}), r(x))>C_{\alpha}\epsilon_x \right\} \\
&\leq  \sup_{x\in \mathcal{X}} \mathbb{P}\left\{\sum_{b=1}^B\kappa_b \mathbb{I}\left\{d(\widehat{r}_b(x), r(x))>\epsilon_x\right\} > \alpha\right\}\\
&\leq \mathbb{P} \left\{\sum_{b=1}^B \kappa_b \sup_{x\in \mathcal{X}} \mathbb{I}\left\{d(\widehat{r}_b(x), r(x))>\epsilon\right\} > \alpha\right\},
\end{align*}
where $\epsilon:= \inf_{x\in \mathcal{X}} \epsilon_x$.
Recall that $B_2\subset \{1,\ldots,B\}$ with $\sum_{b\in B_2}\kappa_b = \tau$, is the set of blocks that are arbitrarily corrupted and $B_1 = \{1,\ldots,B\} \setminus B_2$. Therefore, we have
\begin{align*}
    &\sum_{b=1}^B \kappa_b \sup_{x\in \mathcal{X}} \mathbb{I}\left\{d(\widehat{r}_b(x), r(x))>\epsilon\right\} \\
    &= \sum_{b\in B_1} \kappa_b \sup_{x\in \mathcal{X}} \mathbb{I}\left\{d(\widehat{r}_b(x), r(x))>\epsilon\right\} + \sum_{b\in B_2} \kappa_b \sup_{x\in \mathcal{X}} \mathbb{I}\left\{d(\widehat{r}_b(x), r(x))>\epsilon\right\}\\
    &\leq \sum_{b\in B_1} \kappa_b \sup_{x\in \mathcal{X}} \mathbb{I}\left\{d(\widehat{r}_b(x), r(x))>\epsilon\right\} + \tau.
\end{align*}
which implies
\begin{equation*}
    \left\{\sum_{b=1}^B \kappa_b \sup_{x\in \mathcal{X}} \mathbb{I}\left\{d(\widehat{r}_b(x), r(x))>\epsilon\right\} > \alpha\right\} \subseteq \left\{\sum_{b\in B_1} \kappa_b \sup_{x\in \mathcal{X}} \mathbb{I}\left\{d(\widehat{r}_b(x), r(x))>\epsilon\right\} > \alpha - \tau\right\}.
\end{equation*}
Consequently, we obtain
\begin{align*}
\sup_{x\in \mathcal{X}}  \mathbb{P}\left\{d(\widehat{r}^{(1)}(x,\pmb{\kappa}), r(x))>C_{\alpha}\epsilon \right\} 
&\leq \mathbb{P}\left\{\sum_{b\in B_1} \kappa_b \sup_{x\in \mathcal{X}} \mathbb{I}\left\{d(\widehat{r}_b(x), r(x))>\epsilon\right\} > \alpha - \tau\right\}\\
& \leq \mathbb{P}\left\{\sum_{b\in B_1}\kappa_b W_b > \alpha - \tau\right\},
\end{align*}
where $W_b$ for $b\in B_1$ denote independent Bernoulli random variables with $\mathbb{P}(W_b=1) = p$ and $\mathbb{P}(W_b=0) = 1-p$. The desire result then follows from the Chernoff bound that
$$\sup_{x\in \mathcal{X}}  \mathbb{P}\left\{d(\widehat{r}^{(1)}(x,\pmb{\kappa}), r(x))>C_{\alpha}\epsilon \right\}\leq 
\exp\left\{-\psi(\alpha-\tau, p, \pmb{\kappa})\right\},$$
where
$$\psi(\alpha-\tau, p, \pmb{\kappa}) = \sup_{\lambda>0} \left\{\lambda (\alpha-\tau) - \sum_{b\in B_1} \log (1-p + p \exp(\lambda \kappa_b))\right\}.$$
\end{proof}

\begin{proof}[Proof of Corollary~\ref{corollary::mom}]
Note that the median-of-means estimator is a special case of $\widehat{r}(x,\pmb{\kappa};\ell)$ by specifying $\ell=1$, $\pmb{\kappa}=(1,\ldots, 1)^\top/B$ with independent and homogeneous blocks of balanced sample size, i.e., $\mathcal{P}_b=\mathcal{P}$ for all $b\in B_1$ and $n_1 = \ldots = n_B = \lfloor N/B\rfloor$, where $N=\sum_{b=1}^B n_b$. Without loss of generality, we assume $n:=N/B$ is an integer. 
The deviation bound in Theorem~\ref{theorem::mom deviation} then simplifies to
\[
\sup_{x\in\mathcal{X}} \mathbb{P} \left\{d(\widehat{r}^{(1)}(x,\pmb{\kappa}), r(x))>C_{\alpha}\epsilon\right\} \leq \exp\{-\psi(\alpha-\tau, p,\pmb{\kappa})\},
\]
where
\begin{align*}
\psi(\alpha-\tau, p,\pmb{\kappa}) &= B(1-\tau)\left\{\frac{1-\alpha}{1-\tau}\log\left(\frac{1-\alpha}{(1-\tau)(1-p)}\right)
+\frac{\alpha-\tau}{1-\tau}\log\left(\frac{\alpha - \tau}{(1-\tau)p}\right)\right\}.    
\end{align*}
We then rewrite this deviation bound with the following choices of hyperparameters:
\begin{align*}
p(\tau):= \left(1/2 -\tau\right)^2/2,~
\alpha(\tau):= 2p(\tau) + \tau = \tau^2 + 1/4,~
\epsilon(\tau):= D n^{-\upsilon_j}(p(\tau))^{-1},    
\end{align*}
where $D<\infty$ is a positive constant and $\upsilon_j$, $j=1,2$, corresponds the uniform convergence rate for the block-wise Fr\'{e}chet estimators, i.e., $\upsilon_1 = 1/(2(\alpha' - 1))$ for the global Fr\'{e}chet regression as in Theorem \ref{prop::global uniform conv} and $\upsilon_2 = 1/(\alpha+2\beta-3+\eta)$ for the local Fr\'{e}chet regression as in Theorem \ref{prop::local uniform conv}. Hence, we denote $\psi(\alpha-\tau, p,\pmb{\kappa})=B(1-\tau)\iota(\tau)$, where
\begin{equation*}
    \iota(\tau) = \frac{1-\alpha(\tau)}{1-\tau}\log\left(\frac{1-\alpha(\tau)}{(1-\tau)(1-p(\tau))}\right) + \frac{\alpha(\tau) - \tau}{1-\tau}\log\left(\frac{\alpha(\tau) - \tau}{(1-\tau)p(\tau)}\right),
\end{equation*}
and $\iota(\tau) = O((1/2-\tau)^2)$ as $\tau\to 1/2$.
The number of blocks is then chosen as
\begin{align*}
B(\tau,\delta) &:= \left\lfloor \frac{\log(1/\delta)}{(1-\tau)\iota(\tau)}\right\rfloor + 1,
\end{align*}

The exponent on the right-hand side of the deviation bound is given by
\begin{align*}
\psi\left(\alpha-\tau,p,\pmb{\kappa}\right) 
&=  \left( \frac{\log(1/\delta)}{(1-\tau) \iota(\tau)}-c + 1\right)(1-\tau)\iota(\tau)\\
& = \log(1/\delta) + (1-c)(1-\tau)\iota(\tau)\\
& \ge \log(1/\delta),
\end{align*}
for some constant $c\in [0,1)$ and $\iota(\tau)\ge 0$. 
Hence, the deviation bound simplifies to
\[
\sup_{x\in \mathcal{X}}\mathbb{P}\left\{d(\widehat{r}^{(1)}(x,\pmb{\kappa}), r(x))>C_{\alpha}\epsilon \right\} \le \delta.
\]
Next, we proceed to $C_{\alpha}\epsilon$, which is given by
\begin{align*}
C_{\alpha}\epsilon 
& = \frac{2(1-\alpha(\tau))}{1-2\alpha(\tau)} \frac{D n^{-\upsilon_j}}{p(\tau)}\\
& = \frac{3 - 4\tau^2}{1-4\tau^2} \frac{2D}{(1/2 -\tau)^2} \left(\frac{N}{B(\tau,\delta)}\right)^{-\upsilon_j}\\
& =\frac{3 - 4\tau^2}{1-4\tau^2} \frac{2D}{(1/2 -\tau)^2} \{N(1-\tau) \iota(\tau)\}^{-\upsilon_j} \{B(\tau,\delta) (1-\tau) \iota(\tau)\}^{\upsilon_j}\\
& := c(\tau)N^{-\upsilon_j} \{B(\tau,\delta) (1-\tau) \iota(\tau)\}^{\upsilon_j}
\end{align*}
where
\begin{equation*}
    c(\tau) = \frac{3 - 4\tau^2}{1-4\tau^2} \frac{2D}{(1/2 -\tau)^2} \{(1-\tau) \iota(\tau)\}^{-\upsilon_j},
\end{equation*}
and $c(\tau)=O((1/2-\tau)^{-2\upsilon_j+3})$ as $\tau\to 1/2$. To further simplify $C_{\alpha}\epsilon$, we now focus on the term $\{B(\tau,\delta) (1-\tau) \iota(\tau)\}^{\upsilon_j}$, where
\begin{align*}
B(\tau,\delta) (1-\tau) \iota(\tau) &= \left(\left\lfloor \frac{\log(1/\delta)}{(1-\tau)\iota(\tau)}\right\rfloor + 1\right)(1-\tau) \iota(\tau)\\
&\le \left( \frac{\log(1/\delta)}{(1-\tau)\iota(\tau)}+ 1\right)(1-\tau) \iota(\tau)\\
& = \log(1/\delta) + (1-\tau)\iota(\tau).
\end{align*}
Note that the first term of $\iota(\tau)$ is negative, which implies that
\begin{align*}
B(\tau,\delta) (1-\tau) \iota(\tau) &\leq \log(1/\delta) + (1-\tau) \frac{\alpha(\tau) - \tau}{1-\tau}\log\left(\frac{\alpha(\tau) - \tau}{(1-\tau)p(\tau)}\right)\\
& = \log(1/\delta) + (2p(\tau))\log\left(\frac{2}{1-\tau}\right)\\
& \leq \log(1/\delta) + \log\left(\frac{2}{1-\tau}\right),
\end{align*}
where the last inequality is obtained due to $p\leq 1/2$. Consequently, we have
\begin{align*}
    C_{\alpha}\epsilon \leq c(\tau)N^{-\upsilon_j} \{\log(1/\delta) + \log(2/(1-\tau))\}^{\upsilon_j}.
\end{align*}

It then remains to exmaine the choices of the hyperparameters are valid. First, it is straightforward to see that $\alpha(\tau) = \tau^2 + 1/4<1/2$ and $p(\tau)<\alpha(\tau)$ for all $\tau\in[0,1/2)$. Moreover, we have
\begin{equation*}
\frac{\alpha(\tau) - p(\tau)}{1 - p(\tau)} = \frac{2p(\tau) + \tau - p(\tau)}{1 - p(\tau)} = \frac{p(\tau)}{1 - p(\tau)} + \frac{\tau}{1 - p(\tau)} > \tau.
\end{equation*}
Last, it follows from the Markov inequality that
\begin{align*}
\mathbb{P}\left\{\sup_{x\in\mathcal{X}} d(\widehat{r}_b(x),r(x))>\epsilon\right\}
\le \frac{\mathbb{E}[\sup_{x\in\mathcal{X}}d(\widehat{r}_b(x), r(x))]}{\epsilon},
\end{align*}
where $\mathbb{E}[\sup_{x\in\mathcal{X}}d(\widehat{r}_b(x), r(x))] = O(a_n)$ with $a_n = n^{-1/(2(\alpha' - 1))}$ for the global Fr{\'e}chet regression as in Theorem 2 of \cite{petersen2019frechet} and $a_n = n^{-1/(\alpha+2\beta-3+\eta)}$ for the local Fr{\'e}chet regression as in Theorem 2 of \cite{chen2022uniform}. Consequently, recalling that $\epsilon(\tau):= D n^{-\upsilon_j}(p(\tau))^{-1}$, the right-hand side of the above inequality simplifies for a sufficiently large $n$, where everything but $p$ cancels, i.e.,
\begin{equation*}
\mathbb{P}\left\{\sup_{x\in\mathcal{X}} d(\widehat{r}_b(x),r(x))>\epsilon\right\} \le p.    
\end{equation*}
\end{proof}

\section{Proof of results in Section \ref{sect::level set}}\label{supp sect::proof4}

\begin{proof}[Proof of Proposition \ref{prop::level set est}]
    Observe that under \ref{condition L1}, one has
    \[
    \sup_{x\in\mathcal{X}}\left| \Lambda[\{\widehat{m}_{l}(x)\}_{l=1}^L]- \Lambda[\{m_l(x)\}_{l=1}^L] \right| \leq D_2\sum_{l=1}^L\sup_{x\in\mathcal{X}}d^{\lambda}(\widehat{m}_l(x),m_l(x)) = o_p(1),
    \]
    as $\lambda>1$ and $\sup_{x\in\mathcal{X}}d(\widehat{m}_{l}(x),m_{l}(x)) = o_p(1)$ for $l=1,\ldots,L$. Let $\widehat{\Delta}(x)=\Lambda\{\widehat{m}(x)\}$, the above result shows that $\sup_{x\in\mathcal{X}}\left| \widehat{\Delta}(x) - \Delta(x)\right|=o_p(1)$. Therefore, one can construct an event 
    \[
    \mathcal{A}_n=\{ \sup_{x\in\mathcal{X}}\left| \widehat{\Delta}(x) - \Delta(x)\right| < \epsilon\}
    \]
    with $\mathbb{P}(\mathcal{A}_n)\to 1$ as $n\to\infty$ for any $\epsilon>0$. When $\mathcal{A}_n$ is true, we get
    \[
    \{ x\in\mathcal{X}:\Delta(x)>\rho+\epsilon \} \subset \widehat{L}(\rho) \subset \{ x\in\mathcal{X}:\Delta(x)>\rho-\epsilon \}.
    \]
    This ensures that
    \[
    \mathbb{P}\left\{\liminf_{n\to\infty}\widehat{L}(\rho)  \supset {\rm int}\{L(\rho)\}\right\}\to 1
    \]
    and
    \[
    \mathbb{P}\left\{\limsup_{n\to\infty}\widehat{L}(\rho) \subset \bar{L}(\rho)\right\}\to 1,
    \]
    which completes the proof.
\end{proof}

\begin{proof}[Proof of Proposition \ref{prop::exceedance}]
    Let $\epsilon_n=\sup_{x\in\mathcal{X}}|\widehat{\Delta}(x)-\Delta(x)|$ and thus $\epsilon_n=o_p(1)$. Following the arguments used in the proof of Proposition \ref{prop::level set est}, one gets
    \begin{equation*}
        \{ x\in\mathcal{X}:\Delta(x)>\rho+\epsilon_n \} \subset \{ x\in\mathcal{X}:\widehat{\Delta}(x)>\rho\} \subset \{ x\in\mathcal{X}:\Delta(x)>\rho-\epsilon_n \},
    \end{equation*}
    and 
    \begin{equation*}
        \{ x\in\mathcal{X}:\widehat{\Delta}(x)>\rho+\epsilon_n \} \subset \{ x\in\mathcal{X}:\Delta(x)>\rho\} \subset \{ x\in\mathcal{X}:\widehat{\Delta}(x)>\rho-\epsilon_n \}.
    \end{equation*}
    This implies that
    \begin{align*}
        \lambda(\Delta(x)>\rho+\epsilon_n) \leq \lambda(\widehat{\Delta}(x)>\rho) \leq   \lambda(\Delta(x)>\rho-\epsilon_n),
    \end{align*}
and
    \begin{align*}
        \lambda(\Delta(x)>\rho+2\epsilon_n) &\leq  \lambda(\widehat{\Delta}(x)>\rho+\epsilon_n) \leq  \lambda(\Delta(x)>\rho) \\
        &\leq  \lambda(\widehat{\Delta}(x)>\rho-\epsilon_n) \leq \lambda(\Delta(x)>\rho-2\epsilon_n).
    \end{align*}
Then, for any $\rho$, we obtain
    \begin{align*}
        |\lambda(\widehat{\Delta}(x)>\rho)-\lambda(\Delta(x)>\rho)|\leq & \max \Big\{  \lambda(\widehat{\Delta}(x)>\rho-\epsilon_n)- \lambda(\widehat{\Delta}(x)>\rho+\epsilon_n) ,\\
        &\lambda(\Delta(x)>\rho-\epsilon_n)-\lambda(\Delta(x)>\rho+\epsilon_n)\Big\}\\
        \leq & \lambda(\Delta(x)>\rho-2\epsilon_n)-\lambda(\Delta(x)>\rho+2\epsilon_n).
    \end{align*}
    Note that $\lambda(\Delta(x)>\rho)$ is non-increasing with respect to $\rho$ and further one can see that $\lambda(\Delta(x)>\rho-2\epsilon_n)-\lambda(\Delta(x)>\rho+2\epsilon_n) \leq \lambda(\{\varsigma:\rho-2\epsilon_n<\Delta(x)<\rho+2\epsilon_n\})$. This gives 
    \begin{align*}
        \sup_{\rho\in \bar{\mathbb{R}}}\left|\widehat{\ell}(\rho)-\ell(\rho)\right|\leq \sup_{\rho\in \bar{\mathbb{R}}}\frac{\lambda(\{\varsigma:\rho-2\epsilon_n<\Delta(x)<\rho+2\epsilon_n\})}{\lambda(\mathcal{X})}.
    \end{align*}
    Now let $K(\varepsilon) = \sup_{\rho \in \bar{\mathbb{R}}} \lambda(\{x \in \mathcal{X} : |\Delta(x) - \rho| \leq \varepsilon\})$ which goes to $0$ as $\varepsilon \to 0$. As $\epsilon_n=o_p(1)$, we have $K(\epsilon_n)=o_p(1)$, and thus $\sup_{\rho\in \bar{\mathbb{R}}}\left|\widehat{\ell}(\rho)-\ell(\rho)\right|=o_p(1)$. This completes the proof.

\end{proof}

\section{Data generation for simulation studies}\label{supp sect::data generation}

\subsection{Weighted Fr\'{e}chet aggregation}\label{supp subsect::sim MMoM}

Here we simulate a 10-node time-varying network and take the corresponding graph Laplacian matrix as the response $Y_i$. Under the uncontaminated setting, for a given predictor $X_i$ drawn independent and identically distributed from the uniform distribution $U(0,1)$, the strictly lower-triangular entries of the Laplacian are generated as independent and identically distributed random variables following a beta distribution, $y_{i,j} \sim {\rm Beta}(X_i, 1-X_i)$ for $j = 1, \dots, 45$. The response $Y_i$ is then explicitly constructed via ${\rm vech}(Y_i) = (-y_{i,1}, -y_{i,2}, \dots, -y_{i,45})^\top$, where ${\rm vech}(\cdot)$ denotes the half-vectorization operator that extracts the strictly lower-triangular elements of a symmetric matrix, and the diagonal elements of $Y_i$ are calculated such that the row sums of $Y_i$ equal zero. Because the expectation of a ${\rm Beta}(x, 1-x)$ distribution is $x$, this generation scheme ensures that under the Frobenius metric $d_{\rm F}$, the true conditional Fr\'{e}chet mean satisfies $r(x) = \widetilde{r}(x) = \mathrm{vech}^{-1}(-x, \dots, -x)$, where $\mathrm{vech}^{-1}(\cdot)$ is the inverse mapping back to the graph Laplacian space.

For the outlier generation, the strictly lower-triangular entries of the outlying graph Laplacians are generated as independent and identically distributed random variables drawn from a Gamma distribution ${\rm Gamma}(5, 0.5)$. Unlike the uncontaminated responses whose edge weights are bounded within $[0,1]$ via the Beta distribution, these outlying networks possess an expected edge weight of $10$. After generating the lower-triangular entries, we then follow the similar steps to those used in generating the uncontaminated responses. That is, the outlying graph Laplacian is then generated by applying the standard symmetry constraint and defining the diagonal elements to strictly enforce zero row and column sums, ensuring they reside in the same metric space.

\subsection{Exceedance set estimation}\label{supp subsect::exceedance set simulation}

We first generate two distinct classes of network Laplacian responses $\{Y_i^{(1)}\}_{i=1}^n$ and $\{Y_i^{(2)}\}_{i=1}^n$ with the same support for the predictor. For a given predictor $X_i$ drawn independent and identically distributed from the uniform distribution $U(0,1)$, the strictly lower-triangular entries of $Y_i^{(1)}$ are generated independently as $y_{i,j}^{(1)} \sim {\rm Beta}\big(\sin(\pi X_i), 1 - \sin(\pi X_i)\big)$ for $j = 1, \dots, 45$, while the strictly lower-triangular entries of $Y_i^{(2)}$ are generated independently as $y_{i,j}^{(2)} \sim {\rm Beta}\big(\varpi(X_i), 1 - \varpi(X_i)\big)$ for $j = 1, \dots, 45$ where $\varpi(x) = \sin^2(2\pi x)$. The diagonal elements of either $Y_i^{(1)}$ or $Y_i^{(2)}$ are then calculated such that the row sums of the network Laplacian equal zero. This generation setting ensures that the underlying conditional Fr\'{e}chet means for these two classes are $r_{(1)}(x)={\rm vech}^{-1}\big(-\sin(\pi x), \dots, -\sin(\pi x)\big)$ and $r_{(2)}(x)={\rm vech}^{-1}\big(-\varpi(x), \dots, -\varpi(x)\big)$.

\end{document}